\documentclass{article} 
\usepackage{amsthm}
\usepackage{amsfonts}
\usepackage{amsmath}
\usepackage{amssymb}
\usepackage{fullpage}
\usepackage[usenames]{color}
\definecolor{AliceBlue}{rgb}{0.94,0.97,1.00}
\definecolor{AntiqueWhite1}{rgb}{1.00,0.94,0.86}
\definecolor{AntiqueWhite2}{rgb}{0.93,0.87,0.80}
\definecolor{AntiqueWhite3}{rgb}{0.80,0.75,0.69}
\definecolor{AntiqueWhite4}{rgb}{0.55,0.51,0.47}
\definecolor{AntiqueWhite}{rgb}{0.98,0.92,0.84}
\definecolor{BlanchedAlmond}{rgb}{1.00,0.92,0.80}
\definecolor{BlueViolet}{rgb}{0.54,0.17,0.89}
\definecolor{CadetBlue1}{rgb}{0.60,0.96,1.00}
\definecolor{CadetBlue2}{rgb}{0.56,0.90,0.93}
\definecolor{CadetBlue3}{rgb}{0.48,0.77,0.80}
\definecolor{CadetBlue4}{rgb}{0.33,0.53,0.55}
\definecolor{CadetBlue}{rgb}{0.37,0.62,0.63}
\definecolor{CornflowerBlue}{rgb}{0.39,0.58,0.93}
\definecolor{DarkBlue}{rgb}{0.00,0.00,0.55}
\definecolor{DarkCyan}{rgb}{0.00,0.55,0.55}
\definecolor{DarkGoldenrod1}{rgb}{1.00,0.73,0.06}
\definecolor{DarkGoldenrod2}{rgb}{0.93,0.68,0.05}
\definecolor{DarkGoldenrod3}{rgb}{0.80,0.58,0.05}
\definecolor{DarkGoldenrod4}{rgb}{0.55,0.40,0.03}
\definecolor{DarkGoldenrod}{rgb}{0.72,0.53,0.04}
\definecolor{DarkGray}{rgb}{0.66,0.66,0.66}
\definecolor{DarkGreen}{rgb}{0.00,0.39,0.00}
\definecolor{DarkGrey}{rgb}{0.66,0.66,0.66}
\definecolor{DarkKhaki}{rgb}{0.74,0.72,0.42}
\definecolor{DarkMagenta}{rgb}{0.55,0.00,0.55}
\definecolor{DarkOliveGreen1}{rgb}{0.79,1.00,0.44}
\definecolor{DarkOliveGreen2}{rgb}{0.74,0.93,0.41}
\definecolor{DarkOliveGreen3}{rgb}{0.64,0.80,0.35}
\definecolor{DarkOliveGreen4}{rgb}{0.43,0.55,0.24}
\definecolor{DarkOliveGreen}{rgb}{0.33,0.42,0.18}
\definecolor{DarkOrange1}{rgb}{1.00,0.50,0.00}
\definecolor{DarkOrange2}{rgb}{0.93,0.46,0.00}
\definecolor{DarkOrange3}{rgb}{0.80,0.40,0.00}
\definecolor{DarkOrange4}{rgb}{0.55,0.27,0.00}
\definecolor{DarkOrange}{rgb}{1.00,0.55,0.00}
\definecolor{DarkOrchid1}{rgb}{0.75,0.24,1.00}
\definecolor{DarkOrchid2}{rgb}{0.70,0.23,0.93}
\definecolor{DarkOrchid3}{rgb}{0.60,0.20,0.80}
\definecolor{DarkOrchid4}{rgb}{0.41,0.13,0.55}
\definecolor{DarkOrchid}{rgb}{0.60,0.20,0.80}
\definecolor{DarkRed}{rgb}{0.55,0.00,0.00}
\definecolor{DarkSalmon}{rgb}{0.91,0.59,0.48}
\definecolor{DarkSeaGreen1}{rgb}{0.76,1.00,0.76}
\definecolor{DarkSeaGreen2}{rgb}{0.71,0.93,0.71}
\definecolor{DarkSeaGreen3}{rgb}{0.61,0.80,0.61}
\definecolor{DarkSeaGreen4}{rgb}{0.41,0.55,0.41}
\definecolor{DarkSeaGreen}{rgb}{0.56,0.74,0.56}
\definecolor{DarkSlateBlue}{rgb}{0.28,0.24,0.55}
\definecolor{DarkSlateGray1}{rgb}{0.59,1.00,1.00}
\definecolor{DarkSlateGray2}{rgb}{0.55,0.93,0.93}
\definecolor{DarkSlateGray3}{rgb}{0.47,0.80,0.80}
\definecolor{DarkSlateGray4}{rgb}{0.32,0.55,0.55}
\definecolor{DarkSlateGray}{rgb}{0.18,0.31,0.31}
\definecolor{DarkSlateGrey}{rgb}{0.18,0.31,0.31}
\definecolor{DarkTurquoise}{rgb}{0.00,0.81,0.82}
\definecolor{DarkViolet}{rgb}{0.58,0.00,0.83}
\definecolor{DeepPink1}{rgb}{1.00,0.08,0.58}
\definecolor{DeepPink2}{rgb}{0.93,0.07,0.54}
\definecolor{DeepPink3}{rgb}{0.80,0.06,0.46}
\definecolor{DeepPink4}{rgb}{0.55,0.04,0.31}
\definecolor{DeepPink}{rgb}{1.00,0.08,0.58}
\definecolor{DeepSkyBlue1}{rgb}{0.00,0.75,1.00}
\definecolor{DeepSkyBlue2}{rgb}{0.00,0.70,0.93}
\definecolor{DeepSkyBlue3}{rgb}{0.00,0.60,0.80}
\definecolor{DeepSkyBlue4}{rgb}{0.00,0.41,0.55}
\definecolor{DeepSkyBlue}{rgb}{0.00,0.75,1.00}
\definecolor{DimGray}{rgb}{0.41,0.41,0.41}
\definecolor{DimGrey}{rgb}{0.41,0.41,0.41}
\definecolor{DodgerBlue1}{rgb}{0.12,0.56,1.00}
\definecolor{DodgerBlue2}{rgb}{0.11,0.53,0.93}
\definecolor{DodgerBlue3}{rgb}{0.09,0.45,0.80}
\definecolor{DodgerBlue4}{rgb}{0.06,0.31,0.55}
\definecolor{DodgerBlue}{rgb}{0.12,0.56,1.00}
\definecolor{FloralWhite}{rgb}{1.00,0.98,0.94}
\definecolor{ForestGreen}{rgb}{0.13,0.55,0.13}
\definecolor{GhostWhite}{rgb}{0.97,0.97,1.00}
\definecolor{GreenYellow}{rgb}{0.68,1.00,0.18}
\definecolor{HotPink1}{rgb}{1.00,0.43,0.71}
\definecolor{HotPink2}{rgb}{0.93,0.42,0.65}
\definecolor{HotPink3}{rgb}{0.80,0.38,0.56}
\definecolor{HotPink4}{rgb}{0.55,0.23,0.38}
\definecolor{HotPink}{rgb}{1.00,0.41,0.71}
\definecolor{IndianRed1}{rgb}{1.00,0.42,0.42}
\definecolor{IndianRed2}{rgb}{0.93,0.39,0.39}
\definecolor{IndianRed3}{rgb}{0.80,0.33,0.33}
\definecolor{IndianRed4}{rgb}{0.55,0.23,0.23}
\definecolor{IndianRed}{rgb}{0.80,0.36,0.36}
\definecolor{LavenderBlush1}{rgb}{1.00,0.94,0.96}
\definecolor{LavenderBlush2}{rgb}{0.93,0.88,0.90}
\definecolor{LavenderBlush3}{rgb}{0.80,0.76,0.77}
\definecolor{LavenderBlush4}{rgb}{0.55,0.51,0.53}
\definecolor{LavenderBlush}{rgb}{1.00,0.94,0.96}
\definecolor{LawnGreen}{rgb}{0.49,0.99,0.00}
\definecolor{LemonChiffon1}{rgb}{1.00,0.98,0.80}
\definecolor{LemonChiffon2}{rgb}{0.93,0.91,0.75}
\definecolor{LemonChiffon3}{rgb}{0.80,0.79,0.65}
\definecolor{LemonChiffon4}{rgb}{0.55,0.54,0.44}
\definecolor{LemonChiffon}{rgb}{1.00,0.98,0.80}
\definecolor{LightBlue1}{rgb}{0.75,0.94,1.00}
\definecolor{LightBlue2}{rgb}{0.70,0.87,0.93}
\definecolor{LightBlue3}{rgb}{0.60,0.75,0.80}
\definecolor{LightBlue4}{rgb}{0.41,0.51,0.55}
\definecolor{LightBlue}{rgb}{0.68,0.85,0.90}
\definecolor{LightCoral}{rgb}{0.94,0.50,0.50}
\definecolor{LightCyan1}{rgb}{0.88,1.00,1.00}
\definecolor{LightCyan2}{rgb}{0.82,0.93,0.93}
\definecolor{LightCyan3}{rgb}{0.71,0.80,0.80}
\definecolor{LightCyan4}{rgb}{0.48,0.55,0.55}
\definecolor{LightCyan}{rgb}{0.88,1.00,1.00}
\definecolor{LightGoldenrod1}{rgb}{1.00,0.93,0.55}
\definecolor{LightGoldenrod2}{rgb}{0.93,0.86,0.51}
\definecolor{LightGoldenrod3}{rgb}{0.80,0.75,0.44}
\definecolor{LightGoldenrod4}{rgb}{0.55,0.51,0.30}
\definecolor{LightGoldenrodYellow}{rgb}{0.98,0.98,0.82}
\definecolor{LightGoldenrod}{rgb}{0.93,0.87,0.51}
\definecolor{LightGray}{rgb}{0.83,0.83,0.83}
\definecolor{LightGreen}{rgb}{0.56,0.93,0.56}
\definecolor{LightGrey}{rgb}{0.83,0.83,0.83}
\definecolor{LightPink1}{rgb}{1.00,0.68,0.73}
\definecolor{LightPink2}{rgb}{0.93,0.64,0.68}
\definecolor{LightPink3}{rgb}{0.80,0.55,0.58}
\definecolor{LightPink4}{rgb}{0.55,0.37,0.40}
\definecolor{LightPink}{rgb}{1.00,0.71,0.76}
\definecolor{LightSalmon1}{rgb}{1.00,0.63,0.48}
\definecolor{LightSalmon2}{rgb}{0.93,0.58,0.45}
\definecolor{LightSalmon3}{rgb}{0.80,0.51,0.38}
\definecolor{LightSalmon4}{rgb}{0.55,0.34,0.26}
\definecolor{LightSalmon}{rgb}{1.00,0.63,0.48}
\definecolor{LightSeaGreen}{rgb}{0.13,0.70,0.67}
\definecolor{LightSkyBlue1}{rgb}{0.69,0.89,1.00}
\definecolor{LightSkyBlue2}{rgb}{0.64,0.83,0.93}
\definecolor{LightSkyBlue3}{rgb}{0.55,0.71,0.80}
\definecolor{LightSkyBlue4}{rgb}{0.38,0.48,0.55}
\definecolor{LightSkyBlue}{rgb}{0.53,0.81,0.98}
\definecolor{LightSlateBlue}{rgb}{0.52,0.44,1.00}
\definecolor{LightSlateGray}{rgb}{0.47,0.53,0.60}
\definecolor{LightSlateGrey}{rgb}{0.47,0.53,0.60}
\definecolor{LightSteelBlue1}{rgb}{0.79,0.88,1.00}
\definecolor{LightSteelBlue2}{rgb}{0.74,0.82,0.93}
\definecolor{LightSteelBlue3}{rgb}{0.64,0.71,0.80}
\definecolor{LightSteelBlue4}{rgb}{0.43,0.48,0.55}
\definecolor{LightSteelBlue}{rgb}{0.69,0.77,0.87}
\definecolor{LightYellow1}{rgb}{1.00,1.00,0.88}
\definecolor{LightYellow2}{rgb}{0.93,0.93,0.82}
\definecolor{LightYellow3}{rgb}{0.80,0.80,0.71}
\definecolor{LightYellow4}{rgb}{0.55,0.55,0.48}
\definecolor{LightYellow}{rgb}{1.00,1.00,0.88}
\definecolor{LimeGreen}{rgb}{0.20,0.80,0.20}
\definecolor{MediumAquamarine}{rgb}{0.40,0.80,0.67}
\definecolor{MediumBlue}{rgb}{0.00,0.00,0.80}
\definecolor{MediumOrchid1}{rgb}{0.88,0.40,1.00}
\definecolor{MediumOrchid2}{rgb}{0.82,0.37,0.93}
\definecolor{MediumOrchid3}{rgb}{0.71,0.32,0.80}
\definecolor{MediumOrchid4}{rgb}{0.48,0.22,0.55}
\definecolor{MediumOrchid}{rgb}{0.73,0.33,0.83}
\definecolor{MediumPurple1}{rgb}{0.67,0.51,1.00}
\definecolor{MediumPurple2}{rgb}{0.62,0.47,0.93}
\definecolor{MediumPurple3}{rgb}{0.54,0.41,0.80}
\definecolor{MediumPurple4}{rgb}{0.36,0.28,0.55}
\definecolor{MediumPurple}{rgb}{0.58,0.44,0.86}
\definecolor{MediumSeaGreen}{rgb}{0.24,0.70,0.44}
\definecolor{MediumSlateBlue}{rgb}{0.48,0.41,0.93}
\definecolor{MediumSpringGreen}{rgb}{0.00,0.98,0.60}
\definecolor{MediumTurquoise}{rgb}{0.28,0.82,0.80}
\definecolor{MediumVioletRed}{rgb}{0.78,0.08,0.52}
\definecolor{MidnightBlue}{rgb}{0.10,0.10,0.44}
\definecolor{MintCream}{rgb}{0.96,1.00,0.98}
\definecolor{MistyRose1}{rgb}{1.00,0.89,0.88}
\definecolor{MistyRose2}{rgb}{0.93,0.84,0.82}
\definecolor{MistyRose3}{rgb}{0.80,0.72,0.71}
\definecolor{MistyRose4}{rgb}{0.55,0.49,0.48}
\definecolor{MistyRose}{rgb}{1.00,0.89,0.88}
\definecolor{NavajoWhite1}{rgb}{1.00,0.87,0.68}
\definecolor{NavajoWhite2}{rgb}{0.93,0.81,0.63}
\definecolor{NavajoWhite3}{rgb}{0.80,0.70,0.55}
\definecolor{NavajoWhite4}{rgb}{0.55,0.47,0.37}
\definecolor{NavajoWhite}{rgb}{1.00,0.87,0.68}
\definecolor{NavyBlue}{rgb}{0.00,0.00,0.50}
\definecolor{OldLace}{rgb}{0.99,0.96,0.90}
\definecolor{OliveDrab1}{rgb}{0.75,1.00,0.24}
\definecolor{OliveDrab2}{rgb}{0.70,0.93,0.23}
\definecolor{OliveDrab3}{rgb}{0.60,0.80,0.20}
\definecolor{OliveDrab4}{rgb}{0.41,0.55,0.13}
\definecolor{OliveDrab}{rgb}{0.42,0.56,0.14}
\definecolor{OrangeRed1}{rgb}{1.00,0.27,0.00}
\definecolor{OrangeRed2}{rgb}{0.93,0.25,0.00}
\definecolor{OrangeRed3}{rgb}{0.80,0.22,0.00}
\definecolor{OrangeRed4}{rgb}{0.55,0.15,0.00}
\definecolor{OrangeRed}{rgb}{1.00,0.27,0.00}
\definecolor{PaleGoldenrod}{rgb}{0.93,0.91,0.67}
\definecolor{PaleGreen1}{rgb}{0.60,1.00,0.60}
\definecolor{PaleGreen2}{rgb}{0.56,0.93,0.56}
\definecolor{PaleGreen3}{rgb}{0.49,0.80,0.49}
\definecolor{PaleGreen4}{rgb}{0.33,0.55,0.33}
\definecolor{PaleGreen}{rgb}{0.60,0.98,0.60}
\definecolor{PaleTurquoise1}{rgb}{0.73,1.00,1.00}
\definecolor{PaleTurquoise2}{rgb}{0.68,0.93,0.93}
\definecolor{PaleTurquoise3}{rgb}{0.59,0.80,0.80}
\definecolor{PaleTurquoise4}{rgb}{0.40,0.55,0.55}
\definecolor{PaleTurquoise}{rgb}{0.69,0.93,0.93}
\definecolor{PaleVioletRed1}{rgb}{1.00,0.51,0.67}
\definecolor{PaleVioletRed2}{rgb}{0.93,0.47,0.62}
\definecolor{PaleVioletRed3}{rgb}{0.80,0.41,0.54}
\definecolor{PaleVioletRed4}{rgb}{0.55,0.28,0.36}
\definecolor{PaleVioletRed}{rgb}{0.86,0.44,0.58}
\definecolor{PapayaWhip}{rgb}{1.00,0.94,0.84}
\definecolor{PeachPuff1}{rgb}{1.00,0.85,0.73}
\definecolor{PeachPuff2}{rgb}{0.93,0.80,0.68}
\definecolor{PeachPuff3}{rgb}{0.80,0.69,0.58}
\definecolor{PeachPuff4}{rgb}{0.55,0.47,0.40}
\definecolor{PeachPuff}{rgb}{1.00,0.85,0.73}
\definecolor{PowderBlue}{rgb}{0.69,0.88,0.90}
\definecolor{RosyBrown1}{rgb}{1.00,0.76,0.76}
\definecolor{RosyBrown2}{rgb}{0.93,0.71,0.71}
\definecolor{RosyBrown3}{rgb}{0.80,0.61,0.61}
\definecolor{RosyBrown4}{rgb}{0.55,0.41,0.41}
\definecolor{RosyBrown}{rgb}{0.74,0.56,0.56}
\definecolor{RoyalBlue1}{rgb}{0.28,0.46,1.00}
\definecolor{RoyalBlue2}{rgb}{0.26,0.43,0.93}
\definecolor{RoyalBlue3}{rgb}{0.23,0.37,0.80}
\definecolor{RoyalBlue4}{rgb}{0.15,0.25,0.55}
\definecolor{RoyalBlue}{rgb}{0.25,0.41,0.88}
\definecolor{SaddleBrown}{rgb}{0.55,0.27,0.07}
\definecolor{SandyBrown}{rgb}{0.96,0.64,0.38}
\definecolor{SeaGreen1}{rgb}{0.33,1.00,0.62}
\definecolor{SeaGreen2}{rgb}{0.31,0.93,0.58}
\definecolor{SeaGreen3}{rgb}{0.26,0.80,0.50}
\definecolor{SeaGreen4}{rgb}{0.18,0.55,0.34}
\definecolor{SeaGreen}{rgb}{0.18,0.55,0.34}
\definecolor{SkyBlue1}{rgb}{0.53,0.81,1.00}
\definecolor{SkyBlue2}{rgb}{0.49,0.75,0.93}
\definecolor{SkyBlue3}{rgb}{0.42,0.65,0.80}
\definecolor{SkyBlue4}{rgb}{0.29,0.44,0.55}
\definecolor{SkyBlue}{rgb}{0.53,0.81,0.92}
\definecolor{SlateBlue1}{rgb}{0.51,0.44,1.00}
\definecolor{SlateBlue2}{rgb}{0.48,0.40,0.93}
\definecolor{SlateBlue3}{rgb}{0.41,0.35,0.80}
\definecolor{SlateBlue4}{rgb}{0.28,0.24,0.55}
\definecolor{SlateBlue}{rgb}{0.42,0.35,0.80}
\definecolor{SlateGray1}{rgb}{0.78,0.89,1.00}
\definecolor{SlateGray2}{rgb}{0.73,0.83,0.93}
\definecolor{SlateGray3}{rgb}{0.62,0.71,0.80}
\definecolor{SlateGray4}{rgb}{0.42,0.48,0.55}
\definecolor{SlateGray}{rgb}{0.44,0.50,0.56}
\definecolor{SlateGrey}{rgb}{0.44,0.50,0.56}
\definecolor{SpringGreen1}{rgb}{0.00,1.00,0.50}
\definecolor{SpringGreen2}{rgb}{0.00,0.93,0.46}
\definecolor{SpringGreen3}{rgb}{0.00,0.80,0.40}
\definecolor{SpringGreen4}{rgb}{0.00,0.55,0.27}
\definecolor{SpringGreen}{rgb}{0.00,1.00,0.50}
\definecolor{SteelBlue1}{rgb}{0.39,0.72,1.00}
\definecolor{SteelBlue2}{rgb}{0.36,0.67,0.93}
\definecolor{SteelBlue3}{rgb}{0.31,0.58,0.80}
\definecolor{SteelBlue4}{rgb}{0.21,0.39,0.55}
\definecolor{SteelBlue}{rgb}{0.27,0.51,0.71}
\definecolor{VioletRed1}{rgb}{1.00,0.24,0.59}
\definecolor{VioletRed2}{rgb}{0.93,0.23,0.55}
\definecolor{VioletRed3}{rgb}{0.80,0.20,0.47}
\definecolor{VioletRed4}{rgb}{0.55,0.13,0.32}
\definecolor{VioletRed}{rgb}{0.82,0.13,0.56}
\definecolor{WhiteSmoke}{rgb}{0.96,0.96,0.96}
\definecolor{YellowGreen}{rgb}{0.60,0.80,0.20}
\definecolor{aliceblue}{rgb}{0.94,0.97,1.00}
\definecolor{antiquewhite}{rgb}{0.98,0.92,0.84}
\definecolor{aquamarine1}{rgb}{0.50,1.00,0.83}
\definecolor{aquamarine2}{rgb}{0.46,0.93,0.78}
\definecolor{aquamarine3}{rgb}{0.40,0.80,0.67}
\definecolor{aquamarine4}{rgb}{0.27,0.55,0.45}
\definecolor{aquamarine}{rgb}{0.50,1.00,0.83}
\definecolor{azure1}{rgb}{0.94,1.00,1.00}
\definecolor{azure2}{rgb}{0.88,0.93,0.93}
\definecolor{azure3}{rgb}{0.76,0.80,0.80}
\definecolor{azure4}{rgb}{0.51,0.55,0.55}
\definecolor{azure}{rgb}{0.94,1.00,1.00}
\definecolor{beige}{rgb}{0.96,0.96,0.86}
\definecolor{bisque1}{rgb}{1.00,0.89,0.77}
\definecolor{bisque2}{rgb}{0.93,0.84,0.72}
\definecolor{bisque3}{rgb}{0.80,0.72,0.62}
\definecolor{bisque4}{rgb}{0.55,0.49,0.42}
\definecolor{bisque}{rgb}{1.00,0.89,0.77}
\definecolor{black}{rgb}{0.00,0.00,0.00}
\definecolor{blanchedalmond}{rgb}{1.00,0.92,0.80}
\definecolor{blue1}{rgb}{0.00,0.00,1.00}
\definecolor{blue2}{rgb}{0.00,0.00,0.93}
\definecolor{blue3}{rgb}{0.00,0.00,0.80}
\definecolor{blue4}{rgb}{0.00,0.00,0.55}
\definecolor{blueviolet}{rgb}{0.54,0.17,0.89}
\definecolor{blue}{rgb}{0.00,0.00,1.00}
\definecolor{brown1}{rgb}{1.00,0.25,0.25}
\definecolor{brown2}{rgb}{0.93,0.23,0.23}
\definecolor{brown3}{rgb}{0.80,0.20,0.20}
\definecolor{brown4}{rgb}{0.55,0.14,0.14}
\definecolor{brown}{rgb}{0.65,0.16,0.16}
\definecolor{burlywood1}{rgb}{1.00,0.83,0.61}
\definecolor{burlywood2}{rgb}{0.93,0.77,0.57}
\definecolor{burlywood3}{rgb}{0.80,0.67,0.49}
\definecolor{burlywood4}{rgb}{0.55,0.45,0.33}
\definecolor{burlywood}{rgb}{0.87,0.72,0.53}
\definecolor{cadetblue}{rgb}{0.37,0.62,0.63}
\definecolor{chartreuse1}{rgb}{0.50,1.00,0.00}
\definecolor{chartreuse2}{rgb}{0.46,0.93,0.00}
\definecolor{chartreuse3}{rgb}{0.40,0.80,0.00}
\definecolor{chartreuse4}{rgb}{0.27,0.55,0.00}
\definecolor{chartreuse}{rgb}{0.50,1.00,0.00}
\definecolor{chocolate1}{rgb}{1.00,0.50,0.14}
\definecolor{chocolate2}{rgb}{0.93,0.46,0.13}
\definecolor{chocolate3}{rgb}{0.80,0.40,0.11}
\definecolor{chocolate4}{rgb}{0.55,0.27,0.07}
\definecolor{chocolate}{rgb}{0.82,0.41,0.12}
\definecolor{coral1}{rgb}{1.00,0.45,0.34}
\definecolor{coral2}{rgb}{0.93,0.42,0.31}
\definecolor{coral3}{rgb}{0.80,0.36,0.27}
\definecolor{coral4}{rgb}{0.55,0.24,0.18}
\definecolor{coral}{rgb}{1.00,0.50,0.31}
\definecolor{cornflowerblue}{rgb}{0.39,0.58,0.93}
\definecolor{cornsilk1}{rgb}{1.00,0.97,0.86}
\definecolor{cornsilk2}{rgb}{0.93,0.91,0.80}
\definecolor{cornsilk3}{rgb}{0.80,0.78,0.69}
\definecolor{cornsilk4}{rgb}{0.55,0.53,0.47}
\definecolor{cornsilk}{rgb}{1.00,0.97,0.86}
\definecolor{cyan1}{rgb}{0.00,1.00,1.00}
\definecolor{cyan2}{rgb}{0.00,0.93,0.93}
\definecolor{cyan3}{rgb}{0.00,0.80,0.80}
\definecolor{cyan4}{rgb}{0.00,0.55,0.55}
\definecolor{cyan}{rgb}{0.00,1.00,1.00}
\definecolor{darkblue}{rgb}{0.00,0.00,0.55}
\definecolor{darkcyan}{rgb}{0.00,0.55,0.55}
\definecolor{darkgoldenrod}{rgb}{0.72,0.53,0.04}
\definecolor{darkgray}{rgb}{0.66,0.66,0.66}
\definecolor{darkgreen}{rgb}{0.00,0.39,0.00}
\definecolor{darkgrey}{rgb}{0.66,0.66,0.66}
\definecolor{darkkhaki}{rgb}{0.74,0.72,0.42}
\definecolor{darkmagenta}{rgb}{0.55,0.00,0.55}
\definecolor{darkolive}{rgb}{0.33,0.42,0.18}
\definecolor{darkorange}{rgb}{1.00,0.55,0.00}
\definecolor{darkorchid}{rgb}{0.60,0.20,0.80}
\definecolor{darkred}{rgb}{0.55,0.00,0.00}
\definecolor{darksalmon}{rgb}{0.91,0.59,0.48}
\definecolor{darksea}{rgb}{0.56,0.74,0.56}
\definecolor{darkslate}{rgb}{0.18,0.31,0.31}
\definecolor{darkslate}{rgb}{0.18,0.31,0.31}
\definecolor{darkslate}{rgb}{0.28,0.24,0.55}
\definecolor{darkturquoise}{rgb}{0.00,0.81,0.82}
\definecolor{darkviolet}{rgb}{0.58,0.00,0.83}
\definecolor{deeppink}{rgb}{1.00,0.08,0.58}
\definecolor{deepsky}{rgb}{0.00,0.75,1.00}
\definecolor{dimgray}{rgb}{0.41,0.41,0.41}
\definecolor{dimgrey}{rgb}{0.41,0.41,0.41}
\definecolor{dodgerblue}{rgb}{0.12,0.56,1.00}
\definecolor{firebrick1}{rgb}{1.00,0.19,0.19}
\definecolor{firebrick2}{rgb}{0.93,0.17,0.17}
\definecolor{firebrick3}{rgb}{0.80,0.15,0.15}
\definecolor{firebrick4}{rgb}{0.55,0.10,0.10}
\definecolor{firebrick}{rgb}{0.70,0.13,0.13}
\definecolor{floralwhite}{rgb}{1.00,0.98,0.94}
\definecolor{forestgreen}{rgb}{0.13,0.55,0.13}
\definecolor{gainsboro}{rgb}{0.86,0.86,0.86}
\definecolor{ghostwhite}{rgb}{0.97,0.97,1.00}
\definecolor{gold1}{rgb}{1.00,0.84,0.00}
\definecolor{gold2}{rgb}{0.93,0.79,0.00}
\definecolor{gold3}{rgb}{0.80,0.68,0.00}
\definecolor{gold4}{rgb}{0.55,0.46,0.00}
\definecolor{goldenrod1}{rgb}{1.00,0.76,0.15}
\definecolor{goldenrod2}{rgb}{0.93,0.71,0.13}
\definecolor{goldenrod3}{rgb}{0.80,0.61,0.11}
\definecolor{goldenrod4}{rgb}{0.55,0.41,0.08}
\definecolor{goldenrod}{rgb}{0.85,0.65,0.13}
\definecolor{gold}{rgb}{1.00,0.84,0.00}
\definecolor{gray0}{rgb}{0.00,0.00,0.00}
\definecolor{gray100}{rgb}{1.00,1.00,1.00}
\definecolor{gray10}{rgb}{0.10,0.10,0.10}
\definecolor{gray11}{rgb}{0.11,0.11,0.11}
\definecolor{gray12}{rgb}{0.12,0.12,0.12}
\definecolor{gray13}{rgb}{0.13,0.13,0.13}
\definecolor{gray14}{rgb}{0.14,0.14,0.14}
\definecolor{gray15}{rgb}{0.15,0.15,0.15}
\definecolor{gray16}{rgb}{0.16,0.16,0.16}
\definecolor{gray17}{rgb}{0.17,0.17,0.17}
\definecolor{gray18}{rgb}{0.18,0.18,0.18}
\definecolor{gray19}{rgb}{0.19,0.19,0.19}
\definecolor{gray1}{rgb}{0.01,0.01,0.01}
\definecolor{gray20}{rgb}{0.20,0.20,0.20}
\definecolor{gray21}{rgb}{0.21,0.21,0.21}
\definecolor{gray22}{rgb}{0.22,0.22,0.22}
\definecolor{gray23}{rgb}{0.23,0.23,0.23}
\definecolor{gray24}{rgb}{0.24,0.24,0.24}
\definecolor{gray25}{rgb}{0.25,0.25,0.25}
\definecolor{gray26}{rgb}{0.26,0.26,0.26}
\definecolor{gray27}{rgb}{0.27,0.27,0.27}
\definecolor{gray28}{rgb}{0.28,0.28,0.28}
\definecolor{gray29}{rgb}{0.29,0.29,0.29}
\definecolor{gray2}{rgb}{0.02,0.02,0.02}
\definecolor{gray30}{rgb}{0.30,0.30,0.30}
\definecolor{gray31}{rgb}{0.31,0.31,0.31}
\definecolor{gray32}{rgb}{0.32,0.32,0.32}
\definecolor{gray33}{rgb}{0.33,0.33,0.33}
\definecolor{gray34}{rgb}{0.34,0.34,0.34}
\definecolor{gray35}{rgb}{0.35,0.35,0.35}
\definecolor{gray36}{rgb}{0.36,0.36,0.36}
\definecolor{gray37}{rgb}{0.37,0.37,0.37}
\definecolor{gray38}{rgb}{0.38,0.38,0.38}
\definecolor{gray39}{rgb}{0.39,0.39,0.39}
\definecolor{gray3}{rgb}{0.03,0.03,0.03}
\definecolor{gray40}{rgb}{0.40,0.40,0.40}
\definecolor{gray41}{rgb}{0.41,0.41,0.41}
\definecolor{gray42}{rgb}{0.42,0.42,0.42}
\definecolor{gray43}{rgb}{0.43,0.43,0.43}
\definecolor{gray44}{rgb}{0.44,0.44,0.44}
\definecolor{gray45}{rgb}{0.45,0.45,0.45}
\definecolor{gray46}{rgb}{0.46,0.46,0.46}
\definecolor{gray47}{rgb}{0.47,0.47,0.47}
\definecolor{gray48}{rgb}{0.48,0.48,0.48}
\definecolor{gray49}{rgb}{0.49,0.49,0.49}
\definecolor{gray4}{rgb}{0.04,0.04,0.04}
\definecolor{gray50}{rgb}{0.50,0.50,0.50}
\definecolor{gray51}{rgb}{0.51,0.51,0.51}
\definecolor{gray52}{rgb}{0.52,0.52,0.52}
\definecolor{gray53}{rgb}{0.53,0.53,0.53}
\definecolor{gray54}{rgb}{0.54,0.54,0.54}
\definecolor{gray55}{rgb}{0.55,0.55,0.55}
\definecolor{gray56}{rgb}{0.56,0.56,0.56}
\definecolor{gray57}{rgb}{0.57,0.57,0.57}
\definecolor{gray58}{rgb}{0.58,0.58,0.58}
\definecolor{gray59}{rgb}{0.59,0.59,0.59}
\definecolor{gray5}{rgb}{0.05,0.05,0.05}
\definecolor{gray60}{rgb}{0.60,0.60,0.60}
\definecolor{gray61}{rgb}{0.61,0.61,0.61}
\definecolor{gray62}{rgb}{0.62,0.62,0.62}
\definecolor{gray63}{rgb}{0.63,0.63,0.63}
\definecolor{gray64}{rgb}{0.64,0.64,0.64}
\definecolor{gray65}{rgb}{0.65,0.65,0.65}
\definecolor{gray66}{rgb}{0.66,0.66,0.66}
\definecolor{gray67}{rgb}{0.67,0.67,0.67}
\definecolor{gray68}{rgb}{0.68,0.68,0.68}
\definecolor{gray69}{rgb}{0.69,0.69,0.69}
\definecolor{gray6}{rgb}{0.06,0.06,0.06}
\definecolor{gray70}{rgb}{0.70,0.70,0.70}
\definecolor{gray71}{rgb}{0.71,0.71,0.71}
\definecolor{gray72}{rgb}{0.72,0.72,0.72}
\definecolor{gray73}{rgb}{0.73,0.73,0.73}
\definecolor{gray74}{rgb}{0.74,0.74,0.74}
\definecolor{gray75}{rgb}{0.75,0.75,0.75}
\definecolor{gray76}{rgb}{0.76,0.76,0.76}
\definecolor{gray77}{rgb}{0.77,0.77,0.77}
\definecolor{gray78}{rgb}{0.78,0.78,0.78}
\definecolor{gray79}{rgb}{0.79,0.79,0.79}
\definecolor{gray7}{rgb}{0.07,0.07,0.07}
\definecolor{gray80}{rgb}{0.80,0.80,0.80}
\definecolor{gray81}{rgb}{0.81,0.81,0.81}
\definecolor{gray82}{rgb}{0.82,0.82,0.82}
\definecolor{gray83}{rgb}{0.83,0.83,0.83}
\definecolor{gray84}{rgb}{0.84,0.84,0.84}
\definecolor{gray85}{rgb}{0.85,0.85,0.85}
\definecolor{gray86}{rgb}{0.86,0.86,0.86}
\definecolor{gray87}{rgb}{0.87,0.87,0.87}
\definecolor{gray88}{rgb}{0.88,0.88,0.88}
\definecolor{gray89}{rgb}{0.89,0.89,0.89}
\definecolor{gray8}{rgb}{0.08,0.08,0.08}
\definecolor{gray90}{rgb}{0.90,0.90,0.90}
\definecolor{gray91}{rgb}{0.91,0.91,0.91}
\definecolor{gray92}{rgb}{0.92,0.92,0.92}
\definecolor{gray93}{rgb}{0.93,0.93,0.93}
\definecolor{gray94}{rgb}{0.94,0.94,0.94}
\definecolor{gray95}{rgb}{0.95,0.95,0.95}
\definecolor{gray96}{rgb}{0.96,0.96,0.96}
\definecolor{gray97}{rgb}{0.97,0.97,0.97}
\definecolor{gray98}{rgb}{0.98,0.98,0.98}
\definecolor{gray99}{rgb}{0.99,0.99,0.99}
\definecolor{gray9}{rgb}{0.09,0.09,0.09}
\definecolor{gray}{rgb}{0.75,0.75,0.75}
\definecolor{green1}{rgb}{0.00,1.00,0.00}
\definecolor{green2}{rgb}{0.00,0.93,0.00}
\definecolor{green3}{rgb}{0.00,0.80,0.00}
\definecolor{green4}{rgb}{0.00,0.55,0.00}
\definecolor{greenyellow}{rgb}{0.68,1.00,0.18}
\definecolor{green}{rgb}{0.00,1.00,0.00}
\definecolor{grey0}{rgb}{0.00,0.00,0.00}
\definecolor{grey100}{rgb}{1.00,1.00,1.00}
\definecolor{grey10}{rgb}{0.10,0.10,0.10}
\definecolor{grey11}{rgb}{0.11,0.11,0.11}
\definecolor{grey12}{rgb}{0.12,0.12,0.12}
\definecolor{grey13}{rgb}{0.13,0.13,0.13}
\definecolor{grey14}{rgb}{0.14,0.14,0.14}
\definecolor{grey15}{rgb}{0.15,0.15,0.15}
\definecolor{grey16}{rgb}{0.16,0.16,0.16}
\definecolor{grey17}{rgb}{0.17,0.17,0.17}
\definecolor{grey18}{rgb}{0.18,0.18,0.18}
\definecolor{grey19}{rgb}{0.19,0.19,0.19}
\definecolor{grey1}{rgb}{0.01,0.01,0.01}
\definecolor{grey20}{rgb}{0.20,0.20,0.20}
\definecolor{grey21}{rgb}{0.21,0.21,0.21}
\definecolor{grey22}{rgb}{0.22,0.22,0.22}
\definecolor{grey23}{rgb}{0.23,0.23,0.23}
\definecolor{grey24}{rgb}{0.24,0.24,0.24}
\definecolor{grey25}{rgb}{0.25,0.25,0.25}
\definecolor{grey26}{rgb}{0.26,0.26,0.26}
\definecolor{grey27}{rgb}{0.27,0.27,0.27}
\definecolor{grey28}{rgb}{0.28,0.28,0.28}
\definecolor{grey29}{rgb}{0.29,0.29,0.29}
\definecolor{grey2}{rgb}{0.02,0.02,0.02}
\definecolor{grey30}{rgb}{0.30,0.30,0.30}
\definecolor{grey31}{rgb}{0.31,0.31,0.31}
\definecolor{grey32}{rgb}{0.32,0.32,0.32}
\definecolor{grey33}{rgb}{0.33,0.33,0.33}
\definecolor{grey34}{rgb}{0.34,0.34,0.34}
\definecolor{grey35}{rgb}{0.35,0.35,0.35}
\definecolor{grey36}{rgb}{0.36,0.36,0.36}
\definecolor{grey37}{rgb}{0.37,0.37,0.37}
\definecolor{grey38}{rgb}{0.38,0.38,0.38}
\definecolor{grey39}{rgb}{0.39,0.39,0.39}
\definecolor{grey3}{rgb}{0.03,0.03,0.03}
\definecolor{grey40}{rgb}{0.40,0.40,0.40}
\definecolor{grey41}{rgb}{0.41,0.41,0.41}
\definecolor{grey42}{rgb}{0.42,0.42,0.42}
\definecolor{grey43}{rgb}{0.43,0.43,0.43}
\definecolor{grey44}{rgb}{0.44,0.44,0.44}
\definecolor{grey45}{rgb}{0.45,0.45,0.45}
\definecolor{grey46}{rgb}{0.46,0.46,0.46}
\definecolor{grey47}{rgb}{0.47,0.47,0.47}
\definecolor{grey48}{rgb}{0.48,0.48,0.48}
\definecolor{grey49}{rgb}{0.49,0.49,0.49}
\definecolor{grey4}{rgb}{0.04,0.04,0.04}
\definecolor{grey50}{rgb}{0.50,0.50,0.50}
\definecolor{grey51}{rgb}{0.51,0.51,0.51}
\definecolor{grey52}{rgb}{0.52,0.52,0.52}
\definecolor{grey53}{rgb}{0.53,0.53,0.53}
\definecolor{grey54}{rgb}{0.54,0.54,0.54}
\definecolor{grey55}{rgb}{0.55,0.55,0.55}
\definecolor{grey56}{rgb}{0.56,0.56,0.56}
\definecolor{grey57}{rgb}{0.57,0.57,0.57}
\definecolor{grey58}{rgb}{0.58,0.58,0.58}
\definecolor{grey59}{rgb}{0.59,0.59,0.59}
\definecolor{grey5}{rgb}{0.05,0.05,0.05}
\definecolor{grey60}{rgb}{0.60,0.60,0.60}
\definecolor{grey61}{rgb}{0.61,0.61,0.61}
\definecolor{grey62}{rgb}{0.62,0.62,0.62}
\definecolor{grey63}{rgb}{0.63,0.63,0.63}
\definecolor{grey64}{rgb}{0.64,0.64,0.64}
\definecolor{grey65}{rgb}{0.65,0.65,0.65}
\definecolor{grey66}{rgb}{0.66,0.66,0.66}
\definecolor{grey67}{rgb}{0.67,0.67,0.67}
\definecolor{grey68}{rgb}{0.68,0.68,0.68}
\definecolor{grey69}{rgb}{0.69,0.69,0.69}
\definecolor{grey6}{rgb}{0.06,0.06,0.06}
\definecolor{grey70}{rgb}{0.70,0.70,0.70}
\definecolor{grey71}{rgb}{0.71,0.71,0.71}
\definecolor{grey72}{rgb}{0.72,0.72,0.72}
\definecolor{grey73}{rgb}{0.73,0.73,0.73}
\definecolor{grey74}{rgb}{0.74,0.74,0.74}
\definecolor{grey75}{rgb}{0.75,0.75,0.75}
\definecolor{grey76}{rgb}{0.76,0.76,0.76}
\definecolor{grey77}{rgb}{0.77,0.77,0.77}
\definecolor{grey78}{rgb}{0.78,0.78,0.78}
\definecolor{grey79}{rgb}{0.79,0.79,0.79}
\definecolor{grey7}{rgb}{0.07,0.07,0.07}
\definecolor{grey80}{rgb}{0.80,0.80,0.80}
\definecolor{grey81}{rgb}{0.81,0.81,0.81}
\definecolor{grey82}{rgb}{0.82,0.82,0.82}
\definecolor{grey83}{rgb}{0.83,0.83,0.83}
\definecolor{grey84}{rgb}{0.84,0.84,0.84}
\definecolor{grey85}{rgb}{0.85,0.85,0.85}
\definecolor{grey86}{rgb}{0.86,0.86,0.86}
\definecolor{grey87}{rgb}{0.87,0.87,0.87}
\definecolor{grey88}{rgb}{0.88,0.88,0.88}
\definecolor{grey89}{rgb}{0.89,0.89,0.89}
\definecolor{grey8}{rgb}{0.08,0.08,0.08}
\definecolor{grey90}{rgb}{0.90,0.90,0.90}
\definecolor{grey91}{rgb}{0.91,0.91,0.91}
\definecolor{grey92}{rgb}{0.92,0.92,0.92}
\definecolor{grey93}{rgb}{0.93,0.93,0.93}
\definecolor{grey94}{rgb}{0.94,0.94,0.94}
\definecolor{grey95}{rgb}{0.95,0.95,0.95}
\definecolor{grey96}{rgb}{0.96,0.96,0.96}
\definecolor{grey97}{rgb}{0.97,0.97,0.97}
\definecolor{grey98}{rgb}{0.98,0.98,0.98}
\definecolor{grey99}{rgb}{0.99,0.99,0.99}
\definecolor{grey9}{rgb}{0.09,0.09,0.09}
\definecolor{grey}{rgb}{0.75,0.75,0.75}
\definecolor{honeydew1}{rgb}{0.94,1.00,0.94}
\definecolor{honeydew2}{rgb}{0.88,0.93,0.88}
\definecolor{honeydew3}{rgb}{0.76,0.80,0.76}
\definecolor{honeydew4}{rgb}{0.51,0.55,0.51}
\definecolor{honeydew}{rgb}{0.94,1.00,0.94}
\definecolor{hotpink}{rgb}{1.00,0.41,0.71}
\definecolor{indianred}{rgb}{0.80,0.36,0.36}
\definecolor{ivory1}{rgb}{1.00,1.00,0.94}
\definecolor{ivory2}{rgb}{0.93,0.93,0.88}
\definecolor{ivory3}{rgb}{0.80,0.80,0.76}
\definecolor{ivory4}{rgb}{0.55,0.55,0.51}
\definecolor{ivory}{rgb}{1.00,1.00,0.94}
\definecolor{khaki1}{rgb}{1.00,0.96,0.56}
\definecolor{khaki2}{rgb}{0.93,0.90,0.52}
\definecolor{khaki3}{rgb}{0.80,0.78,0.45}
\definecolor{khaki4}{rgb}{0.55,0.53,0.31}
\definecolor{khaki}{rgb}{0.94,0.90,0.55}
\definecolor{lavenderblush}{rgb}{1.00,0.94,0.96}
\definecolor{lavender}{rgb}{0.90,0.90,0.98}
\definecolor{lawngreen}{rgb}{0.49,0.99,0.00}
\definecolor{lemonchiffon}{rgb}{1.00,0.98,0.80}
\definecolor{lightblue}{rgb}{0.68,0.85,0.90}
\definecolor{lightcoral}{rgb}{0.94,0.50,0.50}
\definecolor{lightcyan}{rgb}{0.88,1.00,1.00}
\definecolor{lightgoldenrod}{rgb}{0.93,0.87,0.51}
\definecolor{lightgoldenrod}{rgb}{0.98,0.98,0.82}
\definecolor{lightgray}{rgb}{0.83,0.83,0.83}
\definecolor{lightgreen}{rgb}{0.56,0.93,0.56}
\definecolor{lightgrey}{rgb}{0.83,0.83,0.83}
\definecolor{lightpink}{rgb}{1.00,0.71,0.76}
\definecolor{lightsalmon}{rgb}{1.00,0.63,0.48}
\definecolor{lightsea}{rgb}{0.13,0.70,0.67}
\definecolor{lightsky}{rgb}{0.53,0.81,0.98}
\definecolor{lightslate}{rgb}{0.47,0.53,0.60}
\definecolor{lightslate}{rgb}{0.47,0.53,0.60}
\definecolor{lightslate}{rgb}{0.52,0.44,1.00}
\definecolor{lightsteel}{rgb}{0.69,0.77,0.87}
\definecolor{lightyellow}{rgb}{1.00,1.00,0.88}
\definecolor{limegreen}{rgb}{0.20,0.80,0.20}
\definecolor{linen}{rgb}{0.98,0.94,0.90}
\definecolor{magenta1}{rgb}{1.00,0.00,1.00}
\definecolor{magenta2}{rgb}{0.93,0.00,0.93}
\definecolor{magenta3}{rgb}{0.80,0.00,0.80}
\definecolor{magenta4}{rgb}{0.55,0.00,0.55}
\definecolor{magenta}{rgb}{1.00,0.00,1.00}
\definecolor{maroon1}{rgb}{1.00,0.20,0.70}
\definecolor{maroon2}{rgb}{0.93,0.19,0.65}
\definecolor{maroon3}{rgb}{0.80,0.16,0.56}
\definecolor{maroon4}{rgb}{0.55,0.11,0.38}
\definecolor{maroon}{rgb}{0.69,0.19,0.38}
\definecolor{mediumaquamarine}{rgb}{0.40,0.80,0.67}
\definecolor{mediumblue}{rgb}{0.00,0.00,0.80}
\definecolor{mediumorchid}{rgb}{0.73,0.33,0.83}
\definecolor{mediumpurple}{rgb}{0.58,0.44,0.86}
\definecolor{mediumsea}{rgb}{0.24,0.70,0.44}
\definecolor{mediumslate}{rgb}{0.48,0.41,0.93}
\definecolor{mediumspring}{rgb}{0.00,0.98,0.60}
\definecolor{mediumturquoise}{rgb}{0.28,0.82,0.80}
\definecolor{mediumviolet}{rgb}{0.78,0.08,0.52}
\definecolor{midnightblue}{rgb}{0.10,0.10,0.44}
\definecolor{mintcream}{rgb}{0.96,1.00,0.98}
\definecolor{mistyrose}{rgb}{1.00,0.89,0.88}
\definecolor{moccasin}{rgb}{1.00,0.89,0.71}
\definecolor{navajowhite}{rgb}{1.00,0.87,0.68}
\definecolor{navyblue}{rgb}{0.00,0.00,0.50}
\definecolor{navy}{rgb}{0.00,0.00,0.50}
\definecolor{oldlace}{rgb}{0.99,0.96,0.90}
\definecolor{olivedrab}{rgb}{0.42,0.56,0.14}
\definecolor{orange1}{rgb}{1.00,0.65,0.00}
\definecolor{orange2}{rgb}{0.93,0.60,0.00}
\definecolor{orange3}{rgb}{0.80,0.52,0.00}
\definecolor{orange4}{rgb}{0.55,0.35,0.00}
\definecolor{orangered}{rgb}{1.00,0.27,0.00}
\definecolor{orange}{rgb}{1.00,0.65,0.00}
\definecolor{orchid1}{rgb}{1.00,0.51,0.98}
\definecolor{orchid2}{rgb}{0.93,0.48,0.91}
\definecolor{orchid3}{rgb}{0.80,0.41,0.79}
\definecolor{orchid4}{rgb}{0.55,0.28,0.54}
\definecolor{orchid}{rgb}{0.85,0.44,0.84}
\definecolor{palegoldenrod}{rgb}{0.93,0.91,0.67}
\definecolor{palegreen}{rgb}{0.60,0.98,0.60}
\definecolor{paleturquoise}{rgb}{0.69,0.93,0.93}
\definecolor{paleviolet}{rgb}{0.86,0.44,0.58}
\definecolor{papayawhip}{rgb}{1.00,0.94,0.84}
\definecolor{peachpuff}{rgb}{1.00,0.85,0.73}
\definecolor{peru}{rgb}{0.80,0.52,0.25}
\definecolor{pink1}{rgb}{1.00,0.71,0.77}
\definecolor{pink2}{rgb}{0.93,0.66,0.72}
\definecolor{pink3}{rgb}{0.80,0.57,0.62}
\definecolor{pink4}{rgb}{0.55,0.39,0.42}
\definecolor{pink}{rgb}{1.00,0.75,0.80}
\definecolor{plum1}{rgb}{1.00,0.73,1.00}
\definecolor{plum2}{rgb}{0.93,0.68,0.93}
\definecolor{plum3}{rgb}{0.80,0.59,0.80}
\definecolor{plum4}{rgb}{0.55,0.40,0.55}
\definecolor{plum}{rgb}{0.87,0.63,0.87}
\definecolor{powderblue}{rgb}{0.69,0.88,0.90}
\definecolor{purple1}{rgb}{0.61,0.19,1.00}
\definecolor{purple2}{rgb}{0.57,0.17,0.93}
\definecolor{purple3}{rgb}{0.49,0.15,0.80}
\definecolor{purple4}{rgb}{0.33,0.10,0.55}
\definecolor{purple}{rgb}{0.63,0.13,0.94}
\definecolor{red1}{rgb}{1.00,0.00,0.00}
\definecolor{red2}{rgb}{0.93,0.00,0.00}
\definecolor{red3}{rgb}{0.80,0.00,0.00}
\definecolor{red4}{rgb}{0.55,0.00,0.00}
\definecolor{red}{rgb}{1.00,0.00,0.00}
\definecolor{rosybrown}{rgb}{0.74,0.56,0.56}
\definecolor{royalblue}{rgb}{0.25,0.41,0.88}
\definecolor{saddlebrown}{rgb}{0.55,0.27,0.07}
\definecolor{salmon1}{rgb}{1.00,0.55,0.41}
\definecolor{salmon2}{rgb}{0.93,0.51,0.38}
\definecolor{salmon3}{rgb}{0.80,0.44,0.33}
\definecolor{salmon4}{rgb}{0.55,0.30,0.22}
\definecolor{salmon}{rgb}{0.98,0.50,0.45}
\definecolor{sandybrown}{rgb}{0.96,0.64,0.38}
\definecolor{seagreen}{rgb}{0.18,0.55,0.34}
\definecolor{seashell1}{rgb}{1.00,0.96,0.93}
\definecolor{seashell2}{rgb}{0.93,0.90,0.87}
\definecolor{seashell3}{rgb}{0.80,0.77,0.75}
\definecolor{seashell4}{rgb}{0.55,0.53,0.51}
\definecolor{seashell}{rgb}{1.00,0.96,0.93}
\definecolor{sienna1}{rgb}{1.00,0.51,0.28}
\definecolor{sienna2}{rgb}{0.93,0.47,0.26}
\definecolor{sienna3}{rgb}{0.80,0.41,0.22}
\definecolor{sienna4}{rgb}{0.55,0.28,0.15}
\definecolor{sienna}{rgb}{0.63,0.32,0.18}
\definecolor{skyblue}{rgb}{0.53,0.81,0.92}
\definecolor{slateblue}{rgb}{0.42,0.35,0.80}
\definecolor{slategray}{rgb}{0.44,0.50,0.56}
\definecolor{slategrey}{rgb}{0.44,0.50,0.56}
\definecolor{snow1}{rgb}{1.00,0.98,0.98}
\definecolor{snow2}{rgb}{0.93,0.91,0.91}
\definecolor{snow3}{rgb}{0.80,0.79,0.79}
\definecolor{snow4}{rgb}{0.55,0.54,0.54}
\definecolor{snow}{rgb}{1.00,0.98,0.98}
\definecolor{springgreen}{rgb}{0.00,1.00,0.50}
\definecolor{steelblue}{rgb}{0.27,0.51,0.71}
\definecolor{tan1}{rgb}{1.00,0.65,0.31}
\definecolor{tan2}{rgb}{0.93,0.60,0.29}
\definecolor{tan3}{rgb}{0.80,0.52,0.25}
\definecolor{tan4}{rgb}{0.55,0.35,0.17}
\definecolor{tan}{rgb}{0.82,0.71,0.55}
\definecolor{thistle1}{rgb}{1.00,0.88,1.00}
\definecolor{thistle2}{rgb}{0.93,0.82,0.93}
\definecolor{thistle3}{rgb}{0.80,0.71,0.80}
\definecolor{thistle4}{rgb}{0.55,0.48,0.55}
\definecolor{thistle}{rgb}{0.85,0.75,0.85}
\definecolor{tomato1}{rgb}{1.00,0.39,0.28}
\definecolor{tomato2}{rgb}{0.93,0.36,0.26}
\definecolor{tomato3}{rgb}{0.80,0.31,0.22}
\definecolor{tomato4}{rgb}{0.55,0.21,0.15}
\definecolor{tomato}{rgb}{1.00,0.39,0.28}
\definecolor{turquoise1}{rgb}{0.00,0.96,1.00}
\definecolor{turquoise2}{rgb}{0.00,0.90,0.93}
\definecolor{turquoise3}{rgb}{0.00,0.77,0.80}
\definecolor{turquoise4}{rgb}{0.00,0.53,0.55}
\definecolor{turquoise}{rgb}{0.25,0.88,0.82}
\definecolor{violetred}{rgb}{0.82,0.13,0.56}
\definecolor{violet}{rgb}{0.93,0.51,0.93}
\definecolor{wheat1}{rgb}{1.00,0.91,0.73}
\definecolor{wheat2}{rgb}{0.93,0.85,0.68}
\definecolor{wheat3}{rgb}{0.80,0.73,0.59}
\definecolor{wheat4}{rgb}{0.55,0.49,0.40}
\definecolor{wheat}{rgb}{0.96,0.87,0.70}
\definecolor{whitesmoke}{rgb}{0.96,0.96,0.96}
\definecolor{white}{rgb}{1.00,1.00,1.00}
\definecolor{yellow1}{rgb}{1.00,1.00,0.00}
\definecolor{yellow2}{rgb}{0.93,0.93,0.00}
\definecolor{yellow3}{rgb}{0.80,0.80,0.00}
\definecolor{yellow4}{rgb}{0.55,0.55,0.00}
\definecolor{yellowgreen}{rgb}{0.60,0.80,0.20}
\definecolor{yellow}{rgb}{1.00,1.00,0.00}

\usepackage{hyperref}
  \hypersetup{
    colorlinks = true,
    urlcolor = {navyblue},     
    linkcolor= navyblue,       
    citecolor=navyblue,        
    filecolor=navyblue,         
    }

\usepackage{bbm}
\usepackage{bbold}
\usepackage{stmaryrd}

\usepackage[all]{xy}
\xyoption{2cell}
\xyoption{curve}
\UseTwocells
\SelectTips{cm}{}

\newcommand{\rules}[2]{\mbox{$\frac%
    {\mbox{\normalsize \rule[-5pt]{0pt}{14pt} $#1$}}
    {\mbox{\normalsize \rule[0pt]{0pt}{10pt}$#2$}}$}}

\theoremstyle{plain}
\newtheorem{theorem}{Theorem}[section]

\newtheorem{proposition}[theorem]{Proposition}
\newtheorem{lemma}[theorem]{Lemma}
\newtheorem{corollary}[theorem]{Corollary}
\theoremstyle{definition}
\newtheorem{definition}[theorem]{Definition}

\newtheorem{example}[theorem]{Example}
\newtheorem{remark}[theorem]{Remark}

\newenvironment{gray}{\begin{color}{gray}}{\end{color}}

\newcommand{\BA}{\mathsf{BA}}
\newcommand{\Stone}{\mathsf{Stone}}
\newcommand{\Cat}{\mathsf{Cat}}
\newcommand{\Set}{\mathsf{Set}}

\newcommand{\Pos}{\mathsf{Pos}}

\newcommand{\Posframed}{\S{\Pos}}

\newcommand{\DL}{\mathsf{DL}}
\newcommand{\Pri}{\mathsf{Pri}}

\newcommand{\twobb}{\mathbbm{2}}
\newcommand{\Two}{\mathbbm{2}}

\newcommand{\acal}{\mathcal{A}}
\newcommand{\bcal}{\mathcal{B}}
\newcommand{\ccal}{\mathcal{C}}

\newcommand{\kcal}{\mathcal{K}}

\newcommand{\xcal}{\mathcal{X}}

\newcommand{\op}{{\mathrm{op}}}
\newcommand{\co}{{\mathrm{co}}}

\newcommand{\id}{\mathrm{id}}

\newcommand{\pskip}{\medskip\noindent}

\newcommand{\eps}{\varepsilon}

\newcommand{\Id}{\mathrm{Id}}

\newcommand{\Rel}{\mathsf{Rel}}

\newcommand{\Span}{\mathsf{Span}}
\newcommand{\Cospan}{\mathsf{Cospan}}

\newcommand{\relf}{\mathit{Rel}}

\newcommand{\commaf}{\mathit{Comma}}
\newcommand{\cocommaf}{\mathit{Cocomma}}
\newcommand{\graphf}{\mathit{Graph}}
\newcommand{\collagef}{\mathit{Collage}}

\newcommand{\Onto}{\mathit{Onto}}
\newcommand{\Emb}{\mathit{Emb}}

\newcommand{\con}{\mathit{con}}
\newcommand{\tot}{\mathit{tot}}

\newcommand{\dashRel}{\mathsf{\textsf{-}Rel}}
\newcommand{\dashPre}{\mathsf{\textsf{-}Pre}}
\newcommand{\dashPos}{\mathsf{\textsf{-}Pos}}
\newcommand{\dashIpl}{\mathsf{\textsf{-}Ipl}}
\newcommand{\CDL}{\mathsf{CDL}}

\title{Stone Duality for Relations}
\author{Alexander Kurz, Andrew Moshier, Achim Jung}
\begin{document}
\date{December 8, 2020}
\maketitle

\begin{abstract} 
We show how Stone duality can be extended from maps to relations. This is achieved by working order-enriched and defining a relation from $A$ to $B$ as both an order-preserving function $A^\op\times B\to \Two$ and as a subobject of $A\times B$. We show that dual adjunctions and equivalences between regular categories, taken in a suitably order-enriched sense, extend to (framed bi)categories of relations.
\end{abstract}

\setcounter{tocdepth}{2}  
{
\small
\tableofcontents      
}

\newpage
\section{Introduction}

In this article we will extend Stone-type dualities  from maps to relations. 
We view relations $A\looparrowright B$ as generalising functions, not as generalising subsets. Accordingly, composition of relations $A\looparrowright B$ and $B\looparrowright C$ and the functorial embedding from maps $A\to B$ into relations $A\looparrowright B$ will play a major role. On the other hand, relations $R\subseteq A_1\times\ldots A_n$  as generalised subsets are outside of the scope of this paper. 

\medskip\noindent
Motivation stems, independently, from domain theory and from duality theory, as we will explain in more detail now.

\paragraph{Domain Theory.}
Starting from Scott~\cite{scott:outline}, domain theory is, at least in part, concerned with describing infinite data as well as continuous functions via finite approximants. This leads to Scott's algebraic domains and approximable maps, the latter being relations between the finite approximants of two domains that capture continuity of functions between the domains themselves. Smyth \cite{smyth92} continued the development of this idea by supposing the finite approximants play the role of propositions in a logic of properties of the domain elements. Abramsky \cite{abramsky:dtlf} investigated a similar idea in the context of SFP domains, providing analysis of a wide variety of domain constructions in terms of relations on the corresponding distributive lattices. Jung and S\"underhauf \cite{jung+sunderhauf} extended the techniques to general stably compact spaces and proximity lattices (distributive lattices equipped with a suitable ``way below'' relation). Kegelmann, Jung and Moshier \cite{multi-lingual} then extended the Jung-S\"underhauf duality to relations on the stably compact spaces. This permitted many constructions (products, coproducts, lifting, etc.) on stably compact spaces to be dealt with by Abramsky's logical form methods. Following up on Kegelmann et al, in a more  purely topological setting, \cite{moshier04} establishes a duality for compact Hausdorff spaces and proximity lattices that satisfy a simple strong form of distributivity.

\medskip\noindent
This duality for compact Hausdorff spaces can be derived from the duality of Boolean algebras and Stone spaces by a sequence of purely category theoretic constructions. In order to do this, one needs to work order-enriched and so the construction starts out from the duality of bounded distributive lattices and Priestley spaces (=ordered Stone spaces) and proceeds as follows.
\begin{itemize}
\item Extend the duality of distributive lattices and Priestley spaces from functions to relations.
\item Complete these relational categories by the (ordered) Karoubi envelope (=ordered splitting of idempotents), obtaining a duality for weakening relations of continuous spaces.
\item Restrict this relational duality to maps.
\end{itemize}
Each step in this construction is purely categorical and, therefore, preserves dual adjunctions. In fact, starting from the dual equivalence of distributive lattices and Priestley spaces we arrive at the dual equivalence of proximity lattices and Nachbin spaces (=ordered compact Hausdorff spaces).\footnote{Splitting of idempotents in relations and then restricting to maps gives the exact completion, see \cite[Sec.3]{carboni+vitale} for ordinary categories and the introduction of \cite{lack-exact-completion} for further references. \cite{bourke:thesis,garner+lack,bourke+garner} develop enriched generalisations of regularity and exactness.} 

\medskip\noindent
In this paper we concentrate on the first step, which consists of extending a duality of maps to a duality of relations.

\paragraph{Duality Theory.} For the applications we have in mind, we need that a relation $A\looparrowright B$ is both on the algebraic side and on the topological side a subobject of $A\times B$, or, in the ordered setting, an upward closed subobject of $A^\op\times B$. But since the dual of a subobject of the product is a quotient of the coproduct and not itself again a subobject of a product, this endeavor seems to be doomed to fail. One of the main points of this article is to show that in the order-enriched setting, for so-called weakening closed relations, it is possible to circumvent these problems by exploiting a duality of certain spans and cospans. 

\medskip\noindent Indeed, at the heart of the construction is the observation that in the order-enriched setting relations can be both tabulated as spans and co-tabulated as cospans. This will allow us to define the Stone dual of a relation $R$ as the cospan obtained from dualising the span tabulating $R$. The main result of this paper shows that this construction extends a given duality of maps to a duality for relations. 

\medskip\noindent
In order to formulate this result precisely we first review order-enriched category theory (Section 2) followed by a study of order-enriched spans and cospans (Section 3).  We then show how the extension from maps to relations works in the category of posets (Section 4). Building on this, we will be in a position to extend the duality of bounded distributive lattices and Priestley spaces to relations (Section 5). As it turns out, this result can be generalized to order-regular categories (Section 6), which do support a general duality theory of relations (Section 7). 

\medskip\noindent\textbf{In a nutshell,} the three technical observations at the heart of the paper are the following.

\medskip\noindent First, if we tabulate a relation $R\subseteq X\times Y$ in sets as a span $X\leftarrow R\rightarrow Y$ and then let $X\leftarrow R'\rightarrow Y$ be the pullback of the pushout of $R$, the two relations $R$ and $R'$ will in general not coincide. On the other hand, if we view $X\leftarrow R\rightarrow Y$ as a span of discrete posets and we let $X\leftarrow R'\rightarrow Y$ be the comma of the cocomma of $R$, then the two relations $R$ and $R'$ are equal. This will be reviewed in detail in Section~\ref{sec:span-cospan-duality}.

\medskip\noindent Second, if we 
\begin{itemize}
\item
tabulate a relation in finite sets as a span $X\leftarrow R\rightarrow Y$, 
\item dualise it to a cospan of Boolean algebras $\Two^X\rightarrow \Two^R\leftarrow\Two^Y$,
\item tabulate it via pullback as a Boolean relation $\Two^X\leftarrow R' \rightarrow\Two^R$,
\item dualise it to a cospan of sets $X\leftarrow \Two^{R'}\rightarrow Y$,
\item tabulate it via pullback as a relation $X\leftarrow R'' \rightarrow Y$,
\end{itemize}
then, in general, the double dual $R''$ will be different from  $R$. On the other hand, if we view $R$ as a relation of discrete posets and we repeat the same steps with the categories of posets and distributive lattices, replacing pullbacks by comma objects, the double dual $R''$ will coincide with the original relation $R$.

\pskip
Third, spans work well on both sides of the duality in order to inherit algebraic and topological structure, allowing us to extend the duality from finite posets and finite distributive lattices to Priestley spaces and distributive lattices (and other similar dualities).

\medskip\noindent\textbf{Examples} of dual relations arise from different questions including the following.
\begin{itemize}
\item Given a topological space equipped with an equivalence relation, preorder or partial order, what is the algebraic structure dual to the quotient of the topological space by its equivalence relation (or by its preorder or by its partial order)?
\item Given a non-deterministic computation formalised as a relation in a category of domains or topological spaces, what is its dual relation between preconditions and postconditions?
\item Given algebraic structure extended with relations, what is its topological dual?
\item In particular, given a sequent calculus formalised as a relation in a category of algebras, what is its dual semantics for which it is sound and complete?
\end{itemize}
Answers to some of these questions in concrete examples (Sections \ref{sec:exles-pos}, \ref{sec:exles-pridl}, \ref{sec:exles-ordreg}, \ref{sec:exles-adjunctions})  are meant to be read before going into the details of the technical developments.

\medskip\noindent\textbf{Contributions} of the paper include:
\begin{itemize}
\item Formula \eqref{eq:Twor-span} for computing the dual of a relation.
\item Example~\ref{exle:unit-interval} showing that, as a consequence of \eqref{eq:Twor-span},   the dual of the relation that quotients Cantor space to the unit interval is the way below relation on the algebra of clopens. 
\item Theorem~\ref{thm:dlpriduality} on the equivalence of Priestley and distributive lattice relations.
\item Theorem~\ref{thm:universalproperty}
on extending functors between concretely order-regular categories  from maps to relations
\item Theorem~\ref{thm:dualityrelations}  on extending equivalences of categories of maps to equivalences of categories of relations.
\item Theorem~\ref{thm:adjunctionrelations} on extending adjunctions of categories of maps to adjunctions of framed bicategories of relations.
\end{itemize}

\paragraph{Related Work.} 

We draw on a  range of previous work. From the point of view of \emph{domain theory} this paper is in the tradition of Abramsky's Domain Theory in Logical From \cite{abramsky:dtlf} and Smyth \cite{smyth92}. Both emphasize domains as systems of data that can be described by finitary (logical) means.  We bring this together with the tradition of domain theory as enriched category theory introduced by Smyth and Plotkin \cite{smyth-plotkin} and continued by eg \cite{wagner,Rutten96,Rutten98,bonsangue-etal,waszkiewicz,stubbe,hofmann}. 
We also rely on Kelly's monograph on \emph{enriched category theory} \cite{kelly} and work by Guitart~\cite{guitart} and Street~\cite{street:fibrations-yoneda,street:fibrations-bicategories} who investigated relations in category-enriched categories whereas we specialise to poset-enriched categories. 
The \emph{categorical theory of relation lifting} started with Barr~\cite{barr} who also showed that the relation lifting of a set-functor is functorial iff the functor preserves weak pullbacks (or exact squares). Work by Trnkova~\cite{trnkova}, Freyd and Scedrov~\cite{FreydS90}, Hermida and Jacobs \cite{hermida+jacobs,hermida}, and Moss~\cite{moss}  has also been influential.
Extending adjunctions (as opposed to equivalences) to categories of relations requires tools from \emph{higher category theory} with work by Grandis and Par\'e \cite{GP99,GP04,Grandis} and Shulman \cite{shulman} being particularly valuable.
In the field of \emph{ordered algebra}, work by Scott \cite{scott:outline,scott:data-types}, Goguen, Thatcher and Wright~\cite{goguen-etal} and, in particular, Bloom and Wright~\cite{p-varieties} and Kelly and Power~\cite{kelly+power} was important, as well as our own continuation \cite{ordered-algebras} which introduced order-regular categories. 
\emph{Weakening relation algebras} are studied by Jipsen and Galatos in \cite{jipsen,galatos+jipsen}.
Our paper is also part of  \emph{coalgebra}, in particular of the line of research extending set-based coalgebra to coalgebras over enriched categories initiated by Rutten~\cite{Rutten98} and Worrell~\cite{Worrell00}. In particular,  we take from \cite{relations-spans,relations-cospans} the insight, ultimately going back to Street~\cite{street:fibrations-bicategories}, that, in the order-enriched setting, relations can be both tabulated and co-tabulated. 
Last but not least, from the field of \emph{duality theory}, we rely on the classical results of Stone \cite{stone:dl} and Priestley \cite{priestley}, summarised in the monographs of Johnstone \cite{StoneSpaces} and Davey and Priestley \cite{davey+priestley}.

\section{Preliminaries on Ordered Category Theory}

We review some known material on order-enriched categories. Most important for us is that  weakening relations can be  both tabulated via spans and co-tabulated via cospans. This observation is pivotal for our duality of relations.

\subsection{Ordered Categories and Weighted Limits}\label{sec:weighted}

An  important aspect of ordered categories is that they offer a richer notion of limits. Of particular importance to us will be the ordered analogues of pullback, pushout and coequalizer, also known as comma object\footnote{The name ``comma object" stems from Lawvere's comma categories.},  co-comma object and co-inserter. Comma objects tabulate (and cocomma objects co-tabulate) relations. Coinserters take quotients wrt theories of inequations. 

\medskip\noindent
Throughout this paper,  $\Pos$ denotes the category of partially ordered sets (aka posets) and order-preserving (aka monotone) functions. 

\medskip\noindent
A \textbf{$\Pos$-category} $\ccal$ is a category in which the homsets are posets and where composition is monotone in both arguments. In other words, a $\Pos$-category is a category enriched over $\Pos$. 
A \textbf{$\Pos$-functor} is a functor that is \textbf{locally monotone}, that is, a functor that preserves the order on the homsets.

\pskip
If $\ccal$ is a $\Pos$-category, then $\ccal^\op$ denotes the $\Pos$-category which turns around the arrows and $\ccal^\co$ denotes the $\Pos$-category which turns around the order on the homsets.

\pskip
Since $\Pos$ is cartesian closed and complete and cocomplete we are in the framework studied in Kelly's monograph \cite{kelly}. If we want to emphasise this, we follow Kelly and prefix the notions with ``$\Pos$-'', but, still following Kelly, we also may drop the prefix if it is clear from the context. If we want to emphasise non enriched categories, we speak of ``ordinary'' categories, ``ordinary'' functors, etc.

\pskip
$\Pos$ is itself a $\Pos$-category. We write $$[A,B]$$ for the poset of maps $A\to B$ ordered pointwise. 

\pskip
Notions such as epi and mono carry over from ordinary category theory to $\Pos$-enriched category theory unchanged. But they are not always the most useful notions. For example, more important to us than injection is \textbf{embedding}, that is, a map $m:A\to B$ in $\Pos$ that is order-reflecting. If $m:A\to B$ is an embedding then $m$ is injective and $A$ inherits the order from $B$.

\begin{definition}
Let $\ccal$ be a $\Pos$-category. An arrow $m:A\to B$ is a \textbf{P-mono} if $\ccal(-,m):\ccal(X,A)\to\ccal(X,B)$ is an embedding. An arrow is a \textbf{P-epi} if it is a P-mono in $\ccal^\op$.
\end{definition}

\begin{remark}
Explicitly, $m$ is a P-mono iff $m\circ f \le m\circ g \ \Rightarrow \ f\le g$ and $e$ is a P-epi iff $f\circ e \le g\circ e \ \Rightarrow \ f\le g$.
\end{remark}

Whereas $\Set$ has epi/mono factorizations, $\Pos$ has P-epi/P-mono factorizations: 

\begin{example}\label{exle:onto-emb}
In $\Pos$ the P-monos are precisely the embeddings and the P-epis are precisely the epis or surjections. They form the $$(\Onto,\Emb)$$ factorization system that will play a major role later.
\end{example}

While pullbacks will continue to play a role in $\Pos$, we also need what could be called order-pullbacks or P-pullbacks, but are more commonly known as quasi-pullbacks or comma objects.

\begin{definition}[comma, P-kernel, cocomma]\label{def:comma}
Given a diagram (aka a cospan)  $A\rightarrow C\leftarrow B$, the \textbf{comma object} (or just comma for short) of the cospan is a span $A\leftarrow W\rightarrow B$ such that in the diagram
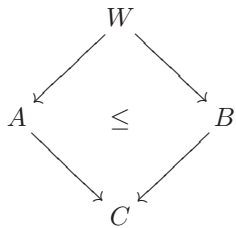
\begin{equation*}
\vcenter{
\xymatrix{
& W\ar[dl]_{}\ar[dr]^{} & \\
A\ar[dr]& \le & B\ar[dl]^{} \\
& C &
}
}
\end{equation*}
the left-hand composition is smaller than the right-hand composition and such that for any other span $A\leftarrow X\rightarrow B $ with this property there is a unique $X\to W$ such that the two triangles in 
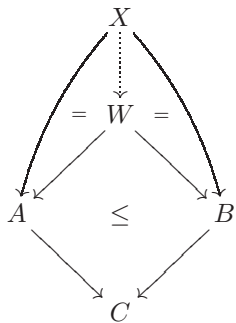
\begin{equation*}
\vcenter{
\xymatrix{
& X\ar@/_/[ddl]^{\ =}\ar@/^/[ddr]_{=\ } \ar@{..>}[d] & \\
& W\ar[dl]_{}\ar[dr]^{} & \\
A\ar[dr]& \le & B\ar[dl]^{} \\
& C &
}
}
\end{equation*}
commute. Moreover, there is a 2-dimensional requirement: If there are two cones $A\stackrel{f_i}{\longleftarrow} X\stackrel{g_i}{\longrightarrow} B$ with $f_1\le f_2$ and $g_1\le g_2$, then also $h_1\le h_2$ for the unique arrows $h_i:X\to W$. 
In the special case where the two legs of the cospan are the same arrow $f$, we speak of the order-kernel or \textbf{P-kernel} of $f$. A \textbf{cocomma} in $\ccal$ is a comma in $\ccal^\op$. 
\end{definition}

While we will encounter comma-objects in other categories than $\Pos$, we will only need to compute it in $\Pos$ itself.

\begin{example}
In $\Pos$,  the comma of the cospan $(j,k)$ is given by $W=\{(a,b) \mid j(a)\le k(b)\}$ together with the two projections on the domain of $j$ and $k$, respectively. The order on $W$ is inherited from the order on $A$ and $B$, that is, the induced $W\to A\times B$ is an embedding. 
\end{example}

The next example highlights one of the reasons why we need to work order-enriched. In the order-enriched setting, the order on a cocomma object $C$ in $A\to C\stackrel{}{\longleftarrow} B$ can encode any weakening relation $R:A\looparrowright B$ (see the next subsection for more on weakening relations). 

\begin{example}
In $\Pos$, the cocomma of a span $A\stackrel{p}{\longleftarrow} R\stackrel{q}{\longrightarrow} B$ is the cospan $A\stackrel{j}{\longrightarrow} C\stackrel{k}{\longleftarrow} B$ where the carrier of $C$ is the disjoint union of $A$ and $B$ and the order on $C$ is inherited from $A$, $B$ and $R$. In detail, $\le_C$ is the smallest partial order satisfying $a \le_C a' \Leftrightarrow a\le_A a'$ and $a \le_C b \Leftarrow aRb$ and $b \le_C b'\Leftrightarrow b\le_B b'$. 
\end{example}

In universal algebra regular factorizations play a crucial role. The regular factorization of an arrow $f$ is obtained by taking the coequalizer of its kernel. In the ordered setting, we factor $f$ by taking the coinserter (or P-coequalizer) of its P-kernel. Intuitively, while coequalizers quotient by equations, coinserters quotient by inequations:

\begin{definition}
Given a pair of  two parallel arrows $(f,g)$ the \textbf{coinserter} $e$ is the universal arrow wrt the property $e\circ f\le e\circ g$. In detail, this means that if there are $k_1\le k_2$ such that $k_i\circ f\le k_i\circ g$ then there are unique $h_1\le h_2$ such that $h_i\circ e= k_i$. An arrow that is a coinserter is also called a \textbf{P-regular epi}.
\end{definition}

\begin{example}
The coinserters in $\Pos$ are precisely the surjections. In fact, in $\Pos$ the notions of surjection, epi, P-epi, and P-regular epi coincide. In the category of preorders, the coinserter of $(p,q)$ with $p,q:X\to Y$ is simply given by $(Y,\sqsubseteq)$ where $\sqsubseteq$ is the smallest preorder containing the order of $Y$ and $\{(p(x), q(x))\mid x\in X\}$. So we see clearly how taking a coinserter corresponds to adding inequations. A coinserter in $\Pos$ is computed by first taking the coinserter in preorders and then quotienting by the equivalence $y\equiv y' \ \Leftrightarrow y\sqsubseteq y' \ \& \ y'\sqsubseteq y$.
\end{example}

\begin{remark}[inserter]\label{rmk:inserter}
An inserter in $\ccal$ is a coinserter in $\ccal^\op$. Inserters will only appear in minor remarks and examples in this paper. It is enough to know that in $\Pos$, the inserter of $(j,k)$ with $(j,k) : X\to Y$ is the subposet of $X$ given by $\{x \mid j(x)\le k(x)\}$.
For a reader who wishes to see examples of how the duality of inserters and coinserters plays out in a setting similar to ours we refer to \cite{dahlqvist+kurz}.
\end{remark}

\begin{remark}[On Terminology]
The point of view of enriched category theory and the one of  universal algebra often suggest different terminology. 
\begin{itemize}
\item Bloom and Wright \cite{p-varieties} noticed that many results in ordered universal algebra can be stated verbatim the same way as the corresponding results in ordinary universal algebra if one is careful about how to define the corresponding notions in the ordered setting. They mark these ordered notions by prefixing them with a ``P-''. Sometimes these notions agree with those from enriched category theory. For example, a P-category is a $\Pos$-category, a P-functor is a $\Pos$-functor, a P-monad is a $\Pos$-monad, but the same is not true for P-monos, P-epis, P-faithful, P-kernel, P-coequalizer. One theme is that P-notions often add a requirement of order-reflection. Another is that P-notions work well with inequational theories instead of only with equational theories. As a rule, in a category with discrete homsets, the P-notions should coincide with the ordinary notions.
\item On the other hand, the category theoretic notions have the advantage that they make sense in other enriched categories. For example, some results in ordered algebra arise as the poset-collapse of more general results from category-enriched categories, which have a  well-developed theory (see eg \cite{street:fibrations-yoneda,street:fibrations-bicategories,street:2-sheaves,kelly+power,bourke:thesis,bourke+garner}) that can be exploited in the poset-enriched setting. 
\item Another advantage of category theoretic notions such as comma object and coinserter is that they include the 2-dimensional aspect of weighted limits, as opposed to Bloom and Wright's P-kernel or P-coequalizer. The 2-dimensional aspect is essential in abstract $\Pos$-categories, but comes for free in $\Pos$ itself, as well as in other concrete $\Pos$-categories, which explains why the difference does not matter for the purposes of this paper. 
\item  We summarize our compromise terminology in Table~\ref{table:terminology}. All of the P-notions are from Bloom and Wright \cite{p-varieties}.

\end{itemize}

\begin{table}
\renewcommand{\arraystretch}{1.2}
\begin{center}
\begin{tabular}{ c|c}
   \hline
  category theory & universal algebra  \\ \hline\hline
  {$\Pos$-category} & \begin{gray}P-category\end{gray} \\ \hline
  {$\Pos$-functor} & \begin{gray}P-functor\end{gray} \\ \hline
   -  & {P-faithful} \\ \hline
   \begin{gray}representably fully faithful\end{gray} & {P-mono} \\ \hline
   - & {P-epi}  \\ \hline
   {comma object} & - \\ \hline
   - & {P-kernel} \\ \hline
    {coinserter} & \begin{gray}P-coequaliser\end{gray} \\ \hline
   \begin{gray}(coinserter)\end{gray} & {P-regular epi} \\ \hline
\end{tabular}
\end{center}
\caption{Summary of Terminology}\label{table:terminology}
\end{table}
\end{remark}

There are other weighted limits than comma objects and inserters. For our purposes, the easiest way to define the totality of all weighted limits is to use a theorem of Kelly \cite[(3.68)]{kelly} which states that if a complete and cocomplete category has the special weighted limits known as powers and the special weighted colimits known as tensors, then it has all weighted limits and all weighted colimits:

\begin{definition}\label{def:power-tensor}
Let $A$ be an object of a $\Pos$-category $\ccal$ and $X\in\Pos$. Then the co-tensor or \textbf{power} $X\pitchfork B$ is defined as the unique up-to-iso solution of the equation
$$[X,\ccal(A,B)]\cong\ccal(A,X\pitchfork B)$$
 and the dual notion of co-power or \textbf{tensor} $X\bullet A$ is determined by
$$[X,\ccal(A,B)]\cong\ccal(X\bullet A,B).$$
\end{definition}

\begin{example}
In posets, the power $X\pitchfork B$ is the poset of monotone functions $X\to B$. In distributive lattices, with $X$ a poset and $B$ a distributive lattice, $X\pitchfork B$ is the distributive lattice of monotone functions $X\to B$.
\end{example}

We can now define completeness in the enriched sense.

\begin{definition}\label{def:complete}
A $\Pos$-category is \textbf{(finitely) complete} if it has (finite) products, equalizers and powers,  and it is \textbf{(finitely) cocomplete} if it has (finite) coproducts, coequalizers and tensors. In particular, a complete $\Pos$-category has commas and inserters and a cocomplete $\Pos$-category has cocommas and coinserters.
\end{definition}

\subsection{Weakening Relations}\label{sec:monotone-relations}
This section introduces the protagonists of this paper, namely  monotone, or weakening-closed, relations. Let 
$$\Rel(\Pos)
\quad\quad
\quad
\textrm{or}
\quad
\quad\quad
\overline\Pos$$ 
denote the $\Pos$-category where objects are posets $A,B,\ldots$, arrows $A\looparrowright B$ are monotone maps $A^\op\times B\to\twobb$, and 2-cells are given pointwise (in other words, if we identify a relation with $\{(a,b)\mid R(a,b)=1\}$, then relations are ordered by set-inclusion). Since $\twobb=\{0<1\}$ is a  poset, all homsets $\overline\Pos(A,B)$ are posets. We let $B(a,b)=1$ if  $a\le_B b$ and $B(a,b)=0$ otherwise. The identity of $A$ is the order of $A$ and composition is ordinary relational composition.  Composition of $R:A\looparrowright B$ and $S:B\looparrowright C$ is written as $R\,;S$ or $S\cdot R$. 

We call these relations \textbf{monotone relations} or \textbf{weakening-closed relations} or \textbf{weakening relations} for short. They are also the $\Pos$-enriched cousins of their category-enriched relatives known as profunctors, distributors, or bimodules. The term weakening-closure derives from the fact that the monotonicity of $A^\op\times B\to\twobb$ amounts to the rule
$$
\frac{a'\le a \,R \,b \le b'}{a' \,R \,b'}
$$
which is known as weakening in the case where  $R$ is a Gentzen-style $\vdash$ in a proof theoretic setting. 

\pskip
For every map $f:A\to B$ in $\Pos$ there is a relation (called companion in \cite{GP04,shulman})
\begin{equation}\label{eq:f_ast}
f_\ast:A\looparrowright B
\end{equation}
given by 
$(a,b) \mapsto B(fa,b):A^\op\times B\to\Two$
and a relation (called adjoint or conjoint in \cite{GP04,shulman})
\begin{equation}\label{eq:f^ast}
f^\ast:B\looparrowright A
\end{equation}
given by 
$(a,b) \mapsto B(b,fa):A\times B^\op\to\Two.$ Recall that a relation $L:A\looparrowright B$ is left-adjoint to $R:B\looparrowright A$, written as $$L\dashv R,$$  if we have (unit) $a\le a' \Rightarrow \exists b\,.\, L(a,b)\wedge R(b,a')$  and (counit) $\exists a\,.\,R(b,a)\wedge L(a,b')\Rightarrow b\le b'$. 
We have that 
$$f_\ast \dashv f^\ast$$ 
in $\overline\Pos$. 
Moreover, the left-adjoints recover the maps among the relations: If we have $L\dashv R$, then there is a monotone function $f$ in $\Pos$ such that $L=f_\ast$ and $R=f^\ast$. 
\footnote{In the discrete setting, a function $f$ and its relation $f_\ast$ are the same set $\{(x,f(x))\}$ of pairs. In the ordered setting, $f_\ast$ corresponds to the set of pairs $\{(x,y)\mid f(x)\le y\}$. To recover $f$ from an adjunction $L\dashv R$ we obtain  from the unit that (i) every $a\in A$ gives rise to an upset $aL=L(a,-)$ and a downset $Ra=R(-,a)$ with non-empty intersection and we obtain from the counit that (ii) this intersection can contain at most one element. Thus $fa=aL\cap Ra$.
} 

\pskip
The functor $(-)_\ast:\Pos\to\overline\Pos$ is covariant on 1-cells and contravariant on 2-cells. 

\pskip
The functor $(-)^\ast:\Pos\to\overline\Pos$ is contravariant on 1-cells and covariant on 2-cells. 

\pskip
This notation can be used to explain how a span $(A\stackrel{p}\leftarrow W\stackrel{q}{\rightarrow} B)$ represents, or \textbf{tabulates}, the relation $\relf(p,q)=q_\ast\cdot p^\ast$ and how a cospan $(A\stackrel{j}\rightarrow C\stackrel{k}{\leftarrow} B)$ represents, or \textbf{cotabulates}, the relation $\relf(j,k)=k^\ast\cdot j_\ast$.

\medskip
We conclude with a couple of useful observations.

\begin{proposition}
The identity relation on $A$ is the comma of $A\stackrel{\id}{\to}A\stackrel{\id}{\leftarrow} A$.
\end{proposition}

\begin{proposition}\label{prop:surjections-embeddings}
A monotone function $m$ is an embedding in $\Pos$ if and only if $m^\ast\cdot m_\ast=\Id$. A monotone function $e$ is a surjection in $\Pos$ if and only if $\Id = e_\ast\cdot e^\ast$. \footnote{One can replace ``$=$'' by ``$\le$'' since the other direction is, respectively, the unit and counit of the adjunction and always holds.}
\end{proposition}

\subsection{Ordered Algebra}\label{sec:ordered-algebra}

Stone duality for relations takes place in an order-enriched setting. To understand the algebraic side of the duality, we review some aspects of order-enriched algebra.
For the purposes of this paper, ordered algebra is $\Pos$-enriched algebra. In particular, all operations are order-preserving. This has the advantage that a relation $A\looparrowright B$ between two ordered algebras can be simply defined as a monotone function $A^\op\times B\to\Two$ such that the legs of the corresponding span are algebra homomorphisms. We explain this now in more detail and conclude with examples of algebraic structure with order-reversing operations.

Our notion of an ordered (quasi)-variety is the one of Bloom and Wright \cite{p-varieties}. As in the ordinary case, a  \textbf{P-variety} can be defined in various equivalent ways. (Recall that a functor is locally monotone if it preserves the order on the homsets.) A P-variety $U:\acal\to\Pos$ is, equivalently,
\begin{itemize}
\item 
a category of algebras with monotone operations for a finitary signature definable by a set of inequations.
\item 
a category of algebras with monotone operations for a finitary signature closed under HSP. Here we need to take H as closure under quotients by inequations, or closure under coinserters,  to use the terminology of Section~\ref{sec:weighted}. Similarly, closure under SP needs to be generalized to include all weighted limits. This can be done by adding closure under powers, or by generalizing closure under S from equalizers to inserters.
\item
a category of algebras for a locally monotone monad $T:\Pos\to\Pos$ that is \emph{strongly finitary} in the sense that it is the $\Pos$-enriched left-Kan extension of its restriction to finite discrete posets.
\end{itemize}
We settle for the last item as our official definition, since it is the most succinct one and  liberates us from repeating the standard definitions of universal algebra such as signature, inequations, and closure under HSP. The equivalence of the last item with the previous two can be found in Theorem 6.9 of \cite{ordered-algebras}, which also contains a full explanation of the technical notions involved as well as further references.

 Convenient properties  that follow from this definition are that coinserters (=quotients by inequations) are surjections and that free constructions as well as the monad $T$ preserve surjections.\footnote{\cite[Thm.6.3]{ordered-algebras} shows that strongly finitary functors preserve surjections. On the other hand, \cite[Exle.6.4]{ordered-algebras} shows that, conversely, being finitary and preserving surjections is not  enough to imply strongly finitary.} This need not be the case for the more general notion of finitary monads on $\Pos$, which are a special case of the notion of algebra theory studied by Kelly and Power~\cite{kelly+power}.

Since surjections coincide with P-regular epis, we can also say that $\acal$ is the category of algebras for a strongly finitary P-regular monad and, since $\Pos$ is P-regular but not regular in the ordinary sense, we may, with a slight abuse of language, simply say that P-varieties are the categories of algebras for a strongly finitary, regular monad on $\Pos$.\footnote{The notion of a P-regular category was introduced in \cite[Def.3.18]{ordered-algebras} where it was simply called regular.}

\begin{example}
The category $\DL$ of bounded distributive lattices is a P-variety.  Note that, as a P-variety, $\DL$ is different from the ordinary variety of distributive lattices as $\DL$ now has ordered homsets. The category $\BA$ of Boolean algebras is the full subcategory of $\DL$ consisting of Boolean algebras. Note that, because $\DL$-morphisms between Boolean algebras preserve negation,  $\DL(B,B')$ is discrete if $B,B'$ are Boolean algebras. 
\end{example}

\begin{remark} 
It was shown in \cite[Thm.12]{dahlqvist+kurz} that the inclusion $\BA\to\DL$ is the free completion of $\BA$ wrt a certain class of inserters. Informally, we may say that $\DL$ is the smallest category containing $\BA$ and closed under $\Pos$-enriched subobjects.
\end{remark}

\begin{definition}\label{def:DLrelation}
A $\DL$-relation $R:A\looparrowright B$ is a weakening relation $UA\looparrowright UB$ that is tabulated by a span in $\DL$.
\end{definition}

\noindent
In other words, $R:A\looparrowright B$ is $\DL$-relation if it is weakening closed and a subalgebra of $A\times B$.
Spelling this out in detail this means that $R$ is closed under the following rules.
\[
\rules{a'\le a R b \le b'}{a'Rb'}
\quad\quad 
\rules{}{0 R 0}
\quad\quad 
\rules{}{1 R 1}
\quad\quad 
\rules{aRb \quad a'Rb'}{(a\wedge a') R (b\wedge b')}
\quad\quad
\rules{aRb \quad a'Rb'}{(a\vee a') R (b\vee b')}
\]
\begin{example} \label{exle:DLrelation}
The property of being a subalgebra interacts with weakening closure in a subtle way.
\begin{enumerate}
\item If $A=B=\twobb\in\DL$, then there are only two relations $A\looparrowright B$, namely the identity $\le_\twobb$ and the total relation.
\item
If $A$ and $B$  are the 3-chain distributive lattice, then the smallest weakening-closed $\DL$-relation $A\looparrowright B$ is $\{(0,0),(0,b), (0,1), (a,1), (1,1)\}$, where the middle elements of the chains are called $a$ and $b$, respectively. This relation is the weakening-closure of the initial span $(p:\twobb\to A,q:\twobb \to B)$. \end{enumerate}
\end{example}

\noindent As we will see in this paper, this interplay between weakening closure and the subalgebra property is crucial to extend Stone duality to relations. It has some, maybe at first sight unexpected, consequences for structures that have order-reversing operations. For example, if we equip Boolean algebras with their natural order then the only weakening-closed $\BA$-relation is the total relation.

\begin{example}\label{exle:BArelation}
Let $R:A\looparrowright B$ be a Boolean relation between Boolean algebras equipped with their natural order. Then $R$ is the total relation. Indeed, because of $(0,0)\in R$, we have, by weakening closure, $(0,1)\in R$ and then by closure under negation $(1,0)\in R$, hence, again by weakening closure, $(a,b)\in R$ for all $a\in A$ and $b\in B$.
\end{example}

This example raises the question of what the order of Boolean algebras should be in the order-enriched setting. There are two possible answers. In the first remark below, the order on a Boolean algebra is discrete, in the second the order is the natural order, as inherited from distributive lattices.

\begin{remark}[The P-variety of Boolean algebras is discrete]
If we want Boolean algebras to form a P-variety, all operations need to be monotone. Since Boolean algebras have negation, the order on the homsets as well as the order on individual Boolean algebras (witnessed by the forgetful functor) must be discrete. This does not contradict closure under ordered quotients since ordered congruences in Boolean algebras are necessarily symmetric and hence equivalence relations \cite[Sec.2.2]{dahlqvist+kurz}. It also does not contradict closure under weighted limits, since the forgetful functor to $\Pos$ has a left-adjoint given by composing the connected component functor $\Pos\to\Set$ with the ordinary free construction of Boolean algebras. More technically, we can say that the category of discretely ordered Boolean algebras is exact in the ordered sense \cite[Exle.3.22]{ordered-algebras}.
It then follows from \cite[Thm.5.9]{ordered-algebras} that the forgetful functor from discretely ordered Boolean algebras to $\Pos$ is a P-variety.
\end{remark}

The previous remark is only of interest since it indicates that the theory of P-varieties specializes to the theory of ordinary varieties in the discrete case. But this discrete point of view fails to exhibit any new order related structure. Therefore, in the rest of the paper, we consider $\BA$ as a full subcategory of the P-variety $\DL$ equipped with the forgetful functor $\BA\to\DL\to\Pos$. From this point of view the homsets between Boolean algebras are still discrete, but their carriers are not.

\begin{remark}[$\BA$ is not a P-variety]
While $\BA\to\DL\to\Pos$ equips Boolean algebras with their natural order, $\BA\to\Pos$ is now not a P-variety, since $\BA$ is not closed under weighted limits. For example, the power $\Two\pitchfork\Two$, see Definition~\ref{def:power-tensor}, is the three element distributive lattice. In fact, every $\DL$ is an inserter of Boolean algebras in a canonical way \cite[Prop.10]{dahlqvist+kurz} and $\DL$ is a closure of $\BA$ under weighted limits.
\end{remark}

Let us note that, mutatis mutandis, the last two remarks also apply to other ordered structures such as Heyting algebras. Our approach can deal with mixed variance only indirectly by embedding the mixed variant signatures into order-preserving signatures. The next example illustrates that this is related to our interest in heterogeneous relations, that is, relations $A\looparrowright B$ where $A\not=B$.

\begin{remark}
Let us illustrate why order-reversing operations present a problem for binary relations $A\looparrowright B$. For example, in the case of Boolean algebras or Heyting algebras, we might want to add, respectively, to Definition~\ref{def:DLrelation} of a $\DL$-relation the clauses
\[
\rules{a\,R\,b}{\neg b \,R \,\neg a}
\quad\quad
\quad\quad
\rules{a_1\,R\,b_1 \quad a_2\,R\,b_2}{(b_1\to a_2) \, R \, (a_1\to b_2)}
\]
as eg in Pigozzi~\cite[Def.2.1]{pigozzi}. The $a$s and $b$s are switching sides and this  presents no problems if $A=B$.  But in this paper we are mainly interested in ``multi-lingual'' relations \cite{multi-lingual} connecting objects $A\not= B$.  
Future work should take a cue from Greco et.al.~\cite{GRMT} who account for rules such as the ones above with the help of opposite relations. 
\end{remark}

\subsection{Ordered Stone Duality}\label{sec:stone-duality}

As we have seen in Section~\ref{sec:ordered-algebra} on Ordered Algebra, weakening relations are more interesting in the ordered category $\DL$ then in the discrete category of Boolean algebras. We therefore decided to treat $\BA$ as a full subcategory of $\DL$ and dualize Boolean relations inside the larger category of distributive lattices. This lines up nicely with the way that Johnstone \cite{StoneSpaces} introduces Stone duality where he first presents the duality of spectral spaces and distributive lattices and then obtains the duality of Stone spaces and Boolean algebras as the discrete restriction. We follow this approach in that we  take the duality for distributive lattice as more fundamental, but find it convenient to rely on Priestley's \cite{priestley} version of the duality  as laid out for example in Davey and Priestley~\cite{davey+priestley}.

Let us recall that the dual equivalence between the category $\DL$ of distributive lattices and the category $\Pri$ of Priestley spaces is mediated by two contravariant functors $\DL(-,\Two)$ and $\Pri(-,\Two)$ which we both abbreviate as $\Two^-$.
We only need to add to this that the two contravariant functors determined by homming into $\Two$

\begin{equation*}
\xymatrix@C=40pt@R=30pt{
\Pri 
\ar@/^/[r]^{{\twobb}^-}  
& {\DL}
\ar@/^/[l]^{{\twobb}^-}  
}
\end{equation*}
are not only a dual equivalence of categories, but also a dual equivalence of $\Pos$-categories, covariant on the order of the homsets. This means, for example, that a cocomma in $\DL$ can be computed as the dual of a comma in $\Pri$. 

\section{The Duality of Spans and Cospans}\label{sec:span-cospan-duality}
Since duality sends the span tabulating a relation to a cospan, we need to understand the relationship of spans and cospans. 
We will see that restricting to weakening-closed spans yields a satisfactory duality. This material owes much to Street \cite{street:fibrations-yoneda,street:fibrations-bicategories} and Guitart~\cite{guitart}.

\subsection{Spans and Cospans}

\label{sec:span-cospan-rel}

Given a $\Pos$-category $\ccal$ and objects $A,B$ we define the $\Pos$-categories
$$\Span(\ccal,A,B)\quad\quad\textrm{and}\quad\quad\Cospan(\ccal,A,B).$$ (We may drop the reference to $\ccal$ in the notation). Objects in  $\Span(\ccal,A,B)$ are spans $(p:W\to A,q:W\to B)$. Arrows $f:(p:W\to A,q:W\to B)\to(p':W'\to A,q':W'\to B)$ are arrows $f\in\ccal$ such that $p'\circ f=p$ and $q'\circ f=q$. $\Cospan$ is defined dually.

\begin{remark}
Every span $(p,q)$ and every cospan $(j,k)$ give rise to relations
$$\relf(p,q)=q_\ast\cdot p^\ast 
\quad\quad\textrm{and}\quad\quad
\relf(j,k)=k^\ast\cdot j_\ast$$
if $\ccal$ is $\Pos$ or just a concrete $\Pos$-category.
\end{remark}

\noindent
In general, if $\ccal$ has comma and cocomma objects, there are $\Pos$-functors
$$\cocommaf:\Span(A,B)\to\Cospan(A,B)$$
and
$$\commaf:\Cospan(A,B)\to\Span(A,B)$$
where $\commaf$ takes a cospan and maps it to its comma square and $\cocommaf$ takes a span and maps it to its co-comma square. Grandis and Pare \cite[Section 5.3]{GP04} describe this as a colax/lax adjunction between double categories of spans as cospans, but the following suffices for our purposes.

\begin{proposition}\label{prop:commacocomma}
$\cocommaf\dashv\commaf$ for all $\Pos$-categories $\ccal$ with comma and cocomma objects. The induced monad and comonad are idempotent. Restricting the functors $\commaf$ and $\cocommaf$ to a skeleton of $\Span(A,B)$ and $\Cospan(A,B)$, this means that $\commaf\circ\cocommaf$ is a closure operator and $\cocommaf\circ\commaf$ is an interior operator.  Moreover, there is a bijection between fixed points of $\commaf\circ\cocommaf$ and fixed points of $\cocommaf\circ\commaf$. Furthermore, if $\ccal=\Pos$, then these fixed points are in bijection with the weakening relations $A\looparrowright B$.
\end{proposition}

\begin{remark}
In case that $\ccal=\Pos$, there is a canonical choice of skeleton of $\Span(A,B)$ given by the weakening closed subsets of $A\times B$. The monad  $\commaf\circ\cocommaf$ then maps a span $(p,q)$ to the graph of the relation $q_\ast\cdot p^\ast$. We write $$\graphf(R)$$ for the graph of a relation $R$. $\commaf$ maps a cospan $(j,k)$ to the graph of the relation $k^\ast\cdot j_\ast$. 
\end{remark}

We can reformulate the definition of graph so that it generalizes to order-regular categories \cite{ordered-algebras}. Instead of fully-faithful we would then say representably fully-faithful (or P-mono) and instead of onto we would say P-regular epi. But in this paper we work concretely over $\Pos$ and we can say surjective and embedding instead.

\begin{definition}[Graph]\label{def:graph}
In the category $\Pos$, we say that a span $(p,q)$ is \textbf{embedding} if  arrow between spans $(p,q)\to\commaf(\cocommaf(p,q))$ is fully faithful; is \textbf{weakening-closed} if $(p,q)\to\commaf(\cocommaf(p,q))$ is onto; is a \textbf{graph} (of a weakening relation) if $(p,q)\to\commaf(\cocommaf(p,q))$ is iso.
A span $(p,q)$ represents a relation $R$ if $q_\ast\cdot p^\ast=R$. We say that $(p,q)$ tabulates $R$ if it is the graph of $R$.
\end{definition}

It follows from the proposition that every relation has not only a unique tabulation as a graph, but also a unique cotabulation, which is known as the collage of a relation and was introduced by Street \cite{street:fibrations-bicategories} to characterize relations in the case of bicategories.

\begin{definition}[Collage]\label{def:collage}
In the category $\Pos$, we say that 
a cospan $(j:A\to C,k:B\to C)$ is \textbf{bipartite} if $\cocommaf(\commaf(j,k))\to(j,k)$ is fully faithful; is \textbf{onto} if $\cocommaf(\commaf(j,k))\to(j,k)$ is onto; is a \textbf{collage} if $\cocommaf(\commaf(j,k))\to(j,k)$ is iso.
A cospan $(j,k)$ represents the relation $k^\ast\cdot j_\ast$ and cotabulates it if $(j,k)$ is bipartite and onto.
\end{definition}

\noindent The terminology is summarised in Table~\ref{table:cospan}.

\begin{example}\label{exle:collage}
In the category $\Pos$ the collage of a relation $R:A\looparrowright B$, or, equivalently, $\cocommaf(p,q)$ of a span tabulating $R$, is given by a poset $C$ such that $C(a,a')=A(a,a')$, $C(a,b)=R(a,b)$ and $C(b,b')=B(b,b')$.  We write 
$$\collagef(R)$$ for this particular cospan cotabulating $R$.
\end{example}

\noindent The next example shows that while the legs of a collage are order-reflecting in $\Pos$, this need not be the case in $\DL$. A similar example can be built in all non-trivial categories of algebras which have a constant. It follows that a general characterization of collages (or cocomma cospans) in algebraic categories needs special investigation.

\begin{example}\label{exle:dlcospan:nonembedding}
Let $(p:W\to A, q:W\to B)$ be the span where $p$ is the identity on the free $\DL$ on one generator $\{a\}$, let $B$ be the initial $\DL$ with elements $\{0 <1\}$, and let $q$ map $a$ to $0$. One verifies that $(q,\id_B)$ is the cocomma of $(p,q)$. And we have $a\not\le 0$ but $q(a)\le q(0)$, so that $q$ is not an embedding. The reason is that we have
$$a \le 0_B = 0_A$$
where the inequation comes from the span and the equation comes from the laws of $\DL$. 
\end{example}

\pskip 
\begin{table}
\begin{center}
\begin{tabular}{|c|c|}
\hline
Spans & Cospans \\
\hline\hline
weakening-closed & bipartite \\
\hline
embedding (full subobject of product) & onto (quotient of coproduct)\\
\hline
graph of  a relation & collage of a relation\\
\hline
\end{tabular}
\caption{Duality of spans and cospans}\label{table:cospan}
\end{center}
\end{table}

\subsection{Exact squares}\label{sec:exactsquares}

Given a diagram 
\begin{equation}\label{eq:square-ab}
\vcenter{
\xymatrix{
& W\ar[dl]_{p}\ar[dr]^{q} & \\
A\ar[dr]_j& \le & B\ar[dl]^{k} \\
& C &
}
}
\end{equation}
in $\Pos$, we always have that $jp\le kq$ implies $\relf(p,q)\le\relf(j,k)$ (``going over is smaller or equal to going under''). A square $((p,q),(j,k))$ with $jp\le kq$ is called exact if $\relf(p,q)=\relf(j,k)$. 
\footnote{Note that $\relf(p,q)=\relf(j,k)$ implies $jp\le kq$ (since $\relf(p,q)=\relf(j,k)$ is $k^\ast\cdot j_\ast = q_\ast\cdot p^\ast$ which implies $j_\ast\cdot p_\ast \ge k_\ast\cdot q_\ast$ which is equivalent to $j\cdot p \le k\cdot q$).}
Without referring to relations this can be expressed equivalently as in 

\begin{definition}
A square in $\Pos$ as in \eqref{eq:square-ab} satisfying $jp\le kq$ is \textbf{exact} if for all $a,b$ such that $ja\le kb$ there is $w$ such that $a\le pw$ and $qw\le b$.
\end{definition}

\noindent In our context, one of the reasons why exact squares are important, is that an exact square says that the span and the cospan represent the same relation. Informally, exact squares represent relations without preference being given to either spans or cospans.
 
 \pskip Let us repeat that \eqref{eq:square-ab} is exact iff 
 \begin{equation}\label{eq:bcc}
 q_\ast\cdot p^\ast=k^\ast\cdot j_\ast,
 \end{equation}
  which is sometimes called the Beck-Chevalley-Condition. It is also important to note that \eqref{eq:bcc} gives the square \eqref{eq:square-ab} a direction, which we denote by 
$$A\looparrowright B.$$
In other words, the span and the cospan in \eqref{eq:square-ab} represent the same relation if read from $A$ to $B$, but not necessarily the other way around (in fact, if the cospan is a collage in $\Pos$, the relation represented by reading the cospan backwards is empty).

\pskip
The following can be verified easily by direct computation.

\begin{proposition}\label{prop:cocommas-exact}
Comma squares and cocomma squares in $\Pos$ are exact.
\end{proposition}

We will see in Section~\ref{sec:CORC}, that most of the results about spans and cospans in $\Pos$ generalise to concretely order-regular categories. The exactness of cocommas is one of the exceptions: It either may fail to hold or require more work. On the other hand, the next two propositions do generalise.

\begin{proposition}\label{prop:commas-in-exact-squares}
A comma in an exact square is the comma of the cospan.
\end{proposition}

\begin{proof} Let $(p',q')$ be a comma and $(j,k)$ a cospan so that the square $(p',q',j,k)$ is exact. Let $(j',k')$ be the cospan of which $(p',q')$ is the comma. Now suppose that $(p,q)$ is a cone over $(j,k)$, that is, $jp\le kq$. We show first that $(p,q)$ is also a cone over $(j',k')$. We know $jp(w)\le kq(w)$ for all $w$. It follows from the exactness of $(p',q',j,k)$ that there is $w'$ such that $p(w)\le p'(w')$ and $q'(w')\le q(w)$, which implies 
$$j'p(w)\le j'p'(w')\le k'q'(w')\le k'q(w).$$
We have shown that $j'p\le k'q$, that is, that $(p,q)$ is a cone over $(j',k')$. Since $(p',q')$ is the comma of $(j',k')$, there is a unique arrow $(p,q)\to(p',q')$. It follows that $(p',q')$ is the comma of $(j,k)$.
\end{proof}

We also have the dual property for cocommas:

\begin{proposition}\label{prop:cocommas-in-exact-squares}
A cocomma in an exact square is the cocomma of the span.
\end{proposition}

\subsection{Identity and Composition of Spans and Cospans}
\label{sec:extension-Pos}

Since we have a correspondence between relations and (co)spans and we know how to compose relations, an obvious question is how to describe composition directly on (co)spans. But let us first quickly look at identities.

\pskip
The span $(\id,\id)$ represents the identity relation, but it is not weakening closed in general. The graph of the identity relation is given by the comma object of the cospan $(\id,\id)$. On the other hand, the collage of the identity relation is simply the cospan $(\id,\id)$.

\pskip
Composition of relations can be done directly on representing spans by taking comma objects (or any exact square, for that matter). If in the diagram 
\begin{equation}\label{eq:spancomp}
\vcenter{
\xymatrix@R=20pt{
&
&
\ar[1,-1]_{r}
\ar[1,1]^{s}
&
&
\\
&
\ar[1,-1]_{p}
\ar[1,1]^{q}
&
&
\ar[1,-1]_{p'}
\ar[1,1]^{q'}
&
\\
&
&
&
&
}
}
\end{equation}
$(r,s)$ is the comma span of $(q,p')$ then $(pr,q's)$ represents $\relf(p',q')\cdot\relf(p,q)$, which is immediate if we have exactness of comma squares.

\pskip It is important to note that this composition does not preserve graphs. For example, if $p=q':2\to 1$ and $q=p'=id_2$ then $\relf(pr,q's)$ is the identity on 1 but $(pr,q's)$ not an embedding span. 

\pskip
But composition of spans does preserve weakening closure:

\begin{proposition}
If in \eqref{eq:spancomp} we have that $(p,q)$ and $(p',q')$ are weakening-closed and 
$(r,s)$ is the comma span of $(q,p')$, then $(pr,q's)$ is weakening-closed.
\end{proposition}

Composition of cospans is done by cocomma squares, dualising \eqref{eq:spancomp}, and relying on exactness of  cocomma squares.

\begin{equation}\label{eq:cospancomp}
\vcenter{
\xymatrix{
\ar[1,1]_{j}
&
&
\ar[1,-1]^{k}
\ar[1,1]_{j'}
&
&
\ar[1,-1]^{k'}
\\
&
\ar[1,1]_{i}
&
&
\ar[1,-1]^{l}
&
\\
&
&
&
&
}}
\end{equation}

\noindent Composition by cospans does not preserve collages. Indeed, similarly to the previous example, if we take $(j,k)$ and $(j',k')$ to be the collages of $\{(0,0),(0,1)\}$ and $\{(0,0),(1,0)\}$ respectively, then $(ij,lk')$ is not a collage (because it is not onto, ie, there are elements neither in the image of $ij$ nor in the image of $lk'$.

\pskip
But composition by cospans does preserve being bipartite: 

\begin{proposition}
If in \eqref{eq:cospancomp} we have that $(j,k)$ and $(j',k')$ are bipartite and $(i,l)$ is the cocomma span of $(k,j')$, then $(ij,lk')$ is bipartite.\end{proposition}

\section{Dual Relations in Posets}\label{sec:Poset}

The purpose of this section is to extend to relations the well-known dualising functor 
$$ f:X\to Y \ \ \mapsto \ \ \Two^f:\Two^Y\to\Two^X$$
taking a monotone function to its inverse image.
As suggested by the previous section, this can be done by applying the functor $\Two^-$ to either the legs of a tabulating span or to the legs of a co-tabulating cospan. We show that these two procedures agree and that $\Two^-$ extends to a functor on $\overline\Pos=\Rel(\Pos)$.

\pskip
The contravariance of $\Two^-$ means that the extension is contravariant on the order of the homsets (2-cells): If $r\subseteq r'$ are two relations, then tabulating them as $(p,q)\to(p',q')$ and applying a contravariant functor $F$ gives cospans 
\begin{equation}\label{eq:contravariant-2-cells}
(Fp',Fq')\to (Fp,Fq).
\end{equation}
As explained in the next remark, it follows that the extension must be covariant on relations (1-cells). (This fits well with relations as ``objects of the arrows-object'' of a  double category \cite{GP04,shulman}, a point of view that will play a role in Section~\ref{sec:extending}.) 

\begin{remark}(Covariance on relations.)
Let $F,G$ be $\Pos$-functors that are contravariant on 1-cells and covariant on 2-cells. Assume that we have a construction $\ccal\mapsto \Rel(\ccal)$ with functors $(-)_\ast:\ccal\to\Rel(\ccal)^\co$ and $(-)^\ast:\ccal\to\Rel(\ccal)^\op$ (or, equivalently, $(-)^\ast:\ccal^\op\to\Rel(\ccal)$). Further assume that there are functors $\overline F$ and $\overline  G$ that are contravariant on 2-cells. Then to complete the diagram
$$
\xymatrix@C=60pt@R=40pt{
\begin{gray}\Rel(\xcal)\end{gray}
\ar@/^/[r]^{\overline F}  
& 
\begin{gray}\Rel(\acal)\end{gray}
\ar@/^/[l]^{\overline G}  
 \\
\xcal \ar@{..>}[u]^{\begin{gray}(-)\end{gray}} 
\ar@/^/[r]^{F} & 
\acal^\op\ar@{..>}[u]_{\begin{gray}(-)\end{gray}}
\ar@/^/[l]^{G}
}
$$
we are forced to set things up in such a way that the extensions $\overline F,\overline G$ are covariant on relations.  Indeed, after fixing one of the embeddings, say we use $(-)_\ast$ on the $\xcal$-side, we need to use the other one $(-)^\ast$ on the $\acal$-side, since this is the only way to accommodate that $\overline F$ and $\overline G$ are contravariant on 2-cells. This in turn forces the extensions  $\overline F,\overline G$ to be covariant on 1-cells.
\end{remark}

The reason for later choosing  $(-)_\ast$ on the space-side and $(-)^\ast$ on the algebra-side is explained at the beginning of Section~\ref{sec:dualitypriestley}.
Here we only need the functors, defined in  Section~\ref{sec:monotone-relations} 
,  $(-)_\ast:\Pos\to\overline\Pos^\co$ and $(-)^\ast:\Pos^\op\to\overline\Pos$, where we continue to abbreviate $\overline\Pos=\Rel(\Pos)$.

\subsection{Extending to Relations via Spans}
We derive condition \eqref{eq:Twor-span}, which allows us to calculate the dual of a relation in specific examples. The formula arises from applying $\Two^-$ to a graph and then converting the resulting cospan to a relation.
Recall that a cospan $(j,k)$ represents the relation $k^\ast\cdot j_\ast$ given by $(x,y)\in k^\ast\cdot j_\ast\ \Leftrightarrow \ j(x)\le k(y)$.

\begin{proposition}[$\overline{\twobb^-}$ via spans]\label{prop:extendviaspan}
Given a weakening relation $r:X\looparrowright Y$ in $\Pos$, define $\overline\Two(r)$ via first converting $r$ into its graph
$X\stackrel{p}{\longleftarrow} R \stackrel{q}{\longrightarrow} Y$
and then applying $\twobb^-$ to the legs of the span, yielding a cospan 
$$\twobb^X\stackrel{\twobb^{p}}{\longrightarrow} 
\twobb^R \stackrel{\twobb^{q}}{\longleftarrow} 
\twobb^Y$$ 
which in turn gives rise to a relation 
\begin{equation}\label{eq:Twor-span-def}
\overline\Two(r) = (\twobb^{q})^\ast\cdot(\twobb^{p})_\ast: \twobb^X\looparrowright \twobb^Y.
\end{equation}
Then 
\begin{align}
(A,B)\in \overline\Two(r) 
\label{eq:Twor-span}
& \ \Leftrightarrow \ R[A]\subseteq B
\end{align}
where $R[A]=\{b \mid \exists a\in A \,.\, aRb\}$.
\end{proposition}

\begin{proof}
We have 
\begin{align*}
(A,B)\in\overline\Two(r) \ & \Leftrightarrow \ \twobb^{p}(A)\subseteq \twobb^{q}(B)
\\ 
& \Leftrightarrow \ \forall x\in X.\, \forall y\in Y.\, x\in A \ \&\  xRy \ \Rightarrow \ y\in B 
\end{align*}
or, in one picture,
\begin{equation}\label{eq:Twor-span-pic}
\vcenter{
\xymatrix{
& R\ar[dl]\ar[dr] & \\
X\ar[dr]_{A} & \le & Y\ar[dl]^{B}\\
& \twobb &
}}
\end{equation}
iff $(A,B)\in\overline\Two(r)$.
\end{proof}

\begin{remark}\label{rmk:many-valued-valuations}
For logics with many-valued valuations in a poset $D$ we have  
$$
(A,B)\in\overline D(r) 
\ \Longleftrightarrow \ \forall x\in X.\, \forall y\in Y.\, xRy \ \Rightarrow \ A(x)\le B(y).
$$
\end{remark}
\subsection{Extending to Relations via Cospans}

In this section, we see that extending $\Two^-$ via cospans gives the same dual relations as the extension via spans from the previous section, see \eqref{eq:Twor-span} and \eqref{eq:Twor-cospan}. The formula \eqref{eq:Twor-cospan} arises from applying $\Two^-$ to a cotabulating cospan and then turning the resulting span into a relation. 
Recall that a span $(p,q)$ represents the relation $q_\ast\cdot p^\ast$ given by $(x,y)\in q_\ast\cdot p^\ast \ \Leftrightarrow \ \exists w\,.\, x\le p(w)\ \& \ q(w)\le y$.

\begin{proposition}[$\overline{\twobb^-}$ via cospans]\label{prop:extendviacospan}
Given a weakening relation $r:X\looparrowright Y$, first convert $r$ into a cospan 
$X\stackrel{j}{\longrightarrow} R \stackrel{k}{\longleftarrow} Y$
and then apply $\twobb^-$, yielding a span 
$$\twobb^X\stackrel{\twobb^{j}}{\longleftarrow} \twobb^R\stackrel{\twobb^{k}}{\longrightarrow} \twobb^Y,$$ 
and hence a 
relation 
\begin{equation}\label{eq:Twor-cospan-def}
\overline\Two(r) = (\twobb^{k})_\ast\cdot(\twobb^{j})^\ast: \twobb^X\looparrowright \twobb^Y.
\end{equation} 
Then we have  $(A,B)\in\overline\Two(r)$ if and only if
\begin{equation}\label{eq:Twor-cospan}
\forall b\in Y.\,\forall a\in X.\, a\in A \ \&\  aRb \ \Rightarrow \ b\in B.\
\end{equation}
\end{proposition}

\begin{proof}
We have by definition of $\overline\Two$ that
\begin{equation}\label{eq:Twor}
(A,B)\in\overline\Two(r) \ \Leftrightarrow \ \exists C\in\twobb^R\,.\,A\subseteq \twobb^{j}(C) \ \& \ \twobb^{k}(C)\subseteq B.
\end{equation}
For ``only if'', assume $x\in A$ and $xRy$. From $A\subseteq \twobb^{j}(C)$ we know $jx\in C$ and from $xRy$ that $jx\le ky$. Since $C:R\to\Two$ is monotone we have $ky\in C$ and it follows from $\twobb^{k}(C)\subseteq B$ that $y\in B$.
For ``if'' define  $C$ to be the upper closure of $\{j(x) \mid x\in A\}$.
\end{proof}

\begin{remark}
Recalling the definition of a collage from Example~\ref{exle:collage}, it is clear that  for the equivalence \eqref{eq:Twor-cospan}, it is crucial that $\Two^R$ consists of upward closed sets. This is also highlighted by  the diagram 
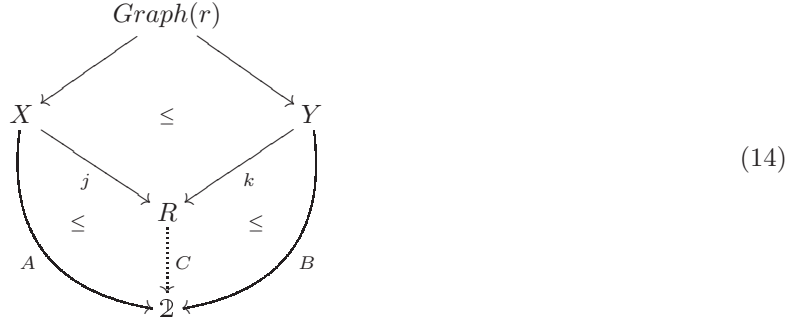
\begin{equation}\label{eq:Tworcospan}
\vcenter{
\xymatrix{
& Graph(r)\ar[dl]_{}\ar[dr]^{} & \\
X\ar[dr]_j \ar@{}[rr]|-{\le} \ar@/_2pc/[ddr]_A\ar@{{}{ }{}}@/_.6pc/[ddr]|-{\le}
& & Y\ar[dl]^{k}  \ar@/^2pc/[ddl]^B\ar@{{}{ }{}}@/^.6pc/[ddl]|-{\le}\\
& R \ar@{..>}[d]^{C}& \\
& \twobb & 
}}
\end{equation}
which can be used to express Proposition~\ref{prop:extendviacospan} more categorically by saying that 
$(A,B)\in\overline\Two(r)$ iff $(A,B)$ is a cocone for the span $\graphf(r)$. As an aside, since $(j,k)$ is a cocomma, we can always find a $C$ for which the ``$\le$'' in the two triangles can be replaced by ``$=$''. This shows that the span $\twobb^X\stackrel{}{\leftarrow} \twobb^R\stackrel{}{\rightarrow} \twobb^Y$ is weakening closed and that the ``$\subseteq$'' in \eqref{eq:Twor} can be replaced by ``$=$''. Finally, comparing \eqref{eq:Twor-span-pic} and \eqref{eq:Tworcospan} explains why we obtain the same result whether dualising the relation $r$ with $R$ being the graph in \eqref{eq:Twor-span-pic} or with $R$ being the collage in  \eqref{eq:Tworcospan}.
\end{remark}

\subsection{Functoriality and Universality of the Extension}

So far in this section, we have seen how to extend the contravariant functor $\Two^-:\Pos\to\Pos$ to relations. In order to know that this extension is functorial and does not depend on a choice of span (or choice of cospan), we need to know that the functor preserves factorizations and exact squares.

\pskip
In more detail, we will employ the general results about extending functors to weakening relations known from \cite[Theorem 4.1]{relations-spans} for the extension via spans and \cite[Theorem 5.10]{relations-cospans} for the extension via cospans. These results have also been presented in the survey \cite{relations-survey} as Theorems~3.8 and 3.10, which may be the most convenient reference for our purposes. We will later need a generalization of \cite[Theorem 4.1]{relations-spans}  from posets to concrete categories over posets. The reader may therefore also refer to Theorem~\ref{thm:universalproperty} (and the dual Theorem~\ref{thm:universalproperty2}) of this paper and instantiate the categories $\xcal$ and $\acal$ with $\Pos$.

\begin{remark}\label{rmk:extendtorelpos}
To conclude that the extension $\overline\Two$ is functorial, we will use the extension-via-spans theorem (see \cite[Thm.4.1]{relations-spans} or \cite[Thm.3.8]{relations-survey} or Theorem~\ref{thm:universalproperty}) which guarantees that the extension $\overline F$ of $F$ in 
$$
\xymatrix@C=4pc{
\overline\Pos^\co \ar[r]^{\ \overline F}& \kcal^\co
\\
\Pos
\ar[ur]_{F}
\ar[u]^{(-)_\ast}
& 
}
$$
is universal and functorial, if $F$ satisfies the following properties.
\begin{enumerate}\label{enumerate}
\item $F$ preserves maps, that is, every $Ff$ has a right adjoint $(Ff)^r$ in $\kcal$ (which is a left-adjoint in $\kcal^\co$).
\item $F$ preserves exact squares, that is,  $Fq\cdot (Fp)^r=(Fk)^r\cdot Fj$ for every exact square \eqref{eq:square-ab}.
\item $Fe\cdot (Fe)^r=\Id$ for all surjections $e$ in $\Pos$.
\end{enumerate}
For the extension via cospans, we have the same theorem with Property 3 being replaced by 
\begin{enumerate}
\setcounter{enumi}{2}
\item $(Fj)^r\cdot (Fj)=\Id$ for all embeddings $j$ in $\Pos$.
\end{enumerate}
To show that $\overline\Two$ is a functor we verify that $(-)^\ast\circ\Two$ satisfies properties 1-3 above. Since the extension is universal and therefore unique, it also follows that the span and the cospan extension agree, giving a different argument for what we have seen by direct calculation in Propositions~\ref{prop:extendviaspan} and \ref{prop:extendviacospan}.
\end{remark}

We first recall the well-known fact that $\Two$ preserves Onto-Embedding factorizations.

\begin{lemma}\label{lem:presfact}
The contravariant functor $\twobb^-:\Pos\to\Pos$ maps surjections to embeddings and embeddings to surjections.
\end{lemma}

\begin{proof}
Let $f:X\to Y$, hence $\twobb^f: [Y,\twobb] \to [X,\twobb]$
. If $f$ is onto and $\twobb^f(p)\le\twobb^f(q)$, then $p\circ f \le q\circ f$ and $p\le q$, proving that $\twobb^f$ is an embedding. If $f$ is an embedding and $p:X\to\twobb$, then there is $q:Y\to\twobb$ such that $q\circ f=p$ (eg, one can take $q$ as the left or right Kan-Extension of $p$ along $f$).
\end{proof}

Of central importance is that $\twobb^-$ preserves exact squares:

\begin{lemma}\label{lem:presexact}
Let 
\begin{equation}\label{eq:square-xy}
\vcenter{
\xymatrix{
& W\ar[dl]_{p}\ar[dr]^{q} & \\
X\ar[dr]_f \ar@{}[rr]|-{\le}& & Y\ar[dl]^{g} \\
& Z & 
}
}
\end{equation}
be an exact square, that is, 
$
f\circ p \,\le\, g\circ q
$
and
$
\forall x,y\,.\, (fx\le gy\ \Rightarrow \ \exists w\,.\,x\le pw \ \& \ qw\le y)
$
or, equivalently,
$$q_\ast\cdot p^\ast = g^\ast\cdot f_\ast.$$
Then
\begin{equation}\label{eq:squareTwo}
\vcenter{
\xymatrix{
& \twobb^W\ar@{<-}[dl]_{\twobb^p}\ar@{<-}[dr]^{\twobb^q} & \\
\twobb^X\ar@{<-}[dr]_{\twobb^f} \ar@{<-}@{}[rr]|-{\le}& & \twobb^Y\ar@{<-}[dl]^{\twobb^g} \\
& \twobb^Z & 
}
}
\end{equation}
is exact, that is,
$$(\twobb^q)^\ast\cdot(\twobb^p)_\ast=(\twobb^g)_\ast\cdot(\twobb^f)^\ast \ .$$
\end{lemma}

\begin{proof}
Assume $ap \le bq$. We have to show that there is $c$ such that $a\le c f$ and $cg \le b$. 
\begin{equation}\label{eq:squareTwo2}
\vcenter{
\xymatrix@R=20pt{
& W\ar[dl]_{p}\ar[dr]^{q} & \\
X\ar[dr]_f \ar@{}[rr]|-{\le} \ar@/_2pc/[ddr]_a\ar@{{}{ }{}}@/_.6pc/[ddr]|-{\le}
& & Y\ar[dl]^{g}  \ar@/^2pc/[ddl]^b\ar@{{}{ }{}}@/^.6pc/[ddl]|-{\le}\\
& Z \ar@{..>}[d]^{c}& \\
& \twobb & 
}
}
\end{equation}
Let $c=\{z\in Z\mid \exists x\in a\,.\, f(x)\le z\}$. Then $a\le cf$. It remains to show $cg\le b$, which follows from $ap\le bq$. Indeed, if $gy\in c$ then there is $x\in a$ such that $fx\le gy$. From exactness, we get a $w$ such that $x\le pw$ and $qw\le y$, which, together with our assumption, implies $a(x)\le b(y)$, that is, due to $x\in a$, the required $y\in b$.
\end{proof}

\begin{remark}\label{rmk:exact-all-spans}
The proof does not depend on the span $(p,q)$ being weakening closed. This can be used to simplify the  computation of dual relations by choosing smaller generating spans.
\end{remark}

Relying on terminology from Section~\ref{sec:monotone-relations} and \ref{sec:extension-Pos}, we are now ready to prove

\begin{theorem}\label{thm:Twor-functor}
The extensions $\overline\Two:\Rel(\Pos)^\co\to\Rel(\Pos)$ of $\Two$ defined by applying $\Two$ to a tabulating span as in Proposition~\ref{prop:extendviaspan} or to a co-tabulating cospan as in Proposition~\ref{prop:extendviacospan} agree and are functorial. They are also universal wrt the properties 1-3 on page \pageref{enumerate}. 
Moreover, $\overline\Two_{X,Y}:\Rel(\Pos)(X,Y)^\op\to\Rel(\Pos)(\Two^X,\Two^Y)$, $R\mapsto \{(a,b)\mid R[a]\subseteq b\}$ is a complete meet-semilattice homomorphism.
\end{theorem}

\begin{proof}
Let $F=(-)^\ast\cdot\Two$ in Remark~\ref{rmk:extendtorelpos}.
To verify Property 1 of Remark~\ref{rmk:extendtorelpos}, we note that 
$F$ takes a map $f:X\to Y$ and sends it to the relation $Ff=(\Two^f)^\ast$, see Section~\ref{sec:monotone-relations}. This verifies that $(\Two^f)^r=(\Two^f)_\ast$ is the left-adjoint of $Ff$ in $\Pos=\kcal^\co$ and the right-adjoint of $Ff$ in $\Pos^\co=\kcal$.

For Property 2, we use that $\Two$ preserves exact squares by Lemma~\ref{lem:presexact}. That $(-)^\ast$, and $(-)_\ast$, preserves exact squares is immediate from writing out the definitions.

Both the span and the cospan version of Property 3 follow from Proposition~\ref{prop:surjections-embeddings} and Lemma~\ref{lem:presfact}.

Finally, we need to verify that $\overline F$ agrees with $\overline\Two$ as defined in Propositions~\ref{prop:extendviaspan} or \ref{prop:extendviacospan}. 
To this end, thanks to the universality and uniqueness of $\overline F$, it suffices to show that ${\overline\Two}\circ(-)_\ast=Ff$.
In case of the extension by spans, on the left-hand side, a map $f:X\to Y$ is sent by $(-)_\ast$ to the span $(p,q)$ of the cospan $(f,\id)$, which in turn is then dualised by 
$\overline \Two$ 
to $(\Two^q)^\ast\cdot(\Two^p)_\ast$, see \eqref{eq:Twor-span-def}. Since $(p,q)$ and $(\id,f)$ both form exact squares with the cospan $(f,\id)$ and since $\Two^-$ preserves exact squares, we have $(\Two^q)^\ast\cdot(\Two^p)_\ast= (\Two^f)^\ast\cdot\id_\ast = Ff$.

In case of the extension by cospans, on the left-hand side, a map $f:X\to Y$ is sent by $(-)_\ast$ to the cospan $(f,\id)$, which in turn is then dualised by 
$\overline \Two$ 
to $\id_\ast\cdot(\Two^f)^\ast$, see \eqref{eq:Twor-cospan-def}, which equals $Ff$.
\end{proof}

\begin{remark}
To extend functors to relations via spans, it is in fact sufficient to require that the functor preserves exact squares with weakening-closed spans, since spans are composed by commas (see Section~\ref{sec:exactsquares}) and commas are weakening closed. 
\end{remark}
\begin{remark}[Independence of choice of span]\label{rmk:allspans}
Under the conditions of the extension theorem, it is the case that the relation lifting on a relation $R$ can be computed by applying the functor to any representing \emph{weakening-closed} span.
But if, as it is the case in our situation, a category has cocomma objects and cocomma objects are exact, then the relation lifting can be computed on any span, including those that are not weakening closed. This follows from the facts that (i) two spans represent the same relation iff they have isomorphic cocommas (Proposition~\ref{prop:commacocomma}), that (ii) cocommas are exact (Proposition~\ref{prop:cocommas-exact}) and that (iii) $\Two^-$ preserves exact squares (Lemma~\ref{lem:presexact} and Remark~\ref{rmk:exact-all-spans}). A dual argument shows that if two cospans represent the same relation, then applying $\Two^-$ to both cospans gives the same relation.\end{remark}

Theorem~\ref{thm:adjunctionrelations} will show that we can extend not only the functor $\Two^-$, but also to the adjunction $\Two^-\dashv(\Two^-)^\op:\Pos^\op\to\Pos$.

\subsection{Examples}\label{sec:exles-pos}

We illustrate different interpretations of the dual $\overline\Two(R)$ of a relation $R$
$$ 
(A,B)\in\overline\Two(R) \ \Leftrightarrow \ R[A]\subseteq B.
$$
through four different applications to Hoare logic, duality theory, domain theory and coalgebraic logic.

\paragraph{Hoare Logic.} First, an example from program verificiation and the relational theory of computation.

\begin{example} If $R\subseteq X\times X$ is the relation representing a non-deterministic computation, then $(A,B)\in\overline\Two(R)$ iff inputs satisfy $A$ then outputs satisify $B$. In other words, $(A,B)\in\overline\Two(R)$ iff $(A,B)$ are a pair of pre- and post-conditions of the computation $R$, or, $$\{A\}R\{B\}$$ in a notation common in program verification and Hoare logic. Note that this is indeed a weakening relation as we have as one of the rules of Hoare logic
$${\frac{A' \le  A \quad\{A \}\  R\ \{B\} \quad B\le B'}{\{A' \}\  R \  \{B'\}}}.$$
Moreover, the meet preserving function $\overline\Two_{X,Y}$ maps a relation to its theory of precondition and postcondition pairs while its left-adjoint
$$
\xymatrix@C=8pc{
\Rel(\Pos)(X,Y)^\op 
\ar@{}[r]|-\bot
\ar@/_1pc/[r]_{\textrm{Theory}}
\ar@{<-}@/^1pc/[r]^{\textrm{Implementation}}
& 
\Rel(\Pos)(\Two^X,\Two^Y)
}
$$
takes a `specification' $S\in\Rel(\Pos)(\Two^X,\Two^Y)$ to its largest relation `implementation' $\bigcup\{R\mid S\subseteq \overline\Two(R)\}$.\qed
\end{example}

\paragraph{Duality Theory.} We describe quotienting by an equivalence relation or preorder in terms of the dual relation. We emphasise that even to describe the dual of equivalence relations on a discrete set one is led to consider weakening relations with respect to a non-discrete order, namely the inclusion order between subsets. From a technical point of view, this stems, on the one hand, from the fact that we work with a dualising object $\Two$ that is equipped with an order and, on the other hand, from the fact that the relationship between spans and cospans is not mediated via pullback/pushout but via comma/cocomma, see Definition~\ref{def:comma} and Section~\ref{sec:span-cospan-duality}. Working with a discrete $\Two$ and with discrete spans/cospans, we would not obtain a dual equivalence between, say, relations on finite sets and relations on finite Boolean algebras.

\begin{example}\label{exle:equivalence-relation}
Let $R$ be a relation on a set $X$. 
\begin{enumerate}
\item If $R$ is reflexive then $\overline\Two(R)\subseteq \Id_{\Two^X}$.
\item If $R$ is reflexive and transitive then $\overline\Two(R)\subseteq \Id_{\Two^X}$ and $\overline\Two(R);\overline\Two(R)\supseteq\overline\Two(R)$. Such a relation $\overline\Two(R)$ is called \emph{interpolative}.
\item If $R$ is an equivalence relation, then $\Two^{X/R}$ is bijective to the set $\{(A,A)\mid R[A]\subseteq A\}$ of reflexive elements of $\overline\Two(R)$.
\end{enumerate}
Reflexive and transitive relations are idempotent relations above identity and interpolative relations are idempotent relations below identity. So item 2 becomes the obvious statement that duality maps idempotent relations above identity to idempotent relations below identity. Since reflexive and transitive relations are monads, we can also view item 2 as the duality of monads and comonads.\qed
\end{example}

The example generalises to posets $X$. The first two items transfer verbatim, noting that $\Id_A$ now refers to the order of $A$.
\begin{example}\label{exle:preorder-quotient}
Let $R$ be a weakening relation on a poset $X$. 
If $R$ is a preorder, then the set $\Two^{X/R}$ of upper sets of $X/R$ is bijective to  $\{(A,A)\mid R[A]\subseteq A\}$. Here $X/R$ is the partial order quotient of $X$ wrt $R$.
\end{example}
\noindent
These observations will lead to new duality results for categories where objects are endo-relations, see Section~\ref{sec:exles-ordreg}.

\paragraph{Bitopological spaces.}
We present an example from the theory of bitopological spaces. A bitopological space $(X,\tau_-,\tau_+)$ is a set $X$ with two topologies. While certain complete lattices, known as frames can be considered as algebraic duals of topological spaces (see \cite{StoneSpaces} for details), pairs of lattices $(L_-,L_+)$ dualise bitopological spaces. This setting is of interest  because adding two weakening relations
$$\con:L_-\looparrowright L_+^\partial \quad\quad\quad \tot:L_+^\partial\looparrowright L_-$$
to the pair of lattices one can characterise a large class of well-known topological spaces by a finitary structure \cite{bitopological}. The functor from bitopological spaces to so-called d-frames is easily explained. It takes a space $(X,\tau_-,\tau_+)$ to the frames $L_-=\tau_-$ and $L_+=\tau_+$ with $\con$ defined as the set of pairs $(a_-,a_+)\in\tau_-\times\tau_+$ such that $a_-\cap a_+=\emptyset$ and $\tot$ as the set of pairs such that $a_-\cup a_+=X$. (The names $\con$ and $\tot$ should remind us of `consistent' and `total'.) The functor from d-frames to bitopological spaces takes a structure $(L_-,L_+,\con,\tot)$ to the bitopological space $(X,\tau_-,\tau_+)$ where $X$ is the set of pairs $(p_-,p_+)$ of frame morphisms $p_-:L_-\to\Two$ and $p_+:L_+\to\Two$ such that 
\begin{align}
\label{eq:con}
\forall (a_-, a_+)\in\con\,.\, p_-(a_-)=0 \ \textrm{or}\  \ p_+(a_+)=0\\
\label{eq:tot}
\forall  (a_+,a_-)\in\tot \,.\,   p_-(a_-)=1 \ \textrm{or}\  \ p_+(a_+)=1
\end{align}
and the topologies $\tau_-$ and $\tau_+$ are generated by basic opens $\{p \mid p(a) = 1\}$ where $a$ ranges over $L_-$ and $L_+$, respectively.
Using formula \eqref{eq:Twor-span} to compute the dual of a relation, one can verify

\begin{example}[Duals of d-frames]
The carrier of the dual of a d-frame $(L_-,L_+,\con,\tot)$ is the intersection $\overline\Two(\con)\cap\overline\Two(\tot)$ of the dual of $\con$ and the dual of $\tot$. For the proof, one verifies that $(p_-,p_+)$ satisfies \eqref{eq:con} iff $(p_-,p_+)\in\overline\Two(\con)$ and that 
$(p_-,p_+)$ satisfies \eqref{eq:tot} iff $(p_-,p_+)\in\overline\Two(\tot)$. 
\end{example}

\paragraph{Modal Logic.}
In modal and coalgebraic logic the notion of $R$-coherent pairs arises from the study of bisimulations for so-called neighbourhood frames \cite{HKP}.

\begin{example}
A quick look at Definition 2.1 in \cite{HKP} of $R$-coherent pairs shows that, given a relation $R\subseteq X_1\times X_2$, a pair $(U_1,U_2)$ with $U_i\subseteq X_i$ is $R$ coherent if $(U_1,U_2)\in\overline\Two(R)$ and $(U_2,U_1)\in\overline\Two(R^{-1})$. 
\end{example}

\noindent
Due to the presence of the converse relation $R^{-1}$ in the definition above, given $X_1\leftarrow R\rightarrow X_2$, the relation of $R$-coherence is the pullback of $\Two^{X_1}\rightarrow \Two^R\leftarrow\Two^{X_2}$. This observation opened the way to coalgebraic generalisations \cite{BH17,GHK}. It would be interesting to pursue these in the ordered setting.

\section{Dual Relations in Priestley Spaces}\label{sec:Priestley}

We will use the results from the previous section on weakening relations to show that the dual equivalence of Priestley spaces and distributive lattices extends from maps to relations.

\subsection{Priestley Spaces and Distributive Lattices}

We start out by defining the category $\Rel(\DL)$ of distributive lattice relations and the category  $\Rel(\Pri)$ of Priestley relations.
We defined $\DL$-relations in Definition~\ref{def:DLrelation}.

\begin{definition}[$\Rel(\DL)$]\label{def:reldl}
The category $\Rel(\DL)$, abbreviated to $\overline\DL$, has the same objects as $\DL$ and $\DL$-relations as arrows. Homsets are ordered by inclusion. 
\end{definition}

$\Pri$-relations can be defined in the same way. Recall that a Priestley space $(X,\le,\tau)$ is a compact Hausdorff space $(X,\tau)$ with an order relation satisfying the Priestley separation axiom, that is, $x\not\le y$ only if there is a clopen downset $U_-$ and a clopen upset $U_+$ such that $U_-\cap U_+=\emptyset$ and $x\in U_+$ and $y\in U_-$.

\begin{definition}[$\Rel(\Pri)$]\label{def:relpri}
A $\Pri$-relation $A\looparrowright B$ is a topologically-closed and upward-closed subspace of $A^\op\times B$.  The category $\Rel(\Pri)$, or $\overline\Pri$ for short, has the same objects as $\Pri$ and $\Pri$-relations as arrows. Homsets are ordered by inclusion.
\end{definition}

For future reference we prove some properties that will be needed later. In particular, the properties below establish that $\DL$ and $\Pri$ are examples of concretely order-regular categories as defined in Section~\ref{sec:poscat}. (Note that if a functor $\acal\to\bcal$ creates limits or lifts limits and $\bcal$ is complete, then the functor preserves limits.)

\begin{proposition}\label{prop:DLPos} \ 
$U:\DL\to\Pos$ creates $\Pos$-limits and $(\Onto,\Emb)$ factorisations. $\DL$ is order-regular. Comma squares in $\DL$ are exact. Identities and composition in $\Rel(\DL)$ are inherited from $\Rel(\Pos)$. 
\end{proposition}

\begin{proof}
These properties of the first two sentences are true for all P-varieties (and P-quasi-varieties) \cite{ordered-algebras}. The others follow from this.
\end{proof}

\begin{proposition} \ \label{prop:PriPos}
The forgetful functor $V:\Pri\to\Pos$ lifts $\Pos$-limits and factorisations uniquely. Comma squares in $\Pri$ are exact. Identities and composition in $\Rel(\DL)$ are inherited from $\Rel(\Pos)$. 

\end{proposition}

\begin{proof}
1) Ordinary  limits in $\Pri$ are equalisers of products equipped with the subspace topology. Cotensors $I\pitchfork X$, with $I$ a  poset and $X$ a Priestley space, are given by $\{(x_i)_{i\in I} \mid i\le_I j \ \Rightarrow \  x_i\le_X x_j\}$, which is a closed subspace of the $|I|$-fold power of $X$ and hence a Priestley space. It follows from \cite[Theorem 3.73]{kelly} that all weighted limits exist. $V$ lifts these limits uniquely, since the property of being a limit prescribes that limits must be equipped with the subspace topology. 
2) $\Pri$ has a factorisation system consisting of embeddings with the subspace topology and surjections. 
3) The statements about comma squares, identities and composition follow from the above.
\end{proof}

The next lemma contains the crucial technical observation.

\begin{lemma}\label{lem:DLPri-pres-exact-squares}
The contravariant functors  $\twobb^-:\DL\to\Pri$ and $\twobb^-:\Pri\to\DL$ preserve exact squares. 
\end{lemma}

\begin{proof}
For the proof, we use the notation of Lemma~\ref{lem:presexact}.  
For $\twobb^-:\DL\to\Pri$, suppose we have the exact square \eqref{eq:square-xy} in $\DL$ and its image under $\Two$ in $\Pri$ as in \eqref{eq:squareTwo}. We need to show that we can find an appropriate $\DL$-morphism $c$ in \eqref{eq:squareTwo2}. The forward image of $a$ via $f$ is a filter basis, that is, 
\begin{equation}\label{eq:fa+}
f[a_+] = \{z\in Z \mid \exists  x\,.\, a(x)=1 \ \& \  f(x)\leq z\}
\end{equation} 
is a filter. Likewise, 
\begin{equation}\label{eq:gb-}
g[b_-] = \{z\in Z\mid \exists y\,.\, b(y)=0\ \& \ z\leq g(y)\}
\end{equation} 
is an ideal. 

Assume  $a\circ p \leq b\circ q$, that is, $\Two^p(a)\subseteq \Two^q(b)$. Then $f[a_+]$ is disjoint from $g[b_-]$. For suppose not. Then for some $x$ and $y$, $a(x)=1$, $b(y)=0$, and $f(x)\leq g(y)$. By exactness, there is a $w$ so that $x\leq p(w)$ and $q(w)\leq y$. But then our assumption tells us that $a(x)\leq b(y)$, contradicting  $a(x)=1$ and $b(y)=0$. Hence $f[a_+]$ is disjoint from $g[b_-]$.

Therefore, by the prime ideal theorem, $f[a_+]$ and $g[b_-]$ extend to some $c\in \Two^Z$, so that $z\in f[a_+]$ implies $c(z)=1$ and $z\in g[b_-]$ implies $c(z)=0$. That is, $a\leq \Two^f(c)$ and $\Two^g(c)\leq b$, as required.

\medskip
For $\twobb^-:\Pri\to\DL$,  suppose \eqref{eq:squareTwo2} is in $\Pri$.
Define $f[a_+]$ and $g[b_-]$ as above. Evidently,
$$f[a_+] = \uparrow f(a^{-1}(\{1\}))\quad\quad\text{and}\quad\quad 
g[b_-] = \downarrow g(b^{-1}(\{0\})).$$
Adapting the argument for $\DL$ above, suppose $\Two^p(a)\subseteq \Two^q(b)$. Then $f[b_+]\cap g[b_-]=\emptyset$. Because $a^{-1}(\{1\})$ is closed, it is compact. So $f(a^{-1}(\{1\}))$ is compact, hence closed. The upper set determined by any closed set is  closed. So $f[a_+]$ is an upper compact set. Likewise, $g[b_-]$ is a lower compact set.

Fix $z\in f[a_+]$. For each $z'\in g[b_-]$, $z\nleq z'$. So there is a  clopen downset $U_{z'}$ containing $z'$ and excluding $z$.
These cover $g[b_-]$. So finitely many, say $U_{z'_0}$, ..., $U_{z'_{n-1}}$, suffice to cover. Thus the intersection of their complements is an upper clopen containing $z$ and disjoint from $g[b_-]$. Call it $V_z$. 
The upper clopens $V_z$ cover $f[a_+]$. And again finitely many, say $V_{z_0}$, ..., $V_{z_{m-1}}$, suffice to cover $f[a_+]$. The union of these is an upper clopen that covers $f[a_+]$ and is disjoint from $g[b_-]$.

Let $c$ be the corresponding element of $\Two^Z$. Then $a\leq \Two^f(c)$ and $\Two^g(c)\leq b$.
\end{proof}

\begin{corollary}\label{cor:cocommas-exact-PriDL}
In $\DL$ and $\Pri$  cocomma squares are exact. 
\end{corollary}

\begin{proof}
The homming-into-$\Two$ functors mediating the dual equivalence between  $\DL$ and $\Pri$ are locally monotone and hence $\Pos$-enriched. Therefore cocommas in $\DL$ (or $\Pri$) are commas in $\Pri$ (or $\DL$), which are exact. And exactness is preserved by $\Two^-$.
\end{proof}

\begin{remark}
Duality is  helpful here. Recall from Example~\ref{exle:collage} that in $\Pos$, the exactness of cocommas was immediately obvious from their explicit characterization of cocommas as collages. But we do not have such a characterization for $\DL$s, see also   Example~\ref{exle:total-relation-DL}. 
\end{remark}

We will see in Remark~\ref{rmk:exact-cocommas} that the relationship of Corollary~\ref{cor:cocommas-exact-PriDL} between exactness of cocommas  and preservation of exact squares extends to other concretely order-regular categories. 

Finally, we will need the following result, which is well-known and follows from the fact that the duality respects the factorisation systems of Priestley spaces and distributive lattices. We sketch a direct proof.

\begin{lemma}\label{lem:ontoemb}
The contravariant functors $\twobb^-:\Pri\to\DL$ and $\twobb^-:\DL\to\Pri$ map surjections to embeddings and embeddings to surjections.
\end{lemma}

\begin{proof}
For $\twobb^-:\Pri\to\DL$, we let $f:X\to Y$ so that $\twobb^f: [Y,\twobb] \to [X,\twobb]$. If $f$ is onto, then $\twobb^f$ is an embedding, for the same reason as in $\Pos$. If $f$ is an embedding and $p:X\to\twobb$ is a clopen upset, then by the Priestley separation axiom there is a  clopen upset $q:Y\to\twobb$ containing $\{f(x)\mid x\in p\}$ and disjoint from $\{f(x)\mid x \notin p\}$. Therefore $q\circ f=p$, ie, $\twobb^f(q)=p$, showing that $\twobb^f$ is onto.

For $\twobb^-:\DL\to\Pri$, we let $f:A\to B$ so that $\twobb^f: [B,\twobb] \to [A,\twobb]$. If $f$ is onto, then $\twobb^f$ is an embedding, for the same reason as in $\Pos$. If $f$ is an embedding and $p:A\to\twobb$ is a prime filter, then by the prime filter theorem there is a prime filter $q:B\to\twobb$ containing $\{f(x)\mid x\in p\}$ and disjoint from $\{f(x)\mid x\notin p\}$. Therefore $q\circ f=p$, ie, $\twobb^f(q)=p$, showing that $\twobb^f$ is onto.
\end{proof}

\subsection{Duality of Relations}\label{sec:dualitypriestley}

Before we can state and prove Theorem~\ref{thm:dlpriduality} about the equivalence of $\DL$ and $\Pri$ relations, we need to describe the set-up summarised in \eqref{eq:pridl}.

Given a function, or deterministic program, $f:X\to Y$ there are two natural ways of associating a relation to $f$. The weakening closed relation given by the `hypergraph' $f_\ast=\{(fx,y) \mid fx\le y\}$ and the co-weakening closed relation given by the `hypograph' $f^\ast=\{(y,fx)\mid y\le fx\}$. 

\pskip If $f$ is Scott-continuous then the hypergraph is closed whereas the hypograph does not have a similar good property. This is one reason we choose to work with the hypergraph on the side of spaces. Technically, this means that the relation associated to $f$ will be $f_\ast=\lambda x,y\,.\,Y(fx,y)$.

\pskip Dually, $f$ will be mapped to $\twobb^f:\twobb^Y\to\twobb^X$. We turn this into a relation by stipulating
$$ a\subseteq f^{-1}(b)$$
or, equivalently, $f[a]\subseteq b$, which agrees with \eqref{eq:Twor-span}.
This means that the relation associated to a $g:B\to A$ in $\DL$ is given by $g^*$ which is
$$g^\ast(a,b) =  A(a,gb)$$

\pskip Recalling that extensions of a contravariant functor are contravariant on 2-cells, see \eqref{eq:contravariant-2-cells}, we obtain
\begin{equation}\label{eq:pridl}
\vcenter{
\xymatrix@C=30pt@R=30pt{
\overline\Pri^\co 
\ar@/^/[r]^{\overline{\twobb}}  
& \overline\DL
\ar@/^/[l]^{\overline{\twobb}}  
 \\
\Pri \ar[u]^{(-)_\ast} 
\ar@/^/[r]^{\twobb^-} & 
\DL^\op\ar[u]_{(-)^\ast}
\ar@/^/[l]^{\twobb^-}
}}
\end{equation}
which is in accordance with the left-hand diagram before Remark~\ref{rmk:extendtorelpos}.

\pskip The functor $\overline\Pri^\co\to\overline\DL$ tabulates a relation $r$ as a  span
$$X\,\stackrel{p}{\longleftarrow}\, R \,\stackrel{q}{\longrightarrow}\, Y$$ 
and maps it to the cospan
$$\twobb^X\,\stackrel{\twobb^{p}}{\longrightarrow}\, 
\twobb^R \,\stackrel{\twobb^{q}}{\longleftarrow}\, 
\twobb^Y$$ 
which in turn gives rise to a relation 
\begin{equation}\label{eq:Twor-def-Pri}
\overline\Two(r) = (\twobb^{q})^\ast\cdot(\twobb^{p})_\ast: \twobb^X\looparrowright \twobb^Y.
\end{equation}
This agrees with the definition of $\overline\Two(r)$ as a functor on $\Pos$ in \eqref{eq:Twor-span-def}, but we need to be aware that here $\Two^X$ refers to the set of Priestley-maps from $X$ to the Priestley space $\Two$.

\pskip The functor $\overline\DL\to\overline\Pri^\co$ is defined in the same way on relations. In detail, it tabulates a relation $r$ as a  span
$$A\,\stackrel{p}{\longleftarrow}\, R \,\stackrel{q}{\longrightarrow}\, B$$ 
and maps it to the cospan
$$\twobb^A\,\stackrel{\twobb^{p}}{\longrightarrow}\, 
\twobb^R \,\stackrel{\twobb^{q}}{\longleftarrow}\, 
\twobb^B$$ 
which in turn gives rise to a relation 
\begin{equation}\label{eq:Twor-def-DL}
\overline\Two(r) = (\twobb^{q})^\ast\cdot(\twobb^{p})_\ast: \twobb^A\looparrowright \twobb^B.
\end{equation}
This again agrees with the definition of $\overline\Two(r)$ as a functor on $\Pos$ in \eqref{eq:Twor-span-def}, but now $\Two^A$ refers to the set of distributive lattice morphisms from $A$ to the distributive lattice $\Two$.

\begin{theorem}\label{thm:dlpriduality}
The equivalence 
$$
\xymatrix@C=30pt@R=30pt{
\Pri \ar@/^/[r]^{\twobb^-} & 
\DL^\op
\ar@/^/[l]^{\twobb^-}
}
$$
extends to an equivalence of categories of relations
\begin{equation}\label{eq:Twor-PriDL}
\xymatrix@C=30pt@R=30pt{
\overline\Pri^\co 
\ar@/^/[r]^{\overline{\twobb}}  
& \overline\DL
\ar@/^/[l]^{\overline{\twobb}}  
}
\end{equation}
where $\Two:\overline\Pri^\co\to\overline\DL$ is defined by \eqref{eq:Twor-def-Pri} and 
$\Two:\overline\DL\to\overline\Pri^\co$ is defined by \eqref{eq:Twor-def-DL}.
\end{theorem}

\begin{proof}
To prove that \eqref{eq:Twor-PriDL} is well-defined, use again Remark~\ref{rmk:extendtorelpos} and proceed as in the proof of Theorem~\ref{thm:Twor-functor}. Property 1 (preservation of maps) follows from the fact that if two $\Pos$ (or $\DL$) relations are adjoint in $\Pos$ then they are adjoint in  $\Pri$ (or $\DL$). Property 2 (preservation of exact squares) is Lemma~\ref{lem:DLPri-pres-exact-squares}. Property 3 (mapping surjections to embeddings) is Lemma~\ref{lem:ontoemb}.

It remains to show that \eqref{eq:Twor-PriDL} is an equivalence of categories. Let $r$ be a $\Pri$-relation and $(p,q)=\graphf(r)$. Let $(p',q')$ be the comma of the cospan $(\Two^p,\Two^q)$. We have to show that $\relf(p,q)=\relf(\Two^{p'},\Two^{q'})$. But this follows from $p={\Two^{\Two^p}}$ and $q={\Two^{\Two^q}}$ (due to Priestley duality) and the square
\begin{equation*}
\vcenter{
\xymatrix{
& \ar[dl]_{\Two^{\Two^p}}\ar[dr]^{{\Two^{\Two^q}}} & \\
\ar[dr]_{\Two^{p'}}& & \ar[dl]^{\Two^{q'}} \\
&  &
}
}
\end{equation*}
 being exact. The latter, in turn, is a consequence of $\Two$ preserving exact squares and the comma-square of $(p',q')$ being exact. The other direction, starting with a $\DL$-relation $r$, is proved in the same way.
\end{proof}

The next proposition allows us to compute the dual of a relation by dualising the legs of a representing span even if it is not weakening closed.

\begin{proposition}\label{prop:allspans}
If two (not-necessarily weakening closed) spans in $\DL$ or $\Pri$ represent the same weakening relation, then their dual cospans do so as well.
\end{proposition}

\begin{proof} 
The proof is the same as for Remark~\ref{rmk:allspans} and uses that commas and cocommas are exact (see Corollary~\ref{cor:cocommas-exact-PriDL}) and that duality preserves exactness (see Lemma~\ref{lem:DLPri-pres-exact-squares}).
\end{proof}

\subsection{Examples}\label{sec:exles-pridl}

Recall that in the category $\Pos$, we characterised cocommas as collages. In particular, in the cocomma $(j,k)$ of a span $(A\leftarrow R\rightarrow B)$, the maps $j$ and $k$ are embeddings. Intuitively, this means that the quotient of $A+B$ by $R$ cannot add inequations to $A$ or to $B$. The next example shows that we cannot say the same about cocommas of bounded distributive lattices.

\begin{example}\label{exle:total-relation-DL}
Let $(A\leftarrow R\rightarrow B)$ be a span in $\DL$ with $R=A\times B$ the total relation. Intuitively, the cocomma of the span should be the trivial $\DL$ since $R$ forces the top of $A$ to be below the bottom of $B$. That this is indeed the case is most easily seen using 2-dimensional duality (Section~\ref{sec:stone-duality}) to compute the cocomma of $R$ as the dual of the graph of the dual relation (=the dual of the comma of the dual of the span of $R$). Indeed, the dual of $(A\leftarrow R\rightarrow B)$ is a cospan injecting into the disjoint union of the dual of $A$ and the dual of $B$. It follows from the disjointness that the comma of this cospan is the empty relation.  Its dual is the cospan in $\DL$ that has the one-element $\DL$ as its apex. 
\end{example}

\noindent
The example above depends crucially on working with \emph{bounded} distributive lattices. It will be of interest to look into the duality of not-necessarily-bounded distributive lattices in the future. 

\medskip
We continue with some examples around the Cantor space, which is a Priestley space with a discrete ordering. The Cantor space is homeomorphic to $2^\mathbb N$ with the product topology, homeomorphic to the Stone dual of the free Boolean algebra over the set $\mathbb N$, and homeomorphic to the ``middle-third'' subspace of the unit interval.

\begin{example}[The ordered Cantor space]\label{exle:ordered-Cantor}
Let $X$ be the middle-third Cantor space and $X\leftarrow{\le}\rightarrow X$ the order inherited form the real numbers. According to \eqref{eq:Twor-span}, the dual ${\sqsubseteq}=\overline\Two(\le)$ is given by $a\sqsubseteq b$ iff $a\subseteq b$.
\end{example}

\noindent
The following proposition shows that we can recover the distributive lattice dual to the Priestley space $(X,\le)$ in a natural way from the dual of $\le$. For the definition of an inserter see Remark~\ref{rmk:inserter}. 

\begin{proposition}
Let $(X,\le)$ be Priestley space. Consider $\le$ as a weakening relation between order-discrete Priestley spaces (i.e., Stone Spaces). Then the inserter of the dual of $\le$ is the distributive lattice dual to $(X,\le)$.
\end{proposition}

\begin{proof}
Recall that the distributive lattice dual to $(X,\le)$ is given by the upper clopens of $X$, hence is a sublattice of the Boolean algebra of clopens $\Two^X$ dualising $X$. 
We only need to show that this sublattice arises as the inserter of $j,k:\Two^X\rightrightarrows\Two^{\le}$, where $(j,k)$ is the cospan dual to the span $X\leftarrow{\le}\rightarrow X$.
We use that inserters in distributive lattices are computed as inserters in $\Pos$.
The inserter of $(j,k)$ is the set of clopens $a\in\Two^X$ such that $j(a)\subseteq k(a)$, that is, such that $\{(x,y)\mid x\in a \ \& \ x\le y\}\subseteq  \{(x,y)\mid y\in a \ \& \ x\le y\}$, which is the set of upwards closed clopens. 
\end{proof}
\noindent
The proposition can also be proved more categorically. Since $(X,\le)$ is the quotient (= coinserter) of $X$ by $\le$, the dual of $(X,\le)$ must be the inserter of the dual of $X$ by the dual of $\le$. 

\medskip
We can summarise the previous example and proposition as follows. The reflexive elements (see also Examples~\ref{exle:equivalence-relation} and ~\ref{exle:preorder-quotient}) of the dual of $X$, that is those clopens $a$ for which ${\le}[a]\subseteq a$, form the dual of the Priestley space $(X,\le)$. We next consider what happens if we start from an ordered Stone space that is not a Priestley space, an example due to Stralka.

\begin{example}[The ersatzkette \cite{stralka-erzatzkette}]\label{exle:ersatzkette}
Let $X$ be the middle-third Cantor space and $x\le y$ be the relation that holds whenever $x$ is the left-hand and $y$ the right-hand endpoint of a middle-third gap. The dual ${\sqsubseteq}=\overline\Two(\le)$ is given by $a\sqsubseteq b$ iff $a\subseteq b$ and $b$ strictly extends $a$ on the right. 
\end{example}

The next example is at the heart of a forthcoming paper on extending Stone type dualities from the zero-dimensional to the compact Hausdorff setting.

\begin{example}[The unit interval]\label{exle:unit-interval}
Let $X$ be the ``middle-third'' Cantor space and $R$ the equivalence relation that identifies the endpoints at both sides of a gap. $X$ is a Stone space. The dual of $X$
is the Boolean algebra $A$ of clopens of $X$. The dual ${\prec}=\overline\Two(R)$ satisfies $\overline\Two(R)(a,b)$ if and only if the closure of $a$ is contained in $b$ or, equivalently, if $a$ is way-below $b$. We will develop the general theory at which this example is hinting at in a sequel paper. In a nutshell, the quotient of $X$ by $R$ is homeomorphic to the unit interval, and, at the same time, dual to the `proximity lattice' $(A,\prec)$. This observation can be extended to a duality for compact Hausdorff spaces and proximity lattices \cite{moshier04}.
\end{example}

\noindent
While the unit interval is the coinserter (or, because of discreteness, the coequalizer) of $X$ wrt $R$, the inserter of the dual of $R$ is not dual to the unit interval. The explanation for this mismatch is that in this case the coinserter and the inserter are not computed in dual categories. In forthcoming work we will present a category of algebras in which the inserter of the dual of $R$ is indeed the dual of the unit interval (obviously, the forgetful functor from this category of algebras to $\Pos$  cannot preserve inserters and, hence, cannot preserve all $\Pos$-limits).

\section{Concretely Order-Regular Categories}\label{sec:poscat}\label{sec:CORC}
In Theorem~\ref{thm:dlpriduality} we extended the duality between distributive lattices  and Priestley spaces from maps to relations.

\pskip
This construction from a duality of maps to a duality of relations is purely category theoretic and does not depend on the particularities of distributive lattices and Priestley spaces. All we need are comma objects and a factorisation system in order to compose relations and a duality of maps that respects this structure in a suitable sense. To work out the precise conditions is the purpose of this section. In the next section we can then prove Theorem~\ref{thm:dualityrelations} as a category theoretic generalization of Theorem~\ref{thm:dlpriduality}.

\pskip
The main results of this section are Definition~\ref{convention} and Theorems~\ref{thm:universalproperty} and \ref{thm:universalproperty2} which generalise the approach described in Remark~\ref{rmk:extendtorelpos} to categories over $\Pos$.

\medskip\noindent
The general setting are two forgetful functors to $\Pos$ and two contravariant functors $P,S$ which are adjoint on the right.
In this section we concentrate on axiomatising the properties of $V$ and $U$ and will return to the adjunction in Section~\ref{sec:extending}.
$$
\xymatrix{
\xcal\ar[d]_V\ar@/^/[r]^P& \acal\ar[d]^U\ar@/^/[l]^S\\
\Pos & \Pos
}
$$
So far we took relations as basic, and spans and cospans as devices to represent relations. This can be transferred to concrete $\Pos$-categories, that is, $\Pos$-categories $\ccal$ with a forgetful functor 
$$U:\ccal\to \Pos.$$
In particular,  a relation $A\looparrowright B$  in $\ccal$ will be a relation $UA\looparrowright UB$. In order to make sure that a relation also respects the structure of $\ccal$, we add the requirement that $UA\looparrowright UB$ can be represented by a span in $\ccal$. Equivalently, we can say that a $\ccal$-relation $A\looparrowright B$ is a subobject of $A\times B$ that is upward closed in $A^\op\times B$.

\pskip Since the forgetful functors will be P-faithful, there is at most one relation in $\ccal$ over any $UA\looparrowright UB$ and the order between the relation is 
inherited from $\Pos$.

\pskip
To make sure that relations in $\ccal$ compose as they do in $\Pos$, we ask $\ccal$ to have comma objects and factorisations preserved by $U$. 

\subsection{Concretely Order-Regular Categories}
\label{sec:CORC2}

The following definition details the assumptions sketched above. For notions such as P-faithful, finite limits, $(\Onto, \Emb)$, etc see Section~\ref{sec:weighted}. 
For weakening-closed embedding spans, exact squares, etc see Section~\ref{sec:span-cospan-duality}.
\footnote{
We could follow a more abstract approach in which one defines a calculus of relations via a given set of squares declared to be exact. Then one only needs to require existence of enough exact squares as well as functors preserving them. But this would require the development of a theory that would distract from the duality theory we are interested in here. In all of our examples, relations are weakening relations in $\Pos$ with, possibly, additional properties. And this is what our notion of concretely order-regular captures.
}

\begin{definition}\label{convention}
A $\Pos$-functor $U:\ccal\to\Pos$, or just the category $\ccal$ if $U$ is understood, satisfying the following properties will be called a \textbf{concretely order-regular category}.
\begin{itemize}
\item $U$ is P-faithful, that is, order-preserving and order-reflecting on homsets.
\item $\ccal$ has and $U$ preserves finite limits in the $\Pos$-enriched sense.
\item $\ccal$ has  a factorisation system $(\cal E, \cal M)$ such that $U\cal E=\Onto$ and $U\cal M=\Emb$ and for all $(\Onto,\Emb)$-factorisations $Uf=e\circ m$ there are unique $e'\in\cal E$ and $m'\in\cal M$ such that $Ue'=e$ and $Um'=m$.  
\end{itemize}
\end{definition}

\begin{remark}
The third item can be replaced by the stronger requirement that $\ccal$ has a P-regular/P-mono factorisation system given by $(U^{-1}\Onto,U^{-1}\Emb)$. This would make sure that concretely order-regular categories are order-regular and still include Priestley spaces since the image of a continuous map between Priestley spaces is closed and, therefore, a Priestley space. 
\end{remark}

\begin{remark}
Definition~\ref{convention} allows us to lift terminology from $\Pos$ to $\ccal$. For example,
\begin{itemize}
\item A surjection/embedding in $\ccal$ is an arrow $f$ such that $Uf$ is a surjection/embedding in $\Pos$.
\item 
A span $(p:W\to A,q:W\to B)$ in $\ccal$ is \textbf{weakening-closed} if $(Up,Uq)$ is weakening-closed in $\Pos$. The span $(p,q)$ is an \textbf{embedding-span} if the image of  $\langle p,q\rangle: W\to A\times B$ under $U$ is an embedding in $\Pos$. 
\item A square is \textbf{exact} in $\ccal$ if its image under $U$ is exact in $\Pos$. It follows that $U$ (by definition) preserves exact squares.
\end{itemize}
\end{remark}

\begin{example}
\begin{itemize}
\item All order-regular categories in the sense of \cite[Def.3.18]{ordered-algebras} are concretely order-regular categories under mild conditions, see \cite[Thm.5.13]{ordered-algebras}. 
This includes all quasi-varieties of ordered algebras as well as ordered compact Hausdorff spaces such as Priestley spaces. 
\item
All regular categories are order-regular categories with discrete homsets. This includes the categories of compact Hausdorff spaces or Stone spaces and the category of Boolean algebras. 
\end{itemize}
\end{example}

\begin{definition}[$\ccal$-relation]\label{def:Crelation}
Let $U:\ccal\to\Pos$ be a concretely order-regular category and $A,B\in\ccal$. A $U$-relation, or simply, a $\ccal$-relation, $A\looparrowright B$ is an isomorphism class of weakening closed embedding spans $A\leftarrow \bullet\rightarrow B$,  or equivalently, an upward closed P-mono subobject of $A^\op\times B$. 
\end{definition}

\begin{definition}
Given a concretely order-regular category $U:\ccal\to\Pos$, the extension 
$$\Rel(U): \Rel(\ccal)\to\Rel(\Pos) \textrm{\quad\quad or shorter \quad\quad }
\overline U: \overline\ccal\to\overline\Pos,$$
 is defined as follows.
$\overline\Pos$ is the category $\Rel(\Pos)$ defined in Section~\ref{sec:monotone-relations}.
$\overline\ccal$ has the same objects as $\ccal$ and $\mathcal C$-relations as arrows. The order on relations is inherited from $\Pos$.
\end{definition}
\begin{remark}
Composition in $\overline\ccal$  is associative (and  $\overline\ccal$ is a category)  since composition of weakening-closed embedding spans can be computed in the base category where it is relational composition. 
$\overline U: \overline \ccal \to\overline\Pos$ is a P-faithful functor since the order on arrows in $\overline \ccal$ is inherited from $\overline\Pos$.
\end{remark}

The next definition generalises the corresponding notions from $\Pos$, see Section~\ref{sec:monotone-relations},  to a concretely order-regular category $\ccal$.

\begin{definition}
The functor $$({-})_\ast:\ccal\to\overline\ccal^\co$$ takes a map $f:A\to B$ and maps it to the comma object of the cospan $(f,\id)$. The functor $$({-})^\ast:\ccal^\op\to\overline\ccal$$ takes $f:A\to B$ and maps it to the comma object of $(\id,f)$.
\end{definition}

Given our assumptions on $U$, we have that $f_\ast(a,b)=B(fa,b)$ for $f:A\to B$ and $f^\ast(a,b)=B(b,fa)$. It is worth emphasising that this means that if $f:A\to B$ is a $\ccal$-morphism, then the $\Pos$-relations $f_\ast$ and $f^\ast$ are also $\ccal$-relations.

\subsection{Extending Functors}
The following extension theorems generalise \cite[Thm.4.1]{relations-spans}. We follow the notation of the survey \cite[Thm.3.8]{relations-survey} which is summarised in Remark~\ref{rmk:extendtorelpos}.  It states, informally speaking,   that a functor extends from maps to relations if it preserves exact squares and maps epis to split epi relations. 

\begin{theorem}
\label{thm:universalproperty}
Let $U:\xcal\to\Pos$ be a concretely order-regular category as in   Definition~\ref{convention}.
The locally monotone functor $({-})_\ast:\xcal\to\overline\xcal^\co$
has the following three properties:
\begin{enumerate}
\item \ \ 
$(-)_\ast$ preserves maps, that is, every $f_\ast$ has  a right-adjoint in $\overline\xcal$.
\item \ \ 
$q_\ast\cdot p^\ast =  g^\ast\cdot f_\ast$
\ for all exact squares in $\xcal$
\begin{equation}\label{eq:square-m1}
\vcenter{
\xymatrix{
& UW\ar[dl]_{Up}\ar[dr]^{Uq} & \\
UX\ar[dr]_{Uf} & \leq & UY\ar[dl]^{Ug} \\
& UZ & 
}
}
\end{equation}
\item \ \ 
$e_\ast\cdot e^\ast=\Id$ \ for all surjections $e$ in $\xcal$. 
\end{enumerate}
Moreover, the functor $({-})_\ast$
is universal w.r.t.\ these three properties
in the following sense: if $\kcal$ is any concretely order-regular category to give a locally monotone functor $H:\overline\xcal^\co\to\kcal^\co$ is the same as to give a locally monotone functor $F:\xcal\to\kcal^\co$ with the following three properties:
\begin{enumerate}
\item \ \ 
Every $Ff$ has a right adjoint in $\kcal$, denoted by $(Ff)^r$.
\item \ \ 
$Fq\cdot (Fp)^r = (Fg)^r\cdot Ff$ \ 
for all exact squares as in~\eqref{eq:square-m1}.
\item \ \ 
$Fe\cdot (Fe)^r=\Id$ \ for all epis $e$.
\end{enumerate}
\end{theorem}

\begin{proof}
Since composition of $\xcal$-relations in $\overline\xcal$ is the same as the composition of the underlying relations in $\overline\Pos$, the properties 1-3 of $(-)_\ast$ follow from the corresponding facts on $\Pos$. For the universal property, given $F$, we define $H(f_\ast)=Ff$ and on a general relation $R$ we let 
$$H(R)=H(cR_\ast\cdot dR^\ast)=F(cR)\cdot F(dR)^r.$$ In the case that $R$ is the tabulation of $f_\ast$, we have $H(R)=F(cR)\cdot F(dR)^r=\id^r\cdot Ff=f_\ast$, because the square defining $R$ as the comma-object of the cospan $(f_\ast,\id)$ is exact and because $F$ satisfies property 2. A similar argument shows that $H$ preserves identities. To show that $H$ preserves composition, note that if $R,S$ are relations in $\overline\xcal$, then applying $F$ to the diagram (which abbreviates $R\cdot S$ to $RS$)
\begin{equation}\label{eq:comp}
\vcenter{
\xymatrix@R=20pt{
&
&
\ar `l[llddd] [llddd]_{d{RS}}
\ar `r[rrddd] [rrddd]^{c{RS}}
&
&
\\
&
&
\ar[1,-1]_{dP}
\ar[1,1]^{cP}
\ar@{->>}[-1,0]_{e}
&
&
\\
&
\ar[1,-1]_{dS}
\ar[1,1]^{cS}
&
&
\ar[1,-1]_{dR}
\ar[1,1]^{cR}
&
\\
&
&
&
&
}
}
\end{equation}
we obtain $H(R\cdot S)$ as the relation represented by the outside span and $HR\cdot HS$ as the relation obtained from composing the bottom zig-zag. These two are the same because $F$ satisfies properties 2 and 3. %
To show that $H$ is locally monotone, let $R\subseteq S$ in $\overline\xcal$, that is, there is $f$ in $\xcal$ such that $dR=dS\circ f$ and $cR=cS\circ f$. Then we calculate in $\kcal$
\begin{align*}
H(R) & = F(cR)\cdot F(dR)^r\\
& = F(cS\circ f)\cdot F(dS\circ f)^r\\
& = F(cS)\cdot Ff \cdot Ff^r\cdot F(dS)^r\\
& \le F(cS)\cdot F(dS)^r\\
&= H(S)
\end{align*}
We have shown that $\overline\xcal\to\kcal$ is locally monotone. Hence $\overline\xcal^\co\to\kcal^\co$ is as well.
\end{proof}

There is a dual version of the theorem. Since we need it later, we write it out in detail for reference.

\begin{theorem}
\label{thm:universalproperty2}
Let $U:\acal\to\Pos$ be an concretely order-regular category as in   Definition~\ref{convention}.
The locally monotone functor $({-})^\ast:\acal^\op\to\overline\acal$
has the following three properties:
\begin{enumerate}
\item \ \ 
Every $f^\ast$ has  a left-adjoint $f_\ast$ in $\overline\acal$.
\item \ \ 
$q_\ast\cdot p^\ast =  g^\ast\cdot f_\ast$
\ for all exact squares in $\acal$
\begin{equation}\label{eq:square-m2}
\vcenter{
\xymatrix{
& UW\ar[dl]_{Up}\ar[dr]^{Uq} & \\
UA\ar[dr]_{Uf} & \leq & UB\ar[dl]^{Ug} \\
& UC & 
}
}
\end{equation}
\item \ \ 
$e_\ast\cdot e^\ast=\Id$ \ for all surjections $e$ in $\acal$.
\end{enumerate}
Moreover, the functor $({-})^\ast$ is universal w.r.t.\ these three properties in the following sense: if $\kcal$ is any $\Pos$-category to give a locally monotone functor $H:\overline\acal\to\kcal$ is the same  as to give a locally monotone functor $F:\acal^\op\to\kcal$ with the following three properties:
\begin{enumerate}
\item \ \ 
Every $Ff$ has a left adjoint in $\kcal$, denoted by $(Ff)_l$.
\item \ \ 
$(Fq)_l\cdot Fp = Fg\cdot (Ff)_l$ \ 
for all exact squares as in~\eqref{eq:square-m2}.
\item \ \ 
$(Fe)_l\cdot Fe=\Id$ \ for all epis $e$.
\end{enumerate}
\end{theorem}

\begin{proof} To aid future calculations, we emphasise some of the places where notation changes wrt to the proof of Theorem~\ref{thm:universalproperty}.
Given $F$ and $f:B\to A$, we define $H(f^\ast)=Ff:FA\to FB$ and for a general relation $R$ we let 
$$H(R)=H(cR_\ast\cdot dR^\ast)=F(cR)_l\cdot F(dR).$$
 In the case that $R$ is the tabulation of $f^\ast$, we have $H(R)=F(cR)_l\cdot F(dR)=Ff\cdot \id_l=Ff$. 
The computation showing that $H$ is locally monotone runs as follows. Let $R\subseteq S$ in $\overline\acal$, that is, there is $f$ in $\acal$ such that $dR=dS\circ f$ and $cR=cS\circ f$. Then we calculate in $\kcal$
\begin{align*}
H(R) & = F(cR)_l\cdot F(dR)\\
& = F(cS\circ f)_l\cdot F(dS\circ f)\\
& = F(cS)_l\cdot Ff_l \cdot Ff\cdot F(dS)\\
& \le F(cS)_l\cdot F(dS)\\
&= H(S)
\end{align*}
showing that $H$ is locally monotone.
 \end{proof}
 
\begin{remark}
From the point of view of relations, the two theorems are the same. In both cases, we extend a functor to a relation $R$ by tabulating the relation as $R=cR_\ast\cdot dR^\ast$ and applying the functor to the legs. We spelled them out both for reference in the next section.
\end{remark}

\begin{remark}
In the previous two theorems, if the category on which the functor $F$ is defined has exact cocommas, or enough exact squares, then we can drop the condition 3. Indeed let $(p,q)$ and $(r,s)$ be two composable spans.
Let $(u,v)$ be the comma of $(q,r)$. 
Let $(x,y)$ be the graph of the composition $(p,q);(r,s)$.
To show that the extension to relations of $F$ preserves composition, we need to show that 
$(Fp,Fq);(Fr,Fs)$ and $(Fx,Fy)$
represent the same relation.
Let $(j,k)$ be a cospan completing $(x,y)$ and  $(pu,qv)$ to exact squares. Since 
$F$ preserves exact squares, all of  $(Fp,Fq);(Fr,Fs)$ and $(F(pu),F(qv))$ and $(Fx,Fy)$ represent the same relation.
\end{remark}

\subsection{Examples}\label{sec:exles-ordreg}

In this section we illustrate Definition~\ref{def:Crelation} of $\ccal$-relations by a range of examples. In particular, we will build some new dualities of categories where objects are equipped with additional structure in the form of $\ccal$-relations for various categories $\ccal$. We will instantiate these general constructions with the duality of $\Pri$ and $\DL$-relations of Theorem~\ref{thm:dlpriduality}. We start with an observation about completely distributive lattice relations.

\begin{example}
The functor $\overline\Two:\Rel(\Pos)^\co\to\Rel(\Pos)$ from Theorem~\ref{thm:Twor-functor} defined by $R\mapsto \{(a,b)\mid R[a]\subseteq b\}$ induces order-isomorphisms $\overline\Two_{X,Y}:\Rel(\Pos)(X,Y)^\op\to\Rel(\CDL)(\Two^X,\Two^Y)$ where $\CDL$ is the category of completely distributive lattices.
\end{example}

\medskip
For the remainder of this section, we let $U:\xcal\to\Pos$ and $V:\acal\to\Pos$ be concrete order-regular categories.  
We start by generalising  Example~\ref{exle:equivalence-relation}, noting that a reflexive and transitive relation is a monad in the category of relations and that an interpolative relation below the identity is a comonad. 

\begin{example}\label{exle:monadinrel}
Let $H:\xcal\to\acal$ be a contravariant functor preserving relations\footnote{A functor preserves relations if it preserves exact squares and factorisations. If the categories in question have enough exact squares, preservation of exact squares is enough.} that $H$ extends to $\overline H:\Rel(\xcal)\to\Rel(\acal)^\co$. Since  $\overline H$ is locally monotone, it maps monads (comonads) in $\Rel(\xcal)$ to comonads (monads) in $\Rel(\acal)$. 
\end{example}

\begin{definition}[$\ccal\dashRel$, $\ccal\dashPre$, $\ccal\dashIpl$]
If $\ccal$ is a concretely order-regular category, then we denote by $
\ccal\dashRel$ the category that has pairs $(C,R)$ as objects where $C\in\ccal$  and $R\subseteq C\times C$ is a $\ccal$-relation and arrows $f:
C,R)\to(C',R')$ are functions $f:C\to C'$ such that $xRy\Rightarrow f(x)R'f(y)$ for all $x,y$ in the underlying poset of $C$. $\ccal\dashPre$ and $\ccal\dashIpl$ are the full subcategories of $\ccal\dashRel$ of preorders and interpolative relations, respectively, see Example~\ref{exle:monadinrel}.
\end{definition}

\begin{example}\footnote{$\Stone\dashPos$ also deserves attention.}
$\Stone\dashRel$ and $\BA\dashRel$ as well as $\Stone\dashPre$ and $\BA\dashIpl$ are dually equivalent.
\end{example}

\medskip
To keep the exposition easy, we now specialise to the example above. But Theorems~\ref{thm:StoneRel} and \ref{thm:orderedStone} below transfer to dual equivalences $(F,G)$ that rely on dualising objects other than $\Two$, see Remark~\ref{rmk:many-valued-valuations}.

\begin{theorem}\label{thm:StoneRel}
Under the standing assumptions of this subsection, the category $\xcal\dashRel$ is dually equivalent to the category $\acal\dashRel$.
\end{theorem}

\begin{proof}
Exploiting Theorem~\ref{thm:dlpriduality}, we only have to show that the dualising functor $\Two^-:\Stone\dashRel\to\BA\dashRel$, and its converse $\Two^-:\BA\dashRel\to\Stone\dashRel$,  preserve homomorphisms. To this end, going back to  \eqref{eq:Twor-span-pic}, we consider 
\begin{equation}\label{eq:Tworspan2}
\vcenter{
\xymatrix{
& R\ar[dl]_{p}\ar[dr]^{q} & 
&& 
& R'\ar[dl]_{p'}\ar[dr]^{q'} &\\
X\ar[dr]_{} \ar@{..>}@/_1pc/[rrrr]_{\quad f} & \le & X\ar[dl]^{}\ar@{..>}@/_1pc/[rrrr]_{f\quad \ } 
&& 
X'\ar[dr]_{a'} & \le & X'\ar[dl]^{b'}\\
& \twobb & 
&& 
& \twobb &
}}
\end{equation}
Assuming $xRy \Rightarrow f(x)R'f(y)$, we have to show that $(a',b')\in\overline\Two(R')\Rightarrow (\Two^f(a'),\Two^f(b'))\in\overline\Two(R)$. In other words, we have to show that 
if $\forall w\in R\,.\, \exists w'\in R'\,.\,fp(w))\le p'(w') \ \& \ q'(w')\le fq(w)))$
and if $a'p'\le b'q'$ 
then $a'fp(w)\le b'fq(w)$ for all $w\in R$.  This is straightforward.
\end{proof}

Using Example~\ref{exle:equivalence-relation}, or Example~\ref{exle:monadinrel},  we can specialise this to preorders.

\begin{theorem}\label{thm:orderedStone}
The category of preordered Stone spaces is dually equivalent to the full subcategory of $\BA\dashRel$ in which objects $(A,R)$ where $R$ is interpolative (= idempotent below the identity).
\end{theorem}

Theorems~\ref{thm:StoneRel} and \ref{thm:orderedStone}, transfer, mutatis mutandis, to other dual equivalences than $\Stone$ and $\BA$ including those that rely on other dualising objects.

\section{Extending Equivalences and Adjunctions}\label{sec:extending}

\pskip
We are interested in extending contravariant adjunctions and equivalences of $\Pos$-categories from maps to relations. In the case of adjunctions, for Theorem~\ref{thm:adjunctionrelations}, we need to appeal to the framed bicategories of Shulman \cite{shulman}. We therefore treat the easier case of equivalences first. Theorem~\ref{thm:dualityrelations} is a direct generalization of Theorem~\ref{thm:dlpriduality} and we recommend to read Section~\ref{sec:dualitypriestley} before reading this one.

\subsection{Extending Equivalences to Categories of Relations}\label{sec:lifting-spans}
Let $U:\acal\to\Pos$ and $V:\xcal\to\Pos$ be two concretely order-regular categories, see Definition~\ref{convention}. 

\pskip 
Given a dual equivalence  $F:\xcal\to\acal$ and $G:\acal\to\xcal$, we will extend it to $\overline\xcal=\Rel(\xcal)$ and $\overline\acal=\Rel(\acal)$ in Theorem~\ref{thm:dualityrelations}.
The plan is to apply Theorems~\ref{thm:universalproperty} and~\ref{thm:universalproperty2} to the situation
$$
\xymatrix@C=30pt@R=30pt{
\overline\xcal^\co 
\ar@/^/[r]^{\overline{F}}  
& \overline\acal
\ar@/^/[l]^{\overline{G}}  
 \\
\xcal \ar[u]^{(-)_\ast} 
\ar@/^/[r]^{F} & 
\acal^\op\ar[u]_{(-)^\ast}
\ar@/^/[l]^{G}
}
$$
To obtain $\overline F$ from Theorem~\ref{thm:universalproperty}, we define the functor $$\xcal\to\overline{\acal}$$
as mapping arrows $(f:X\to Y)$ to relations $Ff^\ast: FX\looparrowright FY$. That is, we have  $(a,b)\in Ff^\ast$ iff $a\le Ff(b)$.

\pskip
Note that $Ff^\ast$ has a left-adjoint in $\overline\acal$ and hence a right adjoint $(Ff)^r=Ff_\ast$ in $\overline\acal^{\co}$ as required by Theorem~\ref{thm:universalproperty}.

\pskip
For the condition that $(F-)^\ast$ preserves exact squares, given an exact square in $\xcal$ 
\begin{equation}\label{eq:square-m3}
\vcenter{
\xymatrix{
& W\ar[dl]_{p}\ar[dr]^{q} & \\
X\ar[dr]_{f} & \leq & Y\ar[dl]^{g} \\
& C & 
}
}
\end{equation}
we need $Fq\cdot (Fp)^r = (Fg)^r\cdot Ff$ in $\overline\acal^{\co}$, 
which is in $\overline\acal$
\begin{equation}\label{eq:ext-contravariant-square1}
Fq^\ast\cdot Fp_\ast = Fg_\ast \cdot Ff^\ast
\end{equation}
as in Lemma~\ref{lem:DLPri-pres-exact-squares} for the case of Priestley spaces and distributive lattices.

\medskip\noindent
We also need that for all epis $e$ in $\xcal$ we have $Fe\cdot (Fe)^r=\Id$ in $\overline\acal^{\co}$,
which is in $\overline\acal$
\begin{equation}\label{eq:ext-contravariant-epi-F}
Fe^\ast\cdot Fe_\ast =\Id,
\end{equation}
which holds iff $F$ maps surjections to embeddings, as in  Lemma~\ref{lem:ontoemb} for the case of distributive lattices and Priestley spaces.

\pskip

\pskip
Following exactly the same line of reasoning as for $\overline F$ above, to obtain $\overline G:\overline\acal\to\overline\xcal^\co$ from Theorem~\ref{thm:universalproperty2}, we let the functor $$\acal^\op\to\overline{\xcal}^\co$$
be given by mapping arrows $g:A\to B$ in $\acal$ to relations $Gg_\ast: GB\looparrowright GA$. That is, we have  $(y,x)\in Gg_\ast$ iff $Gg(y)\le x$.
Note that $Gg_\ast$ has a left adjoint $$(Gg)_l=Gg^\ast$$ in $\overline\xcal^\co$, as required by Theorem~\ref{thm:universalproperty2}.
In order to verify that $G$ satisfies the assumptions of Theorem~\ref{thm:universalproperty2}, given an exact square 
\begin{equation}\label{eq:square-m4}
\vcenter{
\xymatrix{
& W\ar[dl]_{p}\ar[dr]^{q} & \\
A\ar[dr]_{f} & \leq & B\ar[dl]^{g} \quad , \\
& C & 
}
}
\end{equation}
we need to check that $(Gq)_l\cdot Gp = Gg\cdot (Gf)_l$ in $\overline\xcal^{\co}$,
which is in $\overline\xcal$
\begin{equation}\label{eq:ext-contravariant-square2}
Gq^\ast\cdot Gp_\ast = Gg_\ast \cdot Gf^\ast\ .
\end{equation}
We also need to check that for all epis $e$ in $\acal$ we have in $\overline\xcal^{\co}$
$$(Ge)_l\cdot Ge=\Id\ ,$$ 
which is in $\overline\xcal$
\begin{equation}\label{eq:ext-contravariant-epi-G}
Ge^\ast\cdot Ge_\ast =\Id \ ,
\end{equation}
which holds iff $G$ maps surjections to embeddings.

\medskip\noindent
To summarize, we have the following corollaries of Theorems~\ref{thm:universalproperty} and \ref{thm:universalproperty2} about the situation depicted in 
$$
\xymatrix@C=30pt@R=30pt{
\overline\xcal^\co 
\ar@/^/[r]^{\overline{F}}  
& \overline\acal
\ar@/^/[l]^{\overline{G}}  
 \\
\xcal \ar[u]^{(-)_\ast} 
\ar@/^/[r]^{F} & 
\acal^\op\ar[u]_{(-)^\ast}
\ar@/^/[l]^{G}
}
$$

\begin{proposition}\label{prop:contravariant1}
Let $\xcal$ and $\acal$ be concretely order-regular categories (Definition~\ref{convention}).
If a contravariant functor $F:\xcal\to\acal$ preserves exact squares in the sense that  
$Fq^\ast\cdot Fp_\ast = Fg_\ast \cdot Ff^\ast$
for all exact squares as in \eqref{eq:square-m3} and if $F$ takes surjections to embeddings,
then $F$ extends uniquely to a (covariant) functor $\overline\xcal^\co\to\overline\acal$. A relation $r:X\looparrowright Y$ is mapped to $Fr:FX\looparrowright FY$ given by
$$(a,b)\in \overline Fr \ \Leftrightarrow \ Fp(a)\le_{FW} Fq(b)$$
where $(p:W\to X,q:W\to Y)$ is a tabulation of $r$. $Fr$ is tabulated by the comma object of the cospan $(Fp,Fq)$. In case that the relation is a map, that is, in case that $r=f_\ast$ for some $f:X\to Y$ this simplifies to 
$$(a,b)\in \overline F(f_\ast) \ \Leftrightarrow \ a\le_{FX} Ff(b).$$
\end{proposition}

\begin{proof}
We know from Theorem~\ref{thm:universalproperty} (with $H$ being $\overline F$ and $Ff$ being $Ff^\ast$) that $\overline F(r)=Fq^\ast \cdot Fp_\ast$, that is, $\overline Fr(a,b)={FW}(Fp(a), Fq(b))$. In case $r=f_\ast$, since the square defining $(p,q)$ is exact, we have $\overline F(f_\ast)(a,b)=Ff^\ast(a,b)=FX(a,Ff(b))$.
\end{proof}

\begin{remark}\label{rmk:Hoare1}
If $\xcal$ is a category of spaces and $Ff=\twobb^f=f^{-1}$, then $r:X\looparrowright Y$ is mapped to $\overline Fr:FX\looparrowright FY$ such that, see Proposition~\ref{prop:extendviaspan},
$$(a,b)\in \overline Fr \ \ \Longleftrightarrow \ \ (\,x\in a \ \& \ xry \ \Rightarrow \ y\in b \,)$$ which we may write in Hoare-triple notation as
$$ \{a\}r\{b\}.$$
In case $r:X\looparrowright Y$ is a map $f:X\to Y$, that is, if $r=f_\ast$, which is $r(x,y)=Y(fx,y)$, then this can be written as
$$a\subseteq f^{-1} b.$$
\end{remark}

\noindent
The next result is analogous to Proposition~\ref{prop:contravariant1}, but worth spelling out for future reference.

\begin{proposition}\label{prop:contravariant2}
Let $\xcal$ and $\acal$ be concretely order-regular categories (Definition~\ref{convention}).
If a contravariant functor  $G:\acal\to\xcal$ preserves exact squares in the sense that
%
$
Gq^\ast\cdot Gp_\ast = Gg_\ast \cdot Gf^\ast
$
%
for all exact squares as in \eqref{eq:square-m4} and if $G$ takes surjections to embeddings,
then $G$ extends uniquely to a (covariant) functor $\overline\acal\to\overline\xcal^\co$. A relation $r:A\looparrowright B$ is mapped to $Gr:GA\looparrowright GB$ given by
$$(x,y)\in \overline Gr \ \Leftrightarrow \ Gp(x)\le_{GW} Gq(y)$$
where $(p:W\to A,q:W\to B)$ is a tabulation of $r$. $Gr$ is tabulated by the comma object of the cospan $(Gq,Gp)$. In case that the relation is a map, that is, in case that $r=g^\ast$ for some $g:B\to A$ this simplifies to 
$$(x,y)\in \overline G(g^\ast) \ \Leftrightarrow \ x\le_{GA} Gg(y)$$
\end{proposition}

\begin{proof}
We know from Theorem~\ref{thm:universalproperty2} (with $H$ being $\overline G$ and $Fg$ being $Gg_\ast$) that $\overline G(r)=Gq^\ast \cdot Gp_\ast$, that is, $\overline Gr(x,y)={GW}(Gp(x), Gq(y))$. In case $r=g^\ast$, because the square defining $(p,q)$ being exact, we have $\overline G(g_\ast)(x,y)=Gg_\ast(x,y)=GB(Gg(x),y)$.
\end{proof}

\begin{remark}\label{rmk:Hoare2}
If $\acal=\DL$  and $Gg=\twobb^g$, then a relation ${\vdash}:A\looparrowright B$ is mapped to $\overline G({\vdash}):GA\looparrowright GB$ such that for prime filters $x,y$
$$(x,y)\in \overline G({\vdash}) \ \ \Longleftrightarrow \ \ (\,a\in x \ \& \ a\vdash b \ \Rightarrow \ b\in y \,).$$ 
In the words of Remark~\ref{rmk:Hoare1}, $G(\vdash)$ is the largest relation $r$ making the Hoare triple $\{a\}r\{b\}$ true.

In case $\,{\vdash}:A\looparrowright B$ is a map $g:B\to A$, that is, if $\,{\vdash}=g^\ast$, which means $(a\vdash b) \Leftrightarrow (a\le_A g(b))$, then this can be written as $ x\le_{GA} Gg(y)$ which translates as a statement about prime filters into
$$x\subseteq g^{-1}(y)$$
\end{remark}

Before proving that dual equivalences extend from maps to relations, we need to check that the following holds.

\begin{lemma}
Let dually equivalent $F:\xcal\to\acal$, $G:\acal\to\xcal$ satisfy the assumptions of Propositions~\ref{prop:contravariant1} and \ref{prop:contravariant2}. Let $(p:W\to X,q:W\to Y)$ be a span representing the relation $r:X\looparrowright Y$. Then $\overline G\,\overline Fr$ is represented by $(GFp,GFq)$.
\end{lemma}

\begin{proof}
$\overline F r$ is represented by $(Fq,Fp)$. Let $(p',q')$ be its comma object. Because $((p',q'),(Fq,Fp))$ is an exact square,  and $G$ preserves exact squares, we know that $(Gq',Gp')$ and $(GFp,GFq)$ represent the same relation. 
\end{proof}

Combining Propositions~\ref{prop:contravariant1} and \ref{prop:contravariant2}
we obtain the following extension theorem.

\begin{theorem}\label{thm:dualityrelations}
Let $\xcal$ and $\acal$ be concretely order-regular categories (Definition~\ref{convention}). Let $F:\xcal\to\acal$ and $G:\acal\to\xcal$ be a dual equivalence of contravariant functors satisfying the assumptions of Propositions~\ref{prop:contravariant1} and \ref{prop:contravariant2}, namely preservation of exact squares and the mapping of surjections to embeddings.
Then $F$ and $G$ extend to an equivalence $\overline F:\overline\xcal^\co\to\overline\acal$ and $\overline G:\overline\acal\to\overline\xcal^\co$. Restricting this equivalence to maps as in
$$
\xymatrix@C=30pt@R=30pt{
\overline\xcal^\co 
\ar@/^/[r]^{\overline{F}}  
& \overline\acal
\ar@/^/[l]^{\overline{G}}  
 \\
\xcal \ar[u]^{(-)_\ast} 
\ar@/^/[r]^{F} & 
\acal^\op\ar[u]_{(-)^\ast}
\ar@/^/[l]^{G}
}
$$
gives back the dual equivalence $(F,G)$.
\end{theorem}

\begin{proof}
We have to show that the unit and counit are natural wrt relations. Using the previous lemma, it is enough to consider diagrams such as
\begin{equation}
\label{eq:naturality-isos}
\vcenter{
\begin{xymatrix}{
A \ar[d] \ar@{..>}@/^3ex/[rr]^R& \ar[l] W\ar[d] \ar[r] & B \ar[d]\\
GFA \ar@{..>}@/_3ex/[rr]_{\overline G\overline FR}& \ar[l] GFW \ar[r] & GFB
}
\end{xymatrix}
}
\end{equation}
where the upper span tabulates a relation $R$ and the vertical arrows are the unit of $F\dashv G$. The inner squares commute by naturality wrt maps, which implies that the outer rectangle (with the dotted horizontal arrows) commutes since the vertical arrows are isos. 
\end{proof}

\begin{remark}
The previous proof relies on the units being isos. This is where we cannot weaken from dual equivalence to dual adjunctions. We will see how to deal with this with the help of double categories in the next section.
\end{remark}

\begin{remark}\label{rmk:exact-cocommas}
Let $F:\xcal\to\acal$ and $G:\acal\to\xcal$ be a dual equivalence of concretely order-regular categories. Then $F,G$ preserve exact squares if and only if cocommas in $\xcal$ and $\acal$ are exact. 
For ``only if'', note that the dual of a cocomma square in $\acal$ is exact (due to being a comma square). It then follows from the functor preserving exactness that the cocomma square itself must be exact as well.
For ``if'', consider an exact square $(p,q,j,k)$ on one side with span $(p,q)$ and cospan $(j,k)$. Let $(j',k')$ be the cocomma of $(p,q)$ and $(p',q')$ be the comma of $(j',k')$. The squares $(p,q,j',k')$ and  $(p',q',j',k')$ are, respectively, cocomma and comma squares by definition. Then $(p,q,j,k)$ is also a comma square. The dual squares are then also comma and cocomma squares, respectively. Since comma and cocomma squares are exact, so is the dual of $(p,q,j,k)$. 
\end{remark}

\subsection{Extending Adjunctions to Double Categories of Relations}\label{sec:adjunction}
In this section, we are going to extend adjunctions 
$$
\xymatrix@C=30pt@R=30pt{
\xcal
\ar@/^/[r]^{F} & 
\acal^\op
\ar@/^/[l]^{G}
}
$$
to the corresponding categories of relations. As we noted above,  the commutativity of  Diagram~\eqref{eq:naturality-isos} in $\Rel(\xcal)$ depends on the unit of the adjunction being an isomorphism. Accordingly, in general, adjunctions on categories of maps do not extend to adjunctions on categories of relations. This problem can be solved by amalgamating the category of relations and the category of maps into a so-called weak double category \cite{GP99}. As shown in Grandis and Par\'e \cite{GP04} this makes it possible to extend adjunctions to relations and various other structures such as spans/cospans and distributors. An excellent account  can be found in the recent book by Grandis \cite{Grandis}. Framed bicategories are special weak double categories and we will rely on Shulman's \cite{shulman} in the following. 

\paragraph{Framed Bicategories.}\label{sec:fb}

Framed bicategories \cite{shulman}  allow us to have both $\ccal$ and $\Rel(\ccal)$ in one structure, see also Example 2.6 in \cite{shulman}. 
Since we only need a very special case of framed bicategories in this paper, we do not detail the general definition and only explain how any concretely order-regular category $\ccal$ gives rise to a framed bicategory $\S\ccal$.

\pskip 
Framed bicategories are special double categories \cite{GP04}. Informally speaking, a 2-cell in a double category is a square
$$
\xymatrix{
A\ar[d]_f \ar[r]^R \ar@{}[dr]|-\subseteq& B\ar[d]^g \\
C \ar[r]_S & D  
}
$$
where, in our examples, the horizontal arrows are relations, the vertical arrows are maps, and the 2-cell represents a subset-relation as indicated.%
\footnote{Grandis and Par\'e write relations as vertical arrows. }
In other words, forgetting the horizontal structure of $\S\ccal$ gives back $\ccal$ and forgetting the vertical structure of $\S\ccal$, we obtain $\Rel(\ccal)$.
Importantly, it is the double category theoretic view which gives us the right notion of functor and adjunction. 
The technical point where this matters can be seen if we go back to \eqref{eq:naturality-isos} and note that the unit of an adjunction $F\dashv G$ is not, in general, natural wrt relations. 
From a double category theoretic point of view, it suffices that the outer rectangle of \eqref{eq:naturality-isos} commutes up to a 2-cell.

\pskip
More technically, we can define, ignoring issues of size, a double category as an internal category \cite[Ch.8]{borceux2}
$$\xymatrix{\mathbb D_1\ar@<0.5ex>[r]^{\mathit{dom}}\ar@<-.5ex>[r]_{\mathit{cod}}&\mathbb D_0}$$
in $\Cat$.%
\footnote{Following \cite{GP04,shulman}, we should say \emph{weak} internal category, but in our example of relations only the special strict case occurs.}
Note that this point of view breaks the symmetry between `vertical' arrows, which are arrows in $\mathbb D_0$, and `horizontal' arrows, which are objects in $\mathbb D_1$. The internal composition in the $\mathbb D_i$ is vertical composition and the external composition of $\mathbb D_1$ is horizontal composition. Finally, a double category is a framed bicategory if every vertical arrow can be represented by horizontal arrows in a suitable way, see \cite[Sec.1.2 and 1.3]{GP04} and  \cite[Thm.4.1 and Thm.A.2]{shulman} for details. This gives a double category theoretic axiomatisation of the two ways \eqref{eq:f_ast} and \eqref{eq:f^ast} of embedding maps into relations.

\pskip For our purposes, it suffices to know that the construction described in the next proposition is a framed bicategory. This then allows us to use that framed bicategories form a strict 2-category and, therefore, come with  a native notion of adjunction. As it turns out, this notion of adjunction is precisely the one we need in Theorem~\ref{thm:adjunctionrelations} to prove that adjunctions extend from maps to relations.

\pskip We will write $\S\ccal$ for the framed bicategory of relations of the category $\ccal$. A framed bicategory is a double category (strict for us) with the additional property that for every horizontal  1-cell $R$ and every pair $(f,g)$ of vertical 1-cells as in 
\begin{equation*}
\vcenter{
\xymatrix@C=40pt{
A\ar@{..>}[r]^{R(f,g)} \ar[d]_f& B\ar[d]^g \\
C\ar@{->}[r]_R  & D
}}
\end{equation*}
there is a unique cartesian lifting of $R$ along $(f,g)$. In our special case, 
the cartesian lifting (also known as the restriction) of $R:C^\op\times D\to\Two$ along $(f,g)$ will be the relation $R(f,g)$, defined by mapping $(a,b)$ to $R(f(a),g(b))$.

\begin{remark}
For the reader who wants to understand in detail how framed adjunctions apply to our setting, we give a brief guide to the notation of \cite{shulman}. $A,B$ are objects and $f,g$ are vertical 1-cells (maps) and $M,N$ are horizontal 1-cells (relations). We will write $1$ for identity arrows dropping the usual subscript of $1_A:A\to A$ so that $A(1,1)$ is the identity relation on $A$.
In \cite[Def.1]{shulman}, the horizontal 1-cell $U_A$ is $A(1,1)$, the 2-cell $U_f$   records the fact that $A(1,1)\le B(f,f)$, that is, that $f:A\to B$ is monotone. Our notation for the horizontal composition $M \odot N$ is $M;N$ or $N\cdot M$.  The restriction $f^\ast M g^\ast$, that is, the cartesian lifting of $M$ along $(f,g)$,  is $M(f,g)$, or, equivalently, $g^\ast\cdot M \cdot f_\ast$.\footnote{Shulman uses $(-)^\ast$ to denote a cartesian lifting while we use $(-)^\ast$ for the embedding $\Pos\to\Rel(\Pos)$.} The extension $f_! Mg_!$, that is, the op-cartesian lifting of $M$  along $(f,g)$,  is $g_\ast\cdot M\cdot f^\ast$. The base change object ${}_fB$ is $B(f,1)=f_\ast$ and $B_f$ is $B(1,f)=f^\ast$.
\end{remark}

\begin{proposition}
Let $\ccal$ be a concretely order-regular category. Then there is a framed bicategory $\S\ccal$ that has the same objects as $\ccal$, that has the arrows of $\ccal$ as vertical arrows,  and that has the arrows of $\Rel(\ccal)$ as the horizontal arrows. The 2-cells are squares 
\begin{equation}
\vcenter{
\xymatrix{
A\ar@{->}[r]^S \ar[d]_f& B\ar[d]^g \\
C\ar@{->}[r]_R  & D
}}
\end{equation}
such that $S\le R(f,g)$, or, equivalently, any of $g_\ast \cdot S\subseteq R\cdot f_\ast$ or $g_\ast \cdot S\cdot f^\ast\subseteq R$ or $S\subseteq g^\ast\cdot R\cdot f_\ast$.
\end{proposition}

\begin{proof}
With the notation of the remark above, it is immediate to verify condition (iii) of \cite[Thm.4.1]{shulman}.
\end{proof}

\noindent
We will write 
$$\S\ccal^\co \quad\quad\textrm{and} \quad\quad \S\ccal^\op$$ for the framed bicategories that are the same as $\S\ccal$ but have, respectively,  reversed 2-cells and reversed vertical 1-cells.

\paragraph{Extension Theorems.} In the following proposition we assume that we have an adjunction $F\dashv G:\acal^\op\to\xcal$ with $F$ and $G$ satisfying the assumptions that allow us to apply  Propositions~\ref{prop:contravariant1} and \ref{prop:contravariant2} in order to obtain extensions $\S F$ and $\S G$.

\begin{theorem}\label{thm:adjunctionrelations}
Let $\xcal$ and $\acal$ be concretely order-regular categories and let $F\dashv G:\acal^\op\to\xcal$ be an adjunction with both $F$ and $G$ preserving exact squares and mapping surjections to embeddings. Define the extensions $\S F$ and $\S G$ on 0- and 1-cells as $F$ and $G$ and on 2-cells by tabulation as in Proposition~\ref{prop:contravariant1} for $\S F$ and as in Proposition~\ref{prop:contravariant2}  for $\S G$. Then these extensions
$$
\xymatrix@C=60pt@R=30pt{
\S\xcal^\co 
\ar@/^/[r]^{\S{F}}  
& \S\acal
\ar@/^/[l]^{\S{G}}  
}
$$
constitute an adjunction of framed bicategories. Moreover, if $F$ and $G$ are an equivalence, so are $\S F$ and $\S G$.
\end{theorem}

\begin{proof}
First, we have to check that $\S F$ and $\S G$ are framed functors. Defining them on objects and vertical 1-cells as $F$ and $G$ and on horizontal 1-cells as $\overline F$ and $\overline G$ as in Propositions~\ref{prop:contravariant1} and \ref{prop:contravariant2}, $\S F$ and $\S G$ are strong framed functors in the sense of \cite[Defs.6.1,6.14]{shulman}. It remains to see that the units $\eta:\Id\to GF$ and $\eps:\Id\to FG$ of the adjunction extend to framed transformations  \cite[Defs.6.15,6.16]{shulman}. Since our 2-cells are posetal, all 2-cell diagrams between the same 1-cells commute. So it suffices to show that for all relations $R:A\to B$ we have  a 2-cell $(\eta_A,\eta_{B}):R\Rightarrow GFR$, that is, $(\eta_{B})_\ast\cdot R\subseteq GFR\cdot (\eta_A)\ast$ and this follows from the two squares in \eqref{eq:naturality-isos} commuting.
\end{proof}

As a corollary we obtain a result in the same spirit as the equivalence Theorem~\ref{thm:dualityrelations}. But, technically, they are different theorems, because Theorem~\ref{thm:dualityrelations} is about categories where relations are arrows, whereas Corollary~\ref{cor:adjunctionrelations} is about framed bicategories where relations are objects parameterised by maps.

\begin{corollary}\label{cor:adjunctionrelations}
Let $\xcal$ and $\acal$ be concretely order-regular categories and let $F\dashv G:\acal^\op\to\xcal$ be a dual equivalence with both $F$ and $G$ preserving exact squares and mapping surjections to embeddings. Then there is an equivalence between the framed bicategories $\S\xcal^\co$ and $\S\acal^\op$, determined by the action of $F$ and $G$ on vertical arrows.
\end{corollary}

The next theorem shows that the adjunction `homming into $\Two$' extends to relations. It only works for the framed bicategory $\S\Pos$ and has no analogue in terms of $\Rel(\Pos)$.

\begin{corollary}
$\twobb^-\dashv\twobb^-:\Posframed^\op\to\Posframed$ is a framed (and op-framed) adjunction.
\end{corollary}

\subsection{Examples}
\label{sec:exles-adjunctions}

We exhibit two further dualities that satisfy the assumptions of Theorems  \ref{thm:dualityrelations} and \ref{thm:adjunctionrelations}. We start with some remarks on semi-lattices in ordered categories.

$\Pos$-algebras (or $\Pri$-algebras) have monotonic (or monotonic continuous) operations. But if the operations themselves determine a partial order, for example, if one of the operations is associative and idempotent, the underlying partial order does not have to coincide with the algebraically determined order. 
For example, it is possible to have a lattice in $\Pos$ for which the lattice order is not the same as the underlying order (take a lattice with discrete underlying set). So care needs to be taken in specifying how the poset order and derived order relate.

In $\Set$, a semilattice is only conventionally spoken of being a meet or join semilattice depending on intuition. In $\Pos$, a semilattice may actually be a meet or join semilattice (or neither) according to whether the underlying poset order coincides with the order defined by the lattice operation, or its opposite. Thus we call a unital semilattice $(X,*,e)$ in $\Pos$
a \emph{unital meet semilattice} if $x\leq y$ coincides with $x=x*y$, and $e$ is the maximal element. Likewise we call it a \emph{unital join semilattice} if $x\leq y$ coincides with $x*y=y$, and $e$ is the minimal element.

To show that the conditions of Theorems \ref{thm:dualityrelations} and \ref{thm:adjunctionrelations} are satisfied for a particular natural duality, the key step is to verify that the functors mediating the duality preserve exact squares. 

We use the notation of the proof of Lemma~\ref{lem:DLPri-pres-exact-squares}. 

\subsubsection{Hofmann-Mislove-Stralka Duality}

Hofmann-Mislove-Stralka duality \cite{hms} establishes that the duals of unital meet semilattices in $\Pos$ are unital join semilattices in $\Pri$, where the semilattice order and Priestley order coincide.
So take a \emph{Hofmann-Mislove-Stralka space}, or HMS space, to be a unital join semilattice in $\Pri$.

Suppose we have an exact square \eqref{eq:square-xy} in meet semilattices. We must show that the dual square is exact in HMS spaces.

Since $\Two^-$ preserves order on morphisms, $\Two^p\circ \Two^f \leq \Two^q\circ \Two^g$.
Consider some $a\in \Two^X$, $b\in \Two^Y$ so that $\Two^p(a)\leq \Two^q(b)$.
Then $f[a_+]$, see \eqref{eq:fa+}, is a filter. So it corresponds to an element $c$ of $\Two^Z$, which by construction satisfies $a\leq \Two^f(c)$. Also if $\Two^g(c)(y) = 1$, then there is some $x$ so that $a(x)=1$ and $f(x)\leq g(y)$. So by exactness of the given square, pick $w$ so that $x\leq p(w)$ and $q(w)\leq y$.
Hence $1=a(x) \leq a(p(w)) \leq b(q(w)) \leq b(y)$.
We have shown that $a\leq \Two^f(c)$ and $\Two^g(c)\leq b$, that is, the dual square in HMS is exact.

In the other direction, suppose we have an exact square  \eqref{eq:square-xy} in HMS spaces.
Again $\Two^{-}$ preserves order on morphisms, so $\Two^p\circ \Two^f \leq \Two^q\circ \Two^g$.

In an HMS space a closed ideal is principal. This follows from the following observations.
As Priestley spaces, HMS spaces are bitopologically spectral spaces. That is, 
(i) the upper opens constitute a spectral topology, as do the lower opens, 
(ii) the Priestley order is the specialization order for the upper open topology, and is the converse of the specialization order for the lower topology, and (iii)
the Priestley topology is the join of these two spectral topologies. In particular, the upper open topology is sober. So specialization is a dcpo. 
Suppose $I$ is a closed ideal. Since it is a downset, it is closed in the upper open set topology. Suppose $C\cup D\subseteq I$ for two closed sets $C$ and $D$. If $x\in I\setminus C$ and $y\in I\setminus D$, then $x\wedge y\in I\setminus(C\cup D)$. So $I$ is an irreducible closed, and must be the closure (in the upper open set topology) of a point.

Suppose $\Two^p(a)\leq \Two^q(b)$. Then $f[a_+]$ and $g[b_-]$, see \eqref{eq:fa+} and \eqref{eq:gb-}, must be disjoint. For suppose not. Then for some $x$ and $y$, $a(x)=1$, $f(x)\leq g(y)$, and $b(y)=0$. By exactness, there is a $w$ so that $x\leq p(w)$ and $q(w)\leq y$. But then $a(x)\leq b(y)$, contradicting $a(x)=1$ and $b(y)=0$.

Since $f$ is continuous, and $a_+$ is clopen, $f[a_+]$ is compact. And since $g[b_-]$ is a principal ideal, $g[b_-] = \mathord\downarrow g(y_*)$ for some $y_* \in b$.

For each $x\in a_+$, $f(x)\nleq g(y_*)$. So there is a clopen ideal $I_x$ separating them. That is,
$g(y_*)\in I_x$ and $f(x)\notin I_x$. The complements of these $I_x$'s form an open cover of $f[a_+]$. So finitely many suffice, and the intersection of the corresponding clopen ideals contains $g(y_*)$, and is disjoint from $f[a_+]$. This intersection is itself a clopen ideal determining an HMS morphism $c\in \Two^Z$. Clearly,  $a\leq \Two^f(c)$ and $\Two^g(c) \leq b$ directly by the construction. 

\subsubsection{Banaschewski Duality}

Banaschewski \cite{banaschewski-duality} shows, in effect, that the topological duals of posets are bounded distributive lattices in $\Pri$ where the underlying order coincides with the lattice order -- we call such spaces \emph{Banaschewski spaces}.

Suppose (\ref{eq:square-xy}) is an exact square in Banaschewski spaces.
Then $\Two^p\circ \Two^f \leq \Two^q\circ \Two^g$ in $\Pos$.

Suppose $\Two^p(a)\leq \Two^q(b)$. By the same argument as in HMS spaces, $g[b_-]$ is a principal ideal and $f[a_+]$ is a principal filter. Let $y_*$ be the generator of $g[b_-]$ and $x_*$ be the generator of $f[a_+]$. Then $f(x_*)\nleq g(y_*)$ by exactness of the given square. So there is a closed prime ideal separating them.

Suppose (\ref{eq:square-xy}) is an exact square in $\Pos$. 
Then $\Two^p\circ \Two^f \leq \Two^q\circ \Two^g$ in Banaschewski spaces. If $\Two^p(a)\leq \Two^q(b)$, then $f[a_+]$ is an up-set,  $g[b_-]$ is a down-set, and the two are disjoint. So $f[a_+]$
determines an $c$ element of $\Two^C$ that satisfies $a\leq \Two^f(c)$ by construction. 
Clearly, $\Two^g(c)(y)=c(g(y)) \leq b(y)$ for every $y\in Y$. 

\section{Conclusion}

We showed how to extend an equivalence or adjunction from maps to relations. In more detail, Theorem~\ref{thm:dualityrelations} extends a dual equivalence of maps to a dual equivalence of relations, while Theorem~\ref{thm:adjunctionrelations} extends a dual adjunction (or equivalence) of maps to a dual adjunction (or equivalence) of the framed bicategory of relations.

The general framework is that of regular categories in a suitable order-enriched sense. Roughly speaking, the categories in question must have forgetful functors that preserve order-enriched limits and preserve regular factorisations; and the adjoint functors must preserve exact squares and regular factorisations. 

In our experience, to exhibit a particular example of an adjunction or equivalence satisfying these conditions, most of the work will go into verifying preservation of exact squares, see Lemma~\ref{lem:DLPri-pres-exact-squares} for our main example. It is worth noting that the proofs involving the dualising object $\Two$, always follow the same common outline inherited from $\Pos$, with the particularities of the situation entering only in one specific place, see Lemma~\ref{lem:DLPri-pres-exact-squares} and the proofs of Section~\ref{sec:exles-adjunctions} for specific examples.


In a sequel paper, we will apply the duality of relations in order to extend zero-dimensional dualities to continuous ones in a systematic way. As we have seen here, dualities such as the one between ordered Stone spaces and distributive lattices can be extended from maps to relations. Once we have relations, we can split idempotents and then restrict to maps again, obtaining a non-zero dimensional duality.
 
For future investigations, two important questions concern other dualising objects than 
$\Two$. First, while staying inside order-enriched categories, we plan to integrate our work here into the theory of natural dualities as described by Clark and Davey \cite{clark+davey} and to specialise Theorems \ref{thm:dualityrelations} and \ref{thm:adjunctionrelations} to this setting. Also possible relationships with J\'onnson-Tarski duality \cite{JT} and the theory of canonical extensions \cite{Venema:ac,DGP} should be explored.

Second, we want to know whether our approach can be extended to other enrichments than $\Two$ as for example Lawvere metric spaces \cite{lawvere}. In particular, it would be interesting to see whether this could find applications to stochastic relations as studied, for example, in Doberkat \cite{doberkat} and  Panangaden \cite{panangaden}. 

Another question is how much of the theory developed in this paper can be salvaged for functors that do not preserve exact squares. Looking back to Theorem~\ref{thm:adjunctionrelations}, even without the assumption of preservation of exact squares, we are in a situation similar to orthogonal adjunction in the double category of pseudo double categories with lax and colax double functors as in Section 5.3 of  \cite{GP04}.

From a category theoretic point of view, there is the question how much of the theory of regular categories transfers to order-regular categories. While P-varieties feature prominently in our work, recent work by Abramsky and coauthors \cite{ADW,abramsky-shah} on game comonads suggests potential examples of enriched \emph{co}varieties. In particular, one could have a look at games for continuous model theory, which has been given a category theoretic foundation recently by Cho \cite{cho}.

There is also a long list of more specific questions. For example, as discussed after Example~\ref{exle:total-relation-DL}, it should be interesting to look at dual relations of not-necessarily-bounded distributive lattices. Or a wide range of other dualities, for that matter. Finally, there are a number of technical questions, for example whether cocommas are exact in all order-regular categories or how an explicit characterisations of cocommas in various algebraic categories including distributive lattices would look like.

Returning to more fundamental questions, this paper focussed on heterogeneous relations $A\looparrowright B$ in the context of order-enriched algebra. An investigation  into homogeneous relations $A\looparrowright A$ in the presence of order-preserving as well as order-reversing operations is one important topic for future investigation, in particular in connection with some recent work in proof theory of Greco et.al~\cite{GRMT}. Another question is whether our approach can be extended to relations $A_1\times\ldots A_n \looparrowright B_1\times\ldots B_m$.


\end{document}